\documentclass[12pt]{article}
\usepackage{a4wide}
\usepackage{amsmath, amsthm, amsfonts, amssymb, bbm}
\usepackage{graphics}
\usepackage{xypic}
\usepackage[all]{xy}
\usepackage{tikz}
\usepackage[english]{babel}
\usepackage[font=small,format=plain,labelfont=bf,up]{caption}

\usepackage{ifpdf} 
\ifpdf
\usepackage{epstopdf}
\usepackage{hyperref}
\else
\usepackage[hypertex]{hyperref}
\fi

\theoremstyle{plain}

\newtheorem{theorem}{Theorem}
\newtheorem{lemma}[theorem]{Lemma}

\newtheorem{corollary}[theorem]{Corollary}

\theoremstyle{definition}

\newtheorem{remark}[theorem]{Remark}
\newtheorem{definition}[theorem]{Definition}

 \numberwithin{equation}{section}
 \numberwithin{theorem}{section}

\DeclareMathOperator{\cok}{cok}

\DeclareMathOperator{\Hom}{Hom}

\DeclareMathOperator{\Rep}{Rep}
\DeclareMathOperator{\ev}{ev}
\DeclareMathOperator{\coev}{coev}

\DeclareMathOperator{\Cent}{Cent}

\newcommand\be            {\begin{equation}}
\newcommand\ee            {\end{equation}}

\newcommand\nxt{\noindent\raisebox{.08em}{\rule{.44em}{.44em}}\hspace{.4em}}

\newcommand\vect{\mathcal{V}\hspace{-.5pt}ect}
\newcommand\Fun{\mathcal{F}\hspace{-.5pt}un}
\newcommand\Vir{\mathrm{Vir}}

\newcommand{\CxC}{\Cc \times \Cc^{\mathrm{rev}}}
\newcommand{\CoC}{\Cc \boxtimes \Cc^{\mathrm{rev}}}

\newcommand\eps           {\varepsilon}
\newcommand\id            {id}
\newcommand\Id            {I\hspace{-1pt}d}
\newcommand\one           {{\bf1}}

\newcommand\Cb            {\mathbb{C}}
\newcommand\Hb            {\mathbb{H}}
\newcommand\Nb            {\mathbb{N}}
\newcommand\Rb            {\mathbb{R}}
\newcommand\Zb            {\mathbb{Z}}

\newcommand\Ac            {\mathcal{A}}
\newcommand\Bc            {\mathcal{B}}
\newcommand\Cc            {\mathcal{C}}
\renewcommand\Dc            {\mathcal{D}}
\newcommand\Ec            {\mathcal{E}}
\newcommand\Hc            {\mathcal{H}}
\newcommand\Ic            {\mathcal{I}}

\newcommand\Mc            {\mathcal{M}}
\newcommand\Oc            {\mathcal{O}}
\newcommand\Pc            {\mathcal{P}}
\newcommand\Qc            {\mathcal{Q}}

\newcommand\Tc            {\mathcal{T}}
\newcommand\Vc            {\mathcal{V}}
\newcommand\Wc            {\mathcal{W}}
\newcommand\Zc            {\mathcal{Z}}

\newcommand\void[1]{}

\begin{document}

\thispagestyle{empty}
\def\thefootnote{\fnsymbol{footnote}}
\begin{flushright}
ZMP-HH/12-1\\
Hamburger Beitr\"age zur Mathematik 430
\end{flushright}
\vskip 3em
\begin{center}\LARGE
Logarithmic bulk and boundary conformal field theory\\
and the full centre construction
\end{center}

\vskip 2em
\begin{center}\large
  Ingo Runkel%
  $^{a}$\footnote{Email: {\tt ingo.runkel@uni-hamburg.de}},
  Matthias R. Gaberdiel%
  $^{b}$\footnote{Email: {\tt gaberdiel@itp.phys.ethz.ch}}\,%
  and
  Simon Wood%
  $^{c}$\footnote{Email: {\tt simon.wood@ipmu.jp}}
\end{center}
\begin{center}\it$^a$
Fachbereich Mathematik, Universit\"at Hamburg \\
Bundesstra\ss e 55, 20146 Hamburg, Germany
\end{center}
\begin{center}\it$^b$
Institute for Theoretical Physics, ETH Z\"urich \\
8093 Z\"urich, Switzerland
\end{center}
\begin{center}\it$^c$
IPMU, The University of Tokyo\\
5-1-5 Kashiwanoha, Kashiwa, 277-8583, Japan
\end{center}

\vskip 2em
\begin{center}
  January 2012
\end{center}

\vskip 2em

\begin{abstract}
We review the definition of bulk and boundary conformal field theory in a way suited for logarithmic conformal field theory. The notion of a maximal bulk theory which can be non-degenerately joined to a boundary theory is defined. The purpose of this construction is to obtain the more complicated bulk theories from simpler boundary theories. We then describe the algebraic counterpart of the maximal bulk theory, namely the so-called full centre of an algebra in an abelian braided monoidal category. Finally, we illustrate the previous discussion in the example of the $W_{2,3}$-model with central charge 0.
\end{abstract}

\setcounter{footnote}{0}
\def\thefootnote{\arabic{footnote}}

\newpage

\tableofcontents

\newpage

\section{Introduction}

In two-dimensional conformal field theory, one usually considers correlation functions where the fields have power law singularities as they approach each other, for example $\langle \sigma(z) \sigma(w) \rangle = |z-w|^{1/4}$ for the correlator of two spin fields in the critical Ising model. Power law behaviour occurs if the two fields approaching each other are eigenvectors of the generator of infinitesimal scale transformations $\Delta = L_0+\overline L_0$. In unitary theories one has $\Delta^\dagger = \Delta$, so that $\Delta$ can be diagonalised. In non-unitary theories, however, there is no a priori reason to impose diagonalisability of $\Delta$, and in this case additional logarithmic singularities can occur. For example, in the symplectic fermion model of \cite{Gaberdiel:1998ps} the two-point correlator of the partner of the vacuum state reads $\langle \omega(z) \omega(w) \rangle = 4 \log|z-w|$. 

From the point of view of representation theory, in unitary theories the state spaces are direct sums of irreducible representations of (two copies of) the Virasoro algebra, while in non-unitary theories the indecomposable summands may or may not be irreducible. In fact, if $\Delta$ is not diagonalisable one necessarily finds such non-semi-simple behaviour.

\medskip

A general recipe for constructing examples of non-logarithmic rational conformal field theories, known as the `Cardy case', is as follows. Take all irreducible representations $R_i$ of the chiral symmetry, that is, of the algebra formed by all modes of all holomorphic fields, or, more formally, of a vertex operator algebra $\Vc$. The state space relevant to describe the theory on a cylinder, i.e.\ the space of states on the circle, is $\Hc_\mathrm{bulk} = \bigoplus_i R_i \otimes_{\Cb} R_i^*$, where the sum runs over all irreducibles and $R_i^*$ is the conjugate representation. This theory can be placed on a strip, in which case we have to fix the state space $\Hc_\mathrm{bnd}$
on an interval with prescribed boundary conditions. If the two boundary conditions coincide one may take $\Hc_\mathrm{bnd} = R \otimes_f R^*$, where $R$ is an arbitrary representation of the chiral symmetry $\Vc$ (not necessarily irreducible) and `$\otimes_f$' denotes the fusion tensor product. In particular, if we take $R = \Vc$ then $\Hc_\mathrm{bnd} = \Vc \otimes_f \Vc^* \cong \Vc$. That is, the space of boundary states consists of a single irreducible representation, namely the vacuum representation itself. This leads us to the first theme to keep in mind:
\begin{quote} 
  \textsl{For an appropriate choice of boundary condition, the boundary theory is much simpler than the bulk theory.}
\end{quote}
It turns out that in all rational conformal field theories which can be defined on surfaces with or without boundary and which have a unique bulk vacuum, the boundary theory determines the bulk theory uniquely \cite{Fuchs:2001am,Fjelstad:2005,Kong2006b}. The bulk theory is characterised as the `largest possible one' which can be matched to the given boundary theory. This principle has also been checked for some logarithmic models \cite{Gaberdiel:2006pp,Gaberdiel:2007jv,Gaberdiel:2010rg}. The second theme to keep in mind can be phrased as:
\begin{quote} 
  \textsl{For a given boundary theory, one may find a largest possible bulk theory that can be consistently and non-degenerately joined to the boundary theory. This bulk theory, if it exists, is unique.}
\end{quote}
This principle has also been established in the operator algebraic approach to unitary conformal field theory on the half-plane with Minkowski signature \cite{Longo:2004fu}. (Logarithmic models are not accessible in this setting as the formalism requires unitarity.)

On the representation theoretic side, the construction of the bulk theory as the largest one which fits to a specific boundary theory corresponds to starting from an algebra in an abelian braided monoidal category $\Cc$ and finding its `full centre', a commutative algebra in the product category $\CoC$.

\medskip

This paper consists of three parts. In part one, which is section \ref{sec:bb-corr}, the definition of bulk and boundary conformal field theory in terms of its correlation functions is reviewed. Using this definition, the characterisation of the bulk theory as the `largest one' fitting to a given boundary theory is made precise. 

Part two (section \ref{sec:alg-reform}) provides the algebraic counterparts of the conformal field theory notions in section \ref{sec:bb-corr} in the setting of abelian braided monoidal categories. This part contains a fairly detailed review of the Deligne product of abelian categories, as this will play an important role. The main notion in part two is that of the full centre of an algebra. We will recall its definition, derive some of its properties, and link its definition to the maximality condition of the bulk theory associated to a boundary theory from section \ref{sec:bb-corr}.

Part three (section \ref{sec:W23-model}) investigates one specific example of a logarithmic conformal field theory, namely the $W_{2,3}$-model of central charge zero. We chose this model because on the one hand it is still relatively simple, for example it only involves 13 distinct irreducible representations, but on the other hand each `nice' property from non-logarithmic rational theories which is currently known to be violated in logarithmic models with a finite number of irreducibles is already violated in the $W_{2,3}$-model. We discuss properties of a tentative bulk theory for the $W_{2,3}$-model which can be interpreted as a `logarithmic extension' of the underlying unitary minimal model at $c=0$, i.e.\ the trivial theory with a one-dimensional state space. We consider the Virasoro action on states of generalised weight $(0,0)$ and $(2,0)$ and discuss the operator product expansion of some of these fields. We also find that an analogue of the indecomposability parameter $b$ is equal to $-5$. This value has recently appeared in the discussion of bulk theories with $c=0$ \cite{Vasseur:2011ud}.

\medskip

This paper grew out of two talks given by the first author which were based on the joint works \cite{Gaberdiel:2006pp,Gaberdiel:2007jv,Gaberdiel:2009ug,Gaberdiel:2010rg}. We have tried to make this paper to some extent self-contained. In consequence it became slightly lengthy and contains a large amount of review material. Nonetheless, there are also some new results which we briefly list: the discussion of ideals for homomorphisms of conformal field theories in section \ref{sec:genout-ideal}; the reformulation of the computation of the maximal bulk theory in purely categorical language in section \ref{sec:fullcentre} and table \ref{tab:alg-cft-table2}; the treatment in section \ref{sec:alg-reform} of a class of abelian monoidal categories more general than finite tensor categories (as defined in \cite{Etingof:2003}); the existence proof of the full centre in this setting in theorem \ref{thm:full-centre-via-left-centre}; the calculation of the analogue of the indecomposability parameter $b=-5$ in the $W_{2,3}$-bulk theory $R(\one^*)$ and the operator product expansions in this model in section \ref{sec:R1-OPE}.

\newpage

\section{Bulk and boundary correlators} \label{sec:bb-corr}

In this section we give a definition of conformal field theory on the complex plane and on the upper half plane in terms of correlation functions. The presentation is tailored to be self-contained and to make the relation to the algebraic concepts in section \ref{sec:alg-reform} apparent.

\medskip\noindent
{\bf Bibliographical note:} This section is mostly a review. The characterisation of CFT on the complex plane in terms of correlators and operator product expansion is used in \cite{Belavin:1984vu}. Axiomatic formulations close in spirit to the one presented below are \cite{Gaberdiel:1998fs,Huang:2005}  (other approaches can be found in \cite{Friedan:1986ua,Vafa:1987ea,Segal1988,Kapustin:2000aa,HuKriz2004}). The point that the requirement of modular invariance poses severe constraints on a CFT was stressed in \cite{Cardy:1986ie,Cappelli:1987xt}. CFT on the upper half plane as presented below was developed in \cite{Cardy:1984bb,Cardy:1991tv,Lewellen:1991tb}; an axiomatic formulation can be found in \cite{Kong2006b}. The idea to obtain the CFT on the complex plane from correlators of boundary fields was first implemented in \cite{Runkel:1998pm,Runkel:1999dz} and further developed in the context of non-logarithmic rational CFT in \cite{Fuchs:2001am,Fjelstad:2005,Kong:2008ci}. The first application of this principle to logarithmic models can be found in \cite{Gaberdiel:2006pp,Gaberdiel:2007jv}.

\subsection{Consistency conditions for CFT on the complex plane} \label{sec:CFT-C}

We will take the point of view that a two-dimensional conformal field theory (or any statistical or quantum field theory in any dimension, for that matter) is defined in terms of its correlation functions. That is, we are given a \textsl{space of fields} $F$, which is a $\Cb$-vector space whose elements we call \textsl{fields}. The space $F$ is typically infinite dimensional, because with each field it contains all its derivatives (see remark \ref{rem:CFT-plane}\,(iii) below for a precise statement). In addition we have a \textsl{collection of correlators} $(C_n)_{n \in \Zb_{>0}}$. We call  $C_n$ an $n$-\textsl{point correlator}. It assigns a complex number to $n$ fields and $n$ mutually distinct complex numbers, i.e.\
\be
  C_n : (\Cb^n{\setminus}\text{diag}) \times F^n \longrightarrow \Cb \ ,
\ee
where $\Cb^n{\setminus}\text{diag}$ stands for points $(z_1,\dots,z_n) \in \Cb^n$ such that $z_i \neq z_j$ for $i \neq j$. The collection $(C_n)_{n \in \Zb_{>0}}$ must obey the following conditions:
\begin{itemize}
\item[(C1)] Each $C_n$ is smooth in each argument from $\Cb$ and linear in each argument from  $F$.
\item[(C2)] Each $C_n$ is invariant under joint permutation of the arguments in $\Cb^n$ and $F^n$, i.e.\ for each permutation 
$\sigma \in S_n$, 
$$
C_n(z_1,\dots,z_n,\phi_1,\dots,\phi_n) = C_n(z_{\sigma(1)},\dots,z_{\sigma(n)},\phi_{\sigma(1)},\dots,\phi_{\sigma(n)}) \ .
$$
\end{itemize}
In addition, the $C_n$ must allow for an operator product expansion and satisfy invariance conditions, cf.\ (C3)--(C5) below. The customary notation for a correlator is
\be
  C_n(z_1,\dots,z_n,\phi_1,\dots,\phi_n)  \equiv \big\langle\,\phi_1(z_1) \cdots \phi_n(z_n)\,\big\rangle \ ,
\ee
where for us the right hand side is just a notational device. In particular, we do not assign an independent meaning to $\phi(z)$ as an operator. Still, we will say `the field $\phi$ is inserted at position $z$' if the pair $\phi,z$ is an argument of a correlator.

We now turn to the notion of a `short distance expansion' or `operator product expansion'.\footnote{In our setting only the first term makes sense literally, but it is customary to use the second term and abbreviate it as OPE, so we will do the same} The OPE links the $n{+}1$-point and $n$-point correlators. Namely, if two fields, say $\phi_1$ and $\phi_2$, are `close together' in the sense that $z_1$ is closer to $z_2$ than any other insertion point, then
\be \label{eq:ope-v1}
  \big\langle\,\phi_1(z_1) \phi_2(z_2) \phi_3(z_3) \cdots\big\rangle 
  =
  \sum_\alpha f_{\phi_1,\phi_2}^\alpha(z_1-z_2) \cdot \big\langle\,\varphi_\alpha(z_2)  \phi_3(z_3) \cdots\,\big\rangle  \ .
\ee
Here $\varphi_\alpha$ is some basis of $F$ and $f_{\phi_1,\phi_2}^\alpha(x)$ are functions which do not depend on how many or which fields are part of the correlator, apart from $\phi_1$ and $\phi_2$. 

\medskip

More formally, we demand that $F$ is a direct sum $F = \bigoplus_{\Delta \in \Rb} F^{(\Delta)}$ where $F^{(\Delta)}$ are the fields of `generalised scaling dimension' $\Delta$, and we demand that $F$ is bounded below in the sense that $F^{(\Delta)}=0$ for $\Delta \ll 0$. We define $\overline F$ to be the algebraic completion, i.e.\ the direct product $\overline F = \prod_{\Delta \in \Rb} F^{(\Delta)}$. 
The OPE is a map
\be \label{eq:Mz-def}
  M : \Cb^\times \times F \otimes_{\Cb} F \longrightarrow \overline F  
  \quad , \quad (z,v) \mapsto M_z(v) \ ,
\ee
which is linear in $F \otimes_{\Cb} F$. In the notation of \eqref{eq:ope-v1}, this amounts to writing $M_x(\phi_1 \otimes \phi_2) =  \sum_\alpha f_{\phi_1,\phi_2}^\alpha(x) \cdot \varphi_\alpha$, where the sum is typically infinite, hence the need for a completion. 
Note that $\overline F$ comes with canonical projections to `states with scaling dimension $\Delta$ or less', $P_\Delta : \overline F \to \bigoplus_{d \le \Delta} F^{(d)}$. With the help of these, we formulate the OPE condition:
\begin{itemize}
\item[(C3)] 
For $n \ge 1$, $\phi_1,\dots,\phi_{n+1} \in F$, and $(z_1,\dots,z_{n+1}) \in \Cb^{n+1}{\setminus}\text{diag}$ such that $|z_1-z_2| < |z_k-z_2|$ for $k>2$, we have
\be\label{eq:OPE-condition}
\begin{array}{l}
\displaystyle
  C_{n+1}\big(z_1,z_2,z_3,\dots,z_{n+1},\phi_1,\phi_2,\phi_3,\dots,\phi_{n+1}\big)
  \\[.5em]
\displaystyle
  = \lim_{\Delta \to \infty}
  C_{n}\big(z_2,z_3,\dots,z_{n+1},P_\Delta \circ M_{z_1-z_2}(\phi_1 \otimes \phi_2),\phi_3,\dots,\phi_{n+1}\big) \ .
\end{array}
\ee
\end{itemize}
The limiting procedure in \eqref{eq:OPE-condition} is necessary because $C_n$ is defined only on $F$, not on $\overline F$. The existence of the limit is a non-trivial requirement. In fact, if $|z_1-z_2| \ge |z_k-z_2|$ for some $k>2$, the expression on the right will typically diverge for $\Delta \to \infty$. That (C3) is only formulated for the first two arguments of $C_{n+1}$ is not a restriction due to the permutation invariance imposed in (C2). 

\begin{remark}
In addition to (C3), one often requires the existence of a translation invariant vacuum vector, that is, a vector $\Omega \in F^{(0)}$ such that $\langle \Omega(\zeta) \phi_1(z_1) \cdots \phi_n(z_n) \rangle = \langle \phi_1(z_1) \cdots \phi_n(z_n) \rangle$ for $n \ge 1$. We prefer not to include this as an axiom because our example in section \ref{sec:W23-model} below (conjecturally) satisfies (C1)--(C3), as well as (C4) and (C5$'$) to be discussed below,
while not having a vacuum vector.
\end{remark}

Finally, let us describe the coinvariance conditions. Denote by $\Vir$ the Virasoro algebra. We demand the following properties of $F$:
\begin{itemize}
\item $F$ is equipped with the structure of a $\Vir \oplus \Vir$-module. The generators of the first copy of $\Vir$ are denoted by $L_n$ and $C$, and those of the second copy by $\overline L_n$ and $\overline C$. 
\item $F$ has a direct sum decomposition into spaces $F^{(\Delta)}$ of generalised ($L_0+\overline L_0)$-eigenvalue $\Delta$; this decomposition satisfies $F^{(\Delta)}=0$ for $\Delta \ll 0$. (This was already imposed above.)
\item $F$ is locally finite as a $\Cb L_0 \oplus \Cb \overline L_0$ module. This means that acting with $L_0$ and $\overline L_0$ on any vector $v \in F$ generates a finite-dimensional subspace.
\end{itemize}
The last condition guarantees in particular that the exponentials $\exp(\lambda L_0)$ and $\exp(\lambda \overline L_0)$, for $\lambda \in \Cb$, are well defined operators on $F$. The condition holds automatically if all $F^{(\Delta)}$ are finite-dimensional.

There are two types of coinvariance conditions. The first one is easy to formulate and allows one to move insertion points:
\begin{itemize}
\item[(C4)] 
For $n \ge 1$, $\phi_1,\dots,\phi_n \in F$, and $z_1,\dots,z_n \in \Cb^n{\setminus}\text{diag}$,
\be \begin{array}{l}
\frac{d}{dz_1} C_n\big(z_1,\dots,z_n,\phi_1,\dots,\phi_n\big)
\,=\, C_n\big(z_1,\dots,z_n,\,L_{-1} \phi_1\,,\dots,\phi_n\big) \ ,
\\[.5em]
\frac{d}{d\bar z_1} C_n\big(z_1,\dots,z_n,\phi_1,\dots,\phi_n\big)
\,=\, C_n\big(z_1,\dots,z_n,\,\overline L_{-1} \phi_1\,,\dots,\phi_n\big) \ .
\end{array}
\ee
\end{itemize}
By permutation invariance, the fact that (C4) is formulated only for the first argument only is not a restriction.

The second type of coinvariance condition is a bit more involved.
Let $f$ be a meromorphic function on $\Cb \cup \{ \infty \}$ (i.e.\ a rational function)
which has poles at most at the points $z_1,\dots,z_n$ and $\infty$, and which satisfies the growth condition $\lim_{\zeta \to \infty} \zeta^{-3} f(\zeta) = 0$. Denote the expansion parameters around each of the $z_k$ as $f(\zeta) = \sum_{m=-\infty}^\infty f^k_m \cdot (\zeta - z_k)^{m+1}$.
\begin{itemize} 
\item[(C5)] 
For $n \ge 1$, $\phi_1,\dots,\phi_n \in F$, and $z_1,\dots,z_n \in \Cb^n{\setminus}\text{diag}$, and for all $f$ as above,
\be \label{eq:mode-moving}
  \sum_{k=1}^n
  \sum_{m=-\infty}^\infty f^k_m \cdot
  C_n(z_1,\dots,z_n,\phi_1,\dots,\,L_m \phi_k\,,\dots,\phi_n) = 0 \ ,
\ee
and the same condition with $\overline L_m$ in place of $L_m$.
\end{itemize}
The sum over $m$ in \eqref{eq:mode-moving} is actually finite: Since $f$ is meromorphic, $f^k_m = 0$ for $m \ll 0$, and since the grading by generalised scaling dimensions on $F$ is bounded from below, $L_m \phi_k = 0$ for $m \gg 0$.

\begin{remark} \label{rem:bulk-stress}
In place of (C5) one could put the stronger requirement of the existence of a stress tensor. This would be a pair of fields $T,\overline T \in F^{(2)}$ (called the holomorphic and anti-holomorphic component of the stress tensor) such that $\overline L_{-1} T=0$ and $L_{-1} \overline T=0$, and
\be \label{eq:T-OPE}
  M_z(T \otimes \phi) = \sum_{m=-\infty}^\infty z^{-m-2} \, L_m \phi \quad , \quad
  M_z(\overline T \otimes \phi) = \sum_{m=-\infty}^\infty \bar z^{-m-2} \, \overline L_m \phi \ .
\ee
Note that $M_z(T \otimes \phi)$ and $M_z(\overline T \otimes \phi)$ 
are elements of $\overline F$, as they should be. Furthermore, one requires that the limit $\lim_{\zeta \to\infty} |\zeta|^4 \langle T(\zeta) \phi_1(z_1) \cdots \phi_n(z_n) \rangle$ exists for all $n$ and all $z_i$, $\phi_i$, and similar for $\overline T$ (this is $sl(2,\Cb)$-invariance of the out vacuum). The conditions \eqref{eq:mode-moving} arise from the contour integral
\be
  \frac{1}{2 \pi i} \oint  f(\zeta) \cdot \big\langle T(\zeta) \phi_1(z_1) \cdots \phi_n(z_n) \big\rangle\, d\zeta = 0
\ee
where the contour is a big circle enclosing $z_1,\dots,z_n$. Deforming the contour to a union of small circles, one around each $z_i$, and applying the OPE \eqref{eq:T-OPE} results in \eqref{eq:mode-moving}.
\end{remark}

\begin{remark}\label{rem:sphere1}
(i) One consequence of (C5) is that all correlators are translation invariant,
\be
\big\langle \phi_1(z_1{+}s) \cdots \phi_n(z_n{+}s) \big\rangle ~=~
\big\langle \phi_1(z_1) \cdots \phi_n(z_n) \big\rangle 
\qquad \text{for all} ~~ s \in \Cb \ .
\ee
To see this apply (C5) to the constant function $f=1$, in which case $f^k_m = \delta_{m,-1}$ for $k=1,\dots,n$ and so $\sum_{k=1}^n C_n(z_1,\dots,z_n,\phi_1,\dots,L_{-1} \phi_k,\dots,\phi_n) = 0$. Combining the corresponding relation for $\overline L_{-1}$ with (C4) yields translation invariance. Along the same lines one shows covariance (not invariance) under M\"obius transformations which map none of the points $z_1,\dots,z_n$ to infinity.
\\[.3em]
(ii) A CFT on $\Cb$ can be used to define $n$-point correlators on the Riemann sphere. This is done by choosing an isomorphism from the Riemann sphere to $\Cb \cup \{ \infty \}$ such that no field insertion gets mapped to infinity and by then evaluating $C_n$ on the resulting configuration. (One also needs to include local coordinates around the insertions, we skip the details).
\end{remark}

Since the OPE allows one to reduce\footnote{
  Of course, the OPE can only be applied if the condition on the distances of insertion points in (C3) is met. But one can always choose a pair $z_i,z_j$ of distinct points such that $|z_i-z_j|$ is minimal among all distances between pairs of insertion points. If necessary, one can then pick a point $z_i'$ arbitrarily close to $z_i$ such that $|z_i'-z_j|$ is strictly smaller than all other distances. The OPE (C3) applies to the pair $z_i',z_j$ and the value of the correlator at $z_i$ is determined by continuity.}
$C_{n+1}$ to $C_{n}$, all correlators are uniquely determined by the OPE and $C_1$. By translation invariance, $C_1(z,\phi)$ is independent of $z$ and thus yields a function $\Omega^* : F \to \Cb$. It follows from (C5) with $f(\zeta) = (\zeta-z)^{m+1}$ that 
\be
  C_1(z,L_m \phi) = 0 = C_1(z,\overline L_m \phi) \quad \text{for all} \quad  m \le 1 \ . 
\ee
If $\phi \in F^{(\Delta)}$, then by definition $(L_0+\overline L_0 - \Delta)^N \phi = 0$ for some $N>0$. This gives $0 = C_1(z,(L_0+\overline L_0 - \Delta)^N \phi) = (-\Delta)^N C_1(z,\phi)$, 
because, as we just saw, $C_1(z, (L_0+\overline L_0)^k \phi) = 0$ for all $k>0$. It follows that $C_1(z,\phi)$ can be non-zero only if $\Delta=0$.

Let us collect the discussion so far into a definition.

\begin{definition} \label{def:CFT-plane}
A \textsl{conformal field theory on the complex plane} is a triple $(F,M,\Omega^*)$, where
\begin{itemize}
\item $F$ (the \textsl{space of fields}) is a $\Vir \oplus \Vir$-module which is a direct sum of generalised $(L_0{+}\overline L_0)$-eigenspaces $F^{(\Delta)}$ whose generalised eigenvalues are bounded from below, and which is locally finite as a 
$\Cb L_0 \oplus \Cb \overline L_0$ module,
\item $M$ (the \textsl{operator product expansion}) is a function $\Cb^\times \times F \otimes_{\Cb} F \to \overline F$ which is linear in $F \otimes_{\Cb} F$,
\item $\Omega^*$ (the \textsl{out-vacuum}) is a linear map $F^{(0)} \to \Cb$,
\end{itemize}
such that there exists a collection of correlators $(C_n)_{n \in \Zb_{>0}}$ which satisfy (C1)--(C5) and the normalisation condition $C_1(z,\phi) = \langle \Omega^* , \phi \rangle$.
\end{definition}

\begin{remark} \label{rem:CFT-plane}
(i) The definition shows that a CFT contains only a relatively small amount of data which has to satisfy an infinite number of intricate linear and differential equations. It is in fact very hard to prove that a triple $(F,M,\Omega^*)$ gives a CFT. To some extent, the formalism of vertex operator algebras, its representations and intertwining operators was developed with this aim. The VOA formalism allows one to prove that non-logarithmic rational CFTs provide examples of definition \ref{def:CFT-plane}, see \cite{Huang:2005}. We are not aware of a full proof of the existence of a logarithmic CFT in the above sense, e.g.\ using the formalism \cite{Huang:2010}. (This is merely to indicate that logarithmic CFTs are more difficult, not that we doubt their existence).
\\[.3em]
(ii) We have deliberately not included non-degeneracy of the 2-point correlator $\langle \phi(z) \psi(w) \rangle$ into definition \ref{def:CFT-plane}; this will be discussed in the next subsection.
\\[.3em]
(iii) If the space $F$ is finite-dimensional, then the $\Vir \oplus \Vir$-action on $F$ has to be trivial,\footnote{
  All finite dimensional $\Vir$-modules $M$ are trivial. The proof is easy. The Jordan normal form of $L_0$ splits $M$ into generalised $L_0$-eigenspaces. Let $\Lambda\ge0$ be such that all generalised $L_0$-eigenvalues have real parts of absolute value less or equal to $\Lambda$. Since $L_m$ changes the generalised $L_0$-eigenvalue by $-m$, all $L_m$ with $|m|>2 \Lambda$ must act as zero. For $m\neq 0$ and $2N+m \neq 0$ we can write $L_m = [L_{N+m},L_{-N}] / (2N+m)$. For $N$ large enough, both $L_N$ and $L_{m-N}$ act trivially on $M$, and so all $L_m$ with $m \neq 0$ must act trivially. Therefore, also $L_0 = \frac12 [L_1,L_{-1}]$ and $C = 2[L_{2},L_{-2}] - 4 [L_1,L_{-1}]$ act trivially.}
and so in particular $L_{-1}$ and $\overline L_{-1}$ would act trivially on $F$. By (C5) this implies that all correlators are independent of the insertion points. Such a conformal field theory is called a \textsl{topological field theory}.
\end{remark}

Suppose $(F,M,\Omega^*)$ is a conformal field theory. By assumption there exists a collection of correlators $(C_n)_{n \in \Zb_{>0}}$ satisfying (C1)--(C5) and we have seen above that this determines the $C_n$ uniquely. As a small example computation with the above axioms, let us look at $\langle \phi(z) \psi(w) \rangle$. By translation invariance, we may assume $w=0$. By (C3), 
\be
  \langle \phi(z) \psi(0) \rangle = \lim_{\Delta \to \infty} C_1(0,P_\Delta \circ M_{z}(\phi \otimes \psi)) = \langle \Omega^*, M_{z}(\phi \otimes \psi) \rangle \ .
\ee
The limit can be dropped because $\Omega^*$ is non-vanishing only on $F^{(0)}$. Next, by (C5) with $f(\zeta) = \zeta$ we know that $C_2(z,0,(L_0 + z L_{-1})\phi,\psi) + C_2(z,0,\phi,L_0 \psi) = 0$, together with (C4) we find 
\be
  -z\tfrac{d}{dz} \langle \Omega^*, M_{z}(\phi \otimes \psi) \rangle = \langle \Omega^*, M_{z}(L_0 \phi \otimes \psi) \rangle + \langle \Omega^*, M_{z}(\phi \otimes L_0 \psi) \rangle
\ee   
and a corresponding equation with $\tfrac{d}{d\bar z}$ and $\overline L_0$. The solution to these first order differential equations reads
\be \label{eq:explicit-2-pt}
\begin{array}{ll}
\big\langle \phi(z) \psi(0) \big\rangle
&\!\!= \big\langle \Omega^*, M_{1} \circ \exp\!\big\{- \ln(z) (L_0 \otimes \id_F + \id_F \otimes L_0) 
\\[.5em]
& \hspace{9em} 
- \ln(\bar z) (\overline L_0 \otimes \id_F + \id_F \otimes \overline L_0) \big\} \, \phi \otimes \psi \big\rangle \ .
\end{array}
\ee
From this we see two things: when evaluated on $\Omega^*$, $M_z$ is uniquely fixed by $M_1$; and if the action of $L_0$ or $\overline L_0$ has a nilpotent part, the two-point correlators may contain logarithms. In sub-representations of $F$ which are irreducible, $L_0$ acts diagonalisably since $\exp(2 \pi i L_0)$ commutes with all Virasoro modes, and hence by Schur's Lemma has to be a multiple of the identity. In this sense, the appearance of logarithms is linked to (but not equivalent to) the presence of non-semi-simple $\Vir$-modules.

\subsection{Background states, non-degeneracy, and ideals}\label{sec:genout-ideal}

We would like to allow more general out-states -- or background states -- than the out-vacuum $\Omega^*$, namely, we would like to be able to place an arbitrary state from the graded dual of $F$ ``at infinity''. The graded dual of $F$ is by definition the space of linear maps $\overline F \to \Cb$. A more explicit description is\footnote{
  Clearly, every map satisfying the boundedness condition $u(F^{(\Delta)}) = 0$ for $\Delta > \Delta_\mathrm{max}(u)$ is a linear functional on $\overline F$. Suppose conversely that there were a sequence $\{ v_n \}_{n \in \Nb}$ with $v_n \in F^{(\Delta_n)}$ and $\Delta_n \to \infty$ for $n \to \infty$, such that $u(v_n) \neq 0$ for all $n$. Then $u$ would be ill-defined when acting on the element $\sum_n (u(v_n))^{-1} \cdot v_n \in \overline F$.}
\be
  \overline F^* = \big\{ \, u : F \to \Cb \text{ linear } \big| \, \exists \Delta_\mathrm{max}(u) :  u(F^{(\Delta)}) = 0 \text{ for } \Delta > \Delta_\mathrm{max}(u) \big\} \ .
\ee
The graded dual is again a $\Vir \oplus \Vir$-module via $(L_m u)(v) := u(L_{-m}v)$ and $(\overline L_m u)(v) := u(\overline L_{-m}v)$. With this definition, the generalised $(L_0+\overline L_0)$-eigenvalues of $\overline F^*$ are the same as those of $F$ and each element $u \in \overline F^*$ is annihilated by $L_m$ and $\overline L_m$ for large enough $m>0$.

\medskip

Define a \textsl{CFT on $\Cb$ with background states} as a pair $(F,M)$, where $F$ and $M$ are as in definition \ref{def:CFT-plane}. However, for each $u \in \overline F^*$ we now demand the existence of functions $C_n(u|z_1,\dots,z_n,\phi_1,\dots,\phi_n)$, which we will also write as
\be
  {}^u\langle \phi_1(z_1) \cdots \phi_n(z_n) \rangle  \ ,
\ee
and which have to satisfy the following conditions:
\begin{itemize}
\item ${}^u\langle \phi(0) \rangle = u(\phi)$ for all $u \in \overline F^*$, $\phi \in F$.
\item (C1)--(C4) from before, but with $C_n(u|\cdots)$ in place of $C_n(\cdots)$.
\item (C5$'$), which is a modified version of (C5) to be described now.
\end{itemize}
Let $f$ be a rational function on $\Cb \cup \{ \infty \}$ as for (C5), but without imposing the growth condition at infinity. Define the $f^k_m$ as for (C5) and define $f^\infty_m$ via the expansion around infinity: $f(\zeta) = \sum_{m=-\infty}^\infty f^\infty_{m} \zeta^{m+1}$ for $|\zeta|$ larger than all of the $|z_i|$.
\begin{itemize} 
\item[(C5$'$)] 
For $n \ge 1$, $u \in \overline F^*$, $\phi_1,\dots,\phi_n \in F$, and $z_1,\dots,z_n \in \Cb^n{\setminus}\text{diag}$, and for all $f$ as above,
\be \label{eq:mode-moving-generalout}
\begin{array}{l}
\displaystyle
  \sum_{k=1}^n
  \sum_{m=-\infty}^\infty f^k_m \cdot
  C_n(u|z_1,\dots,z_n,\phi_1,\dots,\,L_m \phi_k\,,\dots,\phi_n) 
\\[-.2em]
\displaystyle
  \hspace{10em} = 
  \sum_{m=-\infty}^\infty f^\infty_{-m} \cdot
  C_n(\,L_mu\,|z_1,\dots,z_n,\phi_1,\dots,\phi_n) 
\end{array}
\ee
and the same condition with $\overline L_m$ in place of $L_m$.
\end{itemize}
As for (C5), the sums in \eqref{eq:mode-moving-generalout} only involve a finite number of non-zero terms.

\medskip

Let $(F,M)$ be a CFT on $\Cb$ with background states. Let $\Omega^* \in \overline F^*$ be a primary $sl(2,\Cb)$-invariant state, that is, $L_m \Omega^* = 0 = \overline L_m \Omega^*$ for all $m \ge -1$. Then $(F,M,\Omega^*)$ is a CFT on $\Cb$ in the sense of definition \ref{def:CFT-plane}, with correlators $C_n(\Omega^*|\cdots)$. Indeed, (C5$'$) reduces to (C5) if we fix $u$ to be $\Omega^*$ and impose the growth condition $\lim_{\zeta \to \infty} \zeta^{-3} f(\zeta) = 0$.

\begin{remark} \label{rem:F-trivial}
As was noted in remark \ref{rem:CFT-plane}\,(iii), when $F$ is a trivial $\Vir \oplus \Vir$-module, $(F,M)$ is a topological field theory. One can easily convince oneself that then the pair $(F,M)$ is just a commutative, associative algebra with multiplication $M : F \otimes F \to F$ ($F$ is concentrated in grade $0$, so $\overline F = F$, and $M$ is position independent). Indeed, a useful way to think about a conformal field theory on the complex plane is as a generalisation of a commutative, associative algebra where the product depends on a non-zero complex parameter.
\end{remark}

Continuing the analogy with algebra, let us define a \textsl{homomorphism of CFTs} $(F,M)$ \textsl{and} $(F',M')$ to be a $\Vir \oplus \Vir$-intertwiner $f : F \to F'$ such that $f \circ M_x = M'_x \circ (f \otimes f)$. Since $f(F^{\Delta}) \subset F'{}^{(\Delta)}$, the map $f$ is well-defined as a map $\overline F \to \overline F'$. By an \textsl{ideal} in $F$ we mean a $\Vir \oplus \Vir$-submodule $I$ of $F$ such that for all $\iota \in I$, $\phi \in F$ and $x \in \Cb^\times$
we have $M_x(\iota \otimes \phi) \in \overline I$ and $M_x(\phi \otimes \iota) \in \overline I$ (actually one of the two conditions implies the other). The kernel of a homomorphism is an ideal. Given an ideal $I \subset F$, we obtain a CFT on the quotient $F/I$ such that the canonical projection $\pi : F \to F/I$ is a homomorphism of CFTs.

Another class of examples of ideals is the following. Let $(F,M,\Omega^*)$ be a CFT on $\Cb$. Let $F_0$ be the kernel of the 2-point correlator, i.e.\ fix $z \neq w$ and define
\be \label{eq:F0-def}
  F_0 = \big\{ \eta \in F \,\big|\, \langle \phi(z) \eta(w) \rangle = 0 \text{ for all } \phi \in F \big \} \ .
\ee
From \eqref{eq:explicit-2-pt} one concludes that $F_0$ is independent of $z,w$. It follows from (C5) that $F_0$ is a  $\Vir \oplus \Vir$-submodule of $F$. Let $\eta \in F_0$ and $\phi,\psi \in F$. By expressing  the 3-point correlator $\langle \eta(x) \phi(y) \psi(z) \rangle$ as a limit of two-point correlators via (C3) in two ways, one involving $M_{x-y}(\eta \otimes \phi)$ and one $M_{y-z}(\phi \otimes \psi)$, one sees that $F_0$ is an ideal in $F$. Again because of (C3), a correlator $\langle \phi_1(z_1) \cdots \phi_n(z_n) \rangle$ is zero if at least one of the $\phi_i$ is from $F_0$.

\begin{definition} \label{def:non-deg-CFT}
A conformal field theory on the complex plane $(F,M,\Omega^*)$ is \textsl{non-degenerate} if $F_0$ as defined in \eqref{eq:F0-def} is $\{0\}$.
\end{definition}

\begin{remark} \label{rem:sphere2}
(i) If $\Omega^*(F_0) = \{0\}$, the CFT on the quotient $F/F_0$ has an out-vacuum induced by $\Omega^*$ and is non-degenerate. On the level of correlators, one cannot tell the difference between $F$ and $F/F_0$ and hence it is common to restrict one's attention to non-degenerate CFTs on $\Cb$. However, the device of background states allows one to obtain interesting correlators also for degenerate CFTs.
\\[.3em]
(ii) Let $f : F \to G$ be a homomorphism of the CFTs $(F,M)$ and $(G,N)$. Let $\Gamma^* \in \overline G^*$ be a primary $sl(2,\Cb)$-invariant state. Then $\Omega^* := \Gamma^* \circ f$ is a primary $sl(2,\Cb)$-invariant state in $\overline F^*$. Because of $f \circ M_x = N_x \circ (f \otimes f)$, the correlators of $(F,M,\Omega^*)$ and $(G,N,\Gamma^*)$ are related by
\be \label{eq:corr-after-hom}
  \langle \phi_1(z_1) \cdots \phi_n(z_n) \rangle_{F}
  =
  \langle \phi'_1(z_1) \cdots \phi'_n(z_n) \rangle_{G} ~~,
  \quad \text{where} \quad \phi_i' = f(\phi_i) \ .
\ee
(iii) With the notation of (ii), if $\ker(f) \neq \{0\}$, it follows from \eqref{eq:corr-after-hom} that the CFT $(F,M,\Omega^*)$ is necessarily degenerate. Explicitly, $\langle \phi(z) \psi(w) \rangle_F 
= \langle \Omega^* , M_{z-w}(\phi \otimes \psi) \rangle 
= \langle \Gamma^* , f \circ M_{z-w}(\phi \otimes \psi) \rangle
= \langle \Gamma^* , N_{z-w}(f(\phi) \otimes f(\psi)) \rangle = \langle \phi'(z) \psi'(w) \rangle_G$, so that $\ker(f) \subset F_0$.
\\[.3em]
(iv) If there is an isomorphism $f : F \to \overline F^*$ of $\Vir \oplus \Vir$-modules, one can define the non-degenerate pairing $(u,v) = {}^{f(u)}\langle v(0) \rangle = \langle f(u) ,  v\rangle$ on $F \times F$. This pairing is invariant in the sense that $(L_m u,v) = (u, L_{-m} v)$ and $(\overline L_m u,v) = (u, \overline L_{-m} v)$ for all $u,v \in F$ and $m \in \Zb$. In this situation one can also ask if the inversion $z \mapsto 1/z$ is a symmetry of the theory, i.e.\ if, for all $\phi_i,\psi_i \in F$,
\be
  {}^{f(\psi_1)} \big\langle \phi_1(z_1) \cdots \phi_n(z_n) \, \psi_2(0) \big\rangle
  = {}^{f(\psi_2)} \big\langle \phi_1'(1/z_1) \cdots \phi_n'(1/z_n)\, \psi_1(0) \big\rangle \ ,
\ee
where (see, e.g., \cite[Sect.\,3.2]{Gaberdiel:1999mc})
\be
  \phi_i' = \exp\!\big( \ln(-z_i^{-2}) L_0 + \ln(-\bar z_i^{-2}) \overline L_0 \big) \, \exp\!\big( -z_i^{-1} L_1 - \bar z_i^{-1} \overline L_1 \big) \, \phi_i \ .
\ee
An inversion-covariant CFT with background states provides us with an alternative way to define correlators on the Riemann sphere as compared to remark \ref{rem:sphere1}\,(ii). 
For a Riemann sphere with two or more insertions, choose an isomorphism with $\Cb \cup \{ \infty \}$ which maps one of the insertion points to infinity and evaluate the resulting configuration with $C_n(f(\,\cdot\,)|\cdots)$. Different such choices are related by a M\"obius transformation which maps one of the insertion points (including infinity) to infinity. (As in remark \ref{rem:sphere1}\,(ii) one needs to choose local coordinates around the insertions, we skip the details.)
\end{remark}

In section \ref{sec:R1-OPE}, we will encounter the special situation of a (conjectural) CFT on $\Cb$ with background states $(F,M)$ which has a surjective homomorphism $\pi : F \to \Cb$ to the trivial CFT $(\Cb, \cdot)$ with one-dimensional state space, where `$\cdot$' stands for the product on $\Cb$. Because $\pi$ is a $\Vir \oplus \Vir$-intertwiner, this situation can only occur for $c=0$. If we take $\Omega^* = \pi$ (which is a primary $sl(2,\Cb)$-invariant state in $\overline F^*$) as out-vacuum, the correlators of $(F,M,\Omega^*)$ satisfy
\be \label{eq:project-to-c=0-minmod}
  \langle \phi_1(z_1) \phi_2(z_2) \cdots \phi_n(z_n) \rangle_{F} = \pi(\phi_1) \cdot \pi(\phi_2) \cdots \pi(\phi_n) \ .
\ee
Thus, if we want to tell the theory $(F,M)$ apart from the trivial theory we must consider correlators with background states other than $\Omega^*$. This small observation is the reason for including this subsection.

\subsection{Modular invariant partition functions}\label{sec:mod-inv-pf}

Given a $\Vir \oplus \Vir$-module $F$ as in the definition \ref{def:CFT-plane}, the \textsl{graded trace} of $F$ is
\be
  Z(F;\tau) = \mathrm{tr}_F \Big( \, q^{L_0 - C/24} \, \bar q{}^{\overline L_0 - \overline C/24} \Big) \ ,
\ee
where $q = e^{2 \pi i \tau}$ and $\tau$ is a complex number with $\mathrm{Im}(\tau)>0$. Suppose for the moment that $C$ and $\overline C$ both act on $F$ by multiplication with a number $c$. Then we can rewrite $Z$ as
\be
Z(F;\tau) = \sum_{\Delta} e^{-2 \pi \mathrm{Im}(\tau) \cdot(\Delta-c/12)} \, \, \mathrm{tr}_{F^{(\Delta)}}\! \Big( e^{2 \pi i \mathrm{Re}(\tau)(L_0 - \overline L_0)}\Big) \ .
\ee
The graded trace may be ill-defined, for example $L_0$ and $\overline L_0$ might have infinite common eigenspaces or the sum over $\Delta$ may not converge. If $Z(F;\tau)$ is well-defined, it is a generating function sorting states in $F$  by their scaling dimension (or energy) -- with dual parameter $\mathrm{Im}(\tau)$ -- and by their spin with dual parameter $\mathrm{Re}(\tau)$.

\medskip

Given a conformal field theory on the complex plane $(F,M,\Omega^*)$, one may ask if the set of correlators $C_n$ determined by it is part of a larger family of correlators which allow Riemann surfaces other than the complex plane. The simplest additional surface would be a torus of complex modulus $\tau$,
\be
  T_\tau = \Cb \, / \,(\Zb + \tau \Zb) \ .
\ee
A correlator of $n$ fields on $T_\tau$ is then required to be related to a sum of correlators of $n{+}2$ fields on the Riemann sphere by ``inserting a sum over intermediate states''. Schematically, this is shown in figure \ref{fig:torus-from-plane}. We will not go into any detail, but we point out that the sum is over a basis $\{ \varphi_\alpha \}$ of $F$ and a basis $\{ \varphi'_\alpha \}$ dual to the first basis with respect to the 2-point correlator on the Riemann sphere. For this procedure to make sense, the 2-point correlator has to be non-degenerate. Such correlators on the Riemann sphere can be obtained from a non-degenerate CFT on $\Cb$ via remark \ref{rem:sphere1}\,(ii) or from a CFT with background states and an isomorphism $F \to \overline F^*$ as in remark \ref{rem:sphere2}\,(iv).

\begin{figure}[tb] 
$$
   \raisebox{-25pt}{
  \begin{picture}(100,80)
   \put(0,0){\scalebox{0.50}{\includegraphics{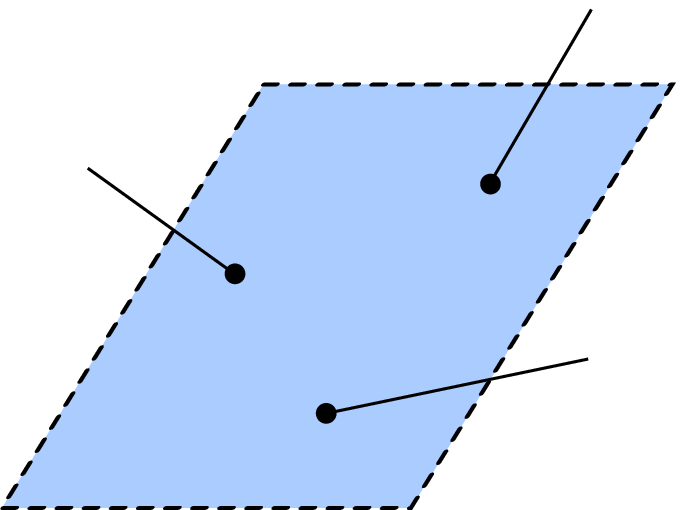}}}
   \put(0,0){
     \setlength{\unitlength}{.50pt}\put(-11,-13){
     \put( 11,110)  {$ \phi_1 $}
     \put(186,155)  {$ \phi_2 $}
      \put(185, 55)  {$ \phi_3 $}
     }\setlength{\unitlength}{1pt}}
  \end{picture}}  
  ~~~\sim~~~ \sum_\alpha~
   \raisebox{-50pt}{
  \begin{picture}(100,120)
   \put(0,8){\scalebox{0.50}{\includegraphics{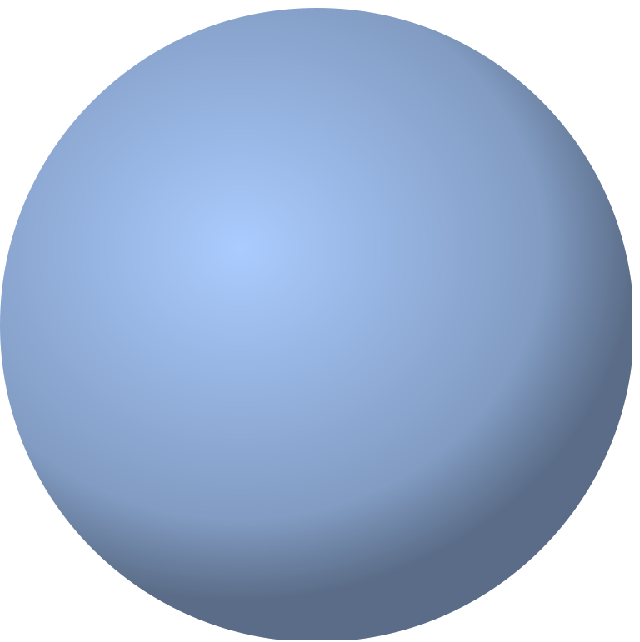}}}
   \put(21,0){\scalebox{0.50}{\includegraphics{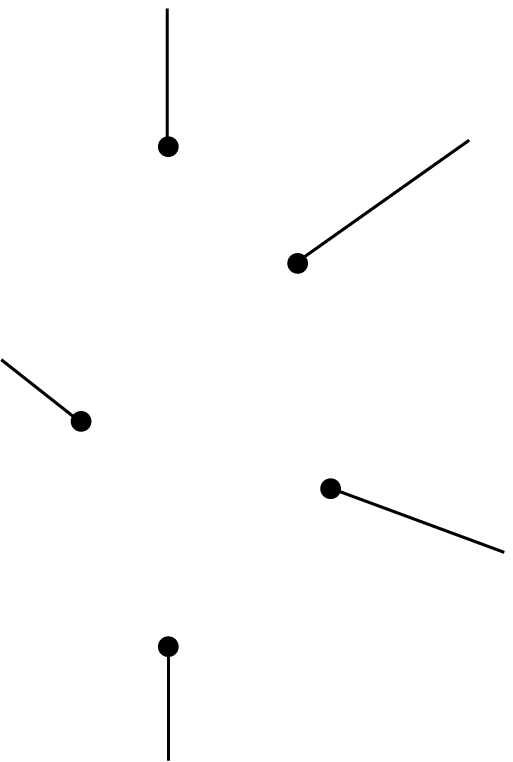}}}
   \put(0,0){
     \setlength{\unitlength}{.50pt}\put(-8,-7){
     \put(93,-8)  {$ \varphi_\alpha $}
     \put(93,238)  {$ \varphi'_\alpha $}
     \put(25,124)  {$ \phi_1 $}
     \put(189,187)  {$ \phi_2 $}
     \put(198,62)  {$ \phi_3 $}
     }\setlength{\unitlength}{1pt}}
  \end{picture}}  
$$
\caption{A correlator of $n$ bulk fields $\phi_i \in F$ on the torus can be expressed as a sum of correlators of $n{+}2$ fields on the Riemann sphere, where the two additional bulk fields are taken from a basis $\{ \varphi_\alpha \}$ of $F$ and its dual basis $\{ \varphi_\alpha' \}$.}
\label{fig:torus-from-plane}
\end{figure}

If the system of correlators on the Riemann sphere form part of a larger collection defined on other Riemann surfaces including the torus, then the amplitude for the torus $T_\tau$ is described by the function $Z(F;\tau)$.
It must therefore only depend on the conformal equivalence class of $T_\tau$, that is, it must be \textsl{modular invariant},
\be
  Z(F;-1/\tau) ~=~ Z(F;\tau{+}1) ~=~  Z(F;\tau) \ .
\ee
The function $Z(F;\tau)$ is called the \textsl{partition function} of the CFT.

\begin{remark} \label{rem:mod-inv-big}
(i) Physically, if the CFT arises as a continuum limit of a two-dimensional critical lattice model, one would expect its partition function to be modular invariant since the lattice model could equally be evaluated in a finite geometry with periodic boundary conditions.
\\[.3em]
(ii) For non-logarithmic rational conformal field theories, modular invariance of the partition function proved to be very constraining. Understanding which  $\Vir \oplus \Vir$-modules $F$ (or $\Vc \otimes_{\Cb} \Vc$-modules for a vertex operator algebra $\Vc$) give rise to a modular invariant graded trace is an important step in attacking classification questions. A typical behaviour in non-logarithmic rational examples is that if $\Vc$ has order $N$ distinct irreducible representations, then a modular invariant $F$ splits into order $N^2$ irreducible direct summands. In this sense, modular invariant CFTs for a fixed $\Vc$ are all ``equally complicated''.\footnote{
  This statement can be made more precise: a non-logarithmic rational CFT with symmetry $\Vc \otimes_{\Cb} \Vc$ which has a unique vacuum and is modular invariant has the property that the categorical dimension of $F$ is equal to the global dimension of $\Rep(\Vc)$, see \cite[Thm.\,3.4]{Kong:2008ci} for details. In particular all such $F$ have the same categorical dimension.}
\end{remark}

\subsection{Consistency conditions for CFT on the upper half plane} \label{sec:CFT-UHP}

\begin{figure}[tb] 
$$
  \raisebox{-35pt}{
  \begin{picture}(165,80)
   \put(0,0){\scalebox{0.50}{\includegraphics{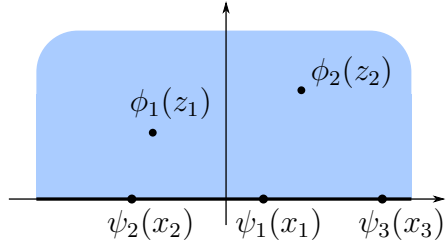}}}
   \put(0,0){
     \setlength{\unitlength}{.50pt}\put(-5,-6){
     \put(97,93)  {$ \phi_1(z_1) $}
     \put(234,116)  {$ \phi_2(z_2) $}
     \put(177,1)  {$ \psi_1(x_1) $}
     \put(80,1)  {$ \psi_2(x_2) $}
     \put(270,1)  {$ \psi_3(x_3) $}
     }\setlength{\unitlength}{1pt}}
  \end{picture}}  
$$
\caption{Bulk and boundary field insertions on the upper half plane. Here $\phi_i \in F$, $\psi_i \in B$ and $\mathrm{Im}(z_i)>0$, $x_i \in \Rb$. The figure describes the correlator $\big\langle \psi_1(x_1)\psi_2(x_2)\psi_3(x_3)\phi_1(z_1)\phi_2(z_2)\big\rangle = U_{3,2}(x_1,x_2,x_3,z_1,z_2,\psi_1,\psi_2,\psi_3,\phi_1,\phi_2)$.}
\label{fig:UHP-correlator}
\end{figure}

The description of conformal field theory on the upper half plane is very similar to that on the complex plane. The main difference is that there are now two spaces of fields: \textsl{bulk fields}, which are the ones already discussed in section \ref{sec:CFT-C} and are inserted in the interior of the upper half plane, and \textsl{boundary fields}, which must be inserted on the real axis, cf.\ figure \ref{fig:UHP-correlator}. Correspondingly, the collection of correlators $(U_{m,n})_{m,n}$ now depends on two integers, $m$ counting the number of boundary fields and $n$ counting the number of bulk fields. Let $\Hb := \{ z \in \Cb \,|\, \mathrm{Im}(z)>0 \}$ be the open upper half plane. Then
\be
  U_{m,n} : (\Rb^m{\setminus}\text{diag}) \times (\Hb^n{\setminus}\text{diag}) \times B^m \times F^n \longrightarrow \Cb \ ,
\ee
where $\Hb^n{\setminus}\text{diag}$ and $\Rb^n{\setminus}\text{diag}$ refer the set of $n$ mutually distinct points. The customary notation is, for $\phi_i \in F$, $\psi_i \in B$, $z_i \in \Hb$ and $x_i \in \Rb$,
\be
\begin{array}{l}
  U_{m,n}(x_1,\dots,x_m,z_1,\dots,z_n,\psi_1,\dots,\psi_m,\phi_1,\dots,\phi_n) 
  \\[.5em]
  \hspace{2em} = \big\langle \psi_1(x_1) \cdots \psi_m(x_m) \phi_1(z_1) \cdots \phi_n(z_n) \big
  \rangle \ .
\end{array}
\ee

\begin{remark}
In the discussion of correlators on the upper half plane above we are implicitly assuming that the entire real axis carries the same boundary condition. In more generality one would allow different intervals to carry different boundary conditions. We will not treat this case explicitly, but we note that it is included in the present formalism: One can always think of the real line with several boundary conditions as a real line with a single boundary condition given by their superposition, together with appropriate boundary field insertions that project to the individual constituents.
\end{remark}

\begin{definition} \label{def:CFT-uhp}
A \textsl{conformal field theory on the upper half plane} is a tuple
$$
(F,M,\Omega^*;B,m,\omega^*;b) \ ,
$$
where
\begin{itemize}
\item $(F,M,\Omega^*)$ is a CFT on the complex plane,
\item $B$ (the \textsl{space of boundary fields}) is a $\Vir$-module which is a direct sum of generalised $L_0$-eigenspaces $B^{(h)}$, whose generalised $L_0$-eigenvalues $h$ are bounded from below.\footnote{
  The space $B$ is automatically locally finite as a $\Cb L_0$ module (cf.\  section \ref{sec:CFT-C} for the definition of `locally finite'). This is so because any vector $v$ in $B$ can be written as a finite sum of vectors $v_h \in B^{(h)}$, and on each $v_h$ we have $(L_0-h)^N v_h = 0$ for some large enough $N$. For $F$, the same argument only gives local finiteness as a $\Cb(L_0 + \overline L_0)$ module, which is why local finiteness as a $\Cb L_0 \oplus \Cb \overline L_0$ module was included as a separate condition.}
\item $m$ (the \textsl{boundary OPE}) is a map $\Rb_{>0} \times B \otimes_{\Cb} B \to \overline B$, linear in $B \otimes_{\Cb} B$,
\item $\omega^*$ (the \textsl{out-vacuum on the upper half plane}) is a linear function $B^{(0)} \to \Cb$,
\item $b$ (the \textsl{bulk-boundary map}) is a map $\Rb_{>0} \times F \to \overline B$, linear in $F$,
\end{itemize}
such that there exists a collection of correlators $(U_{m,n})_{m,n}$, with $m,n \in \Zb_{\ge 0}$ and $(m,n) \neq (0,0)$, which satisfy (B1)--(B5) in appendix \ref{app:B1-B5}, as well as the normalisation condition $U_{1,0}(0,\psi) = \langle \omega^* , \psi \rangle$ for all $\psi \in B$.
\end{definition}

Conditions (B1)--(B5) are the same in spirit as (C1)--(C5), just more tiresome to write down, and they have been moved to an appendix for this reason. Here we merely note that there are now three different types of short distance expansions. The OPE of two bulk fields as in (C3), the expansion of a bulk field $\phi$ close to the boundary in terms of boundary fields via $b_y(\phi) \in \overline B$, and the OPE of two boundary fields $(\psi ,\psi') \mapsto m_x(\psi \otimes \psi') \in \overline B$. 

\medskip

The basic class of examples is provided by the Virasoro minimal models with $A$-series modular invariant. In this case the central charge is $c = 1 - 6 (p{-}q)^2 / pq$ with $p,q \ge 2$ and coprime. Denote by $i$ a Kac-label for that central charge and by $R_i$ the corresponding irreducible representation\footnote{
  Of course there are many more irreducible representations of the Virasoro algebra with this value of the central charge, but only those corresponding to entries in the Kac table are also representations of the simple Virasoro vertex operator algebra with this central charge.
}
of $\Vir$. Then $F = \bigoplus_{i} R_i \otimes_{\Cb} R_i$, where $i$ runs over all Kac-labels (modulo their $\Zb/2$-identification) and for $B$ we can take the vacuum representation $B = R_{(1,1)}$ of $L_0$-weight $0$. There are many more possible spaces of boundary fields for this bulk theory, namely $B = U \otimes_f U^*$, where $U$ is any direct sum of the $R_i$ and $\otimes_f$ denotes the fusion product (resulting again in a direct sum of the $R_i$ according to the fusion rules). 

Let us stress again the point made in the introduction and in remark \ref{rem:mod-inv-big}\,(ii). The space of bulk fields in a modular invariant CFT tends to be `big' in the sense that it involves many different irreducible representations (in logarithmic CFT this should be taken as a statement about the composition series or about the character). On the other hand, there often exists a CFT on the upper half plane with bulk fields $F$ and a much simpler set of boundary fields $B$ involving only very few irreducible representations. One may thus attempt to first gain control over the boundary theory and then try to construct a fitting bulk theory. This is the topic of the next subsection.

\subsection{From boundary to bulk} \label{sec:boundary-to-bulk}

In this subsection we make precise the following idea: Given a boundary theory, i.e.\ a space of boundary fields and their correlators on the upper half plane, try to build the `biggest bulk theory' that can be made to fit to this boundary theory. We will find that this bulk theory, if it exists, is unique. The algebraic version of the question of existence and the description of the data $F$, $M$, and $b$ will be addressed in section \ref{sec:alg-reform}.

\begin{remark}
In sections \ref{sec:CFT-C}--\ref{sec:CFT-UHP} we have discussed CFTs with Virasoro symmetry. This can be generalised to other vertex operator algebras $\Vc$ as underlying symmetry of the CFT. The space of bulk fields $F$ is then a representation of $\Vc \otimes_{\Cb} \Vc$ and the space of boundary fields $B$ a representation of $\Vc$. The coinvariance conditions become those of $\Vc$ and will contain the Virasoro conditions in (C5) and (B5) as a subset. It is important for us to allow this generalisation, even if we avoided spelling out the formalism for general $\Vc$. The reason is that in the examples we study, we want the category of representations $\Rep(\Vc)$ to have certain finiteness properties (in particular a finite number of irreducibles, more details will follow in condition (PF) in section \ref{sec:deligne-prod} below). If we were only to allow $\Vc$ to be the simple Virasoro vertex operator algebra at a given central charge, the finiteness conditions would limit us to (non-logarithmic) minimal models. Therefore, we allow for more general $\Vc$, in particular the vertex operator algebra $\Wc$ at $c=0$ for the $W_{2,3}$-model discussed in section \ref{sec:W23-model}. It is not currently clear to us to which extent the construction below is the right ansatz if we were to drop these finiteness conditions.
\end{remark}

\begin{definition} \label{def:theo-on-bnd}
A \textsl{boundary theory} is a triple $(B,m,\omega^*)$ with $B$, $m$, $\omega^*$ as in definition \ref{def:CFT-uhp}, such that there exists a collection of correlators on the upper half plane $( U_{m,0} )_{m \in \Zb_{>0}}$ involving only boundary fields, and which satisfy (B1)--(B5) restricted to $U_{m,0}$, as well as $U_{1,0}(0,\psi) = \langle \omega^* , \psi \rangle$ for all $\psi \in B$.
\end{definition}

Thus, a CFT on the upper half plane consists of a CFT on the complex plane, a boundary theory, and a consistent interaction between them via the bulk-boundary map. In analogy with definition \ref{def:non-deg-CFT} we say

\begin{definition} \label{def:non-deg-bnd}
A boundary theory $(B,m,\omega^*)$ is called \textsl{non-degenerate} if for all $x,y \in \Rb$ with $x \neq y$ and $\psi \in B$ there is a $\psi' \in B$ such that $\langle \psi(x) \psi'(y) \rangle \neq 0$.
\end{definition}

\begin{remark} \label{rem:bnd-bgstate}
(i)
Continuing from remark \ref{rem:F-trivial}, it is again helpful to briefly consider the much simpler special case of topological field theory. One checks that a non-degenerate boundary theory $(B,m,\omega^*)$ with trivial $\Vir$-action on $B$ is the same as an associative but not necessarily commutative algebra $B$, together with a map $\omega^* : B \to \Cb$ such that $(a,b) \mapsto \langle \omega^* , a \cdot b \rangle$ is a non-degenerate pairing on $B$.
\\[.3em]
(ii)
As in section \ref{sec:genout-ideal} one can introduce boundary theories with background states $(B,m)$ which involve a modified version of condition (B5). We have chosen not to discuss boundary theories with background states in detail. The construction of the `biggest bulk theory' below is therefore formulated in terms of a non-degenerate boundary theory, but one could alternatively use a boundary theory with background states.
\end{remark}

Let us now fix a non-degenerate boundary theory $(B,m,\omega^*)$. The `biggest bulk theory' will be characterised as a terminal object in a category of pairs $\Pc$, which we proceed to define. An object of $\Pc$ is a pair $(F', b')$, where 
\begin{itemize}
\item 
$F'$ is a `candidate space of bulk fields'. Namely it is a $\Vir \oplus \Vir$-module with boundedness condition as in definition \ref{def:CFT-plane} (or more generally a $\Vc \otimes_{\Cb} \Vc$-module).
\item
$b'$ is a `candidate bulk-boundary map'. By this we mean that $b' : \Rb_{>0} \times F' \to \overline B$ as in definition \ref{def:CFT-uhp}, such that there exists a function $U'_{1,1} : \Rb \times \Hb \times B \times F' \to \Cb$ (a `candidate correlator' of one bulk field and one boundary field) which satisfies the derivative property (B4), the coinvariance condition (B5), and which for $|x| > y$ can be expressed through the candidate bulk-boundary map and the boundary 2-point correlator as
\be \label{eq:bulk-bnd-coinv-Pdef}
   U'_{1,1}(x,iy,\psi,\phi) = \lim_{h \to \infty} U_{2,0}(x,0,\psi, P_h \circ b'_{y}(\phi))
   \qquad ; ~~ \psi \in B ~,~~ \phi \in F' \ .
\ee
Here $U_{2,0}$ is a boundary correlator from definition \ref{def:theo-on-bnd} which is uniquely fixed by $(B,m,\omega^*)$, and $P_h$ is the canonical projection $\overline B \to \bigoplus_{d \le h} B^{(d)}$, analogous to $P_\Delta$ in (C3).
\item
$b'$ has to be \textsl{central}, a condition which we will detail momentarily.
\end{itemize}

\begin{figure}[tb] 
\begin{center}
\raisebox{35pt}{a)}
   \raisebox{-35pt}{
  \begin{picture}(165,80)
   \put(0,0){\scalebox{0.50}{\includegraphics{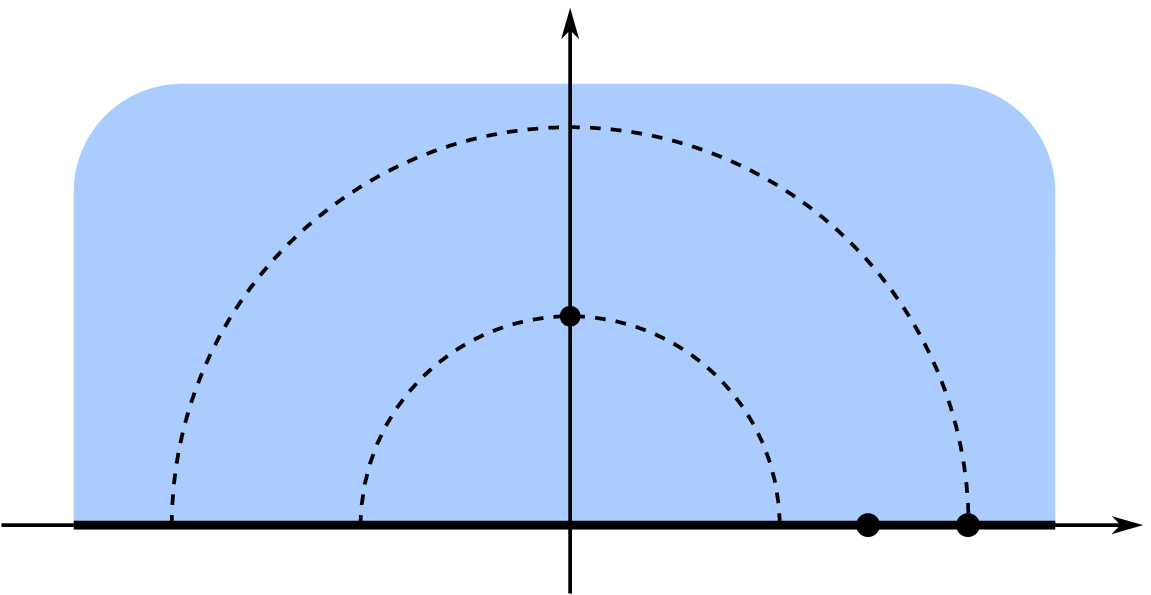}}}
   \put(0,0){
     \setlength{\unitlength}{.50pt}\put(-5,-5){
     \put(153,92)  {\scriptsize $ \phi $}
     \put(176,72)  {\scriptsize $ iy $}
     \put(289,35)  {\scriptsize $ \psi' $}
     \put(255,35)  {\scriptsize $ \psi $}
     \put(280,9)  {\scriptsize $ s $}
     \put(250,9)  {\scriptsize $ x $}
     \put(225,9)  {\scriptsize $ y $}
     \put(40,9)  {\scriptsize $ {-}s $}
     \put(93,9)  {\scriptsize $ {-}y $}
     }\setlength{\unitlength}{1pt}}
  \end{picture}}  
\hspace{5em}
\raisebox{35pt}{b)}
   \raisebox{-35pt}{
  \begin{picture}(165,80)
   \put(0,0){\scalebox{0.50}{\includegraphics{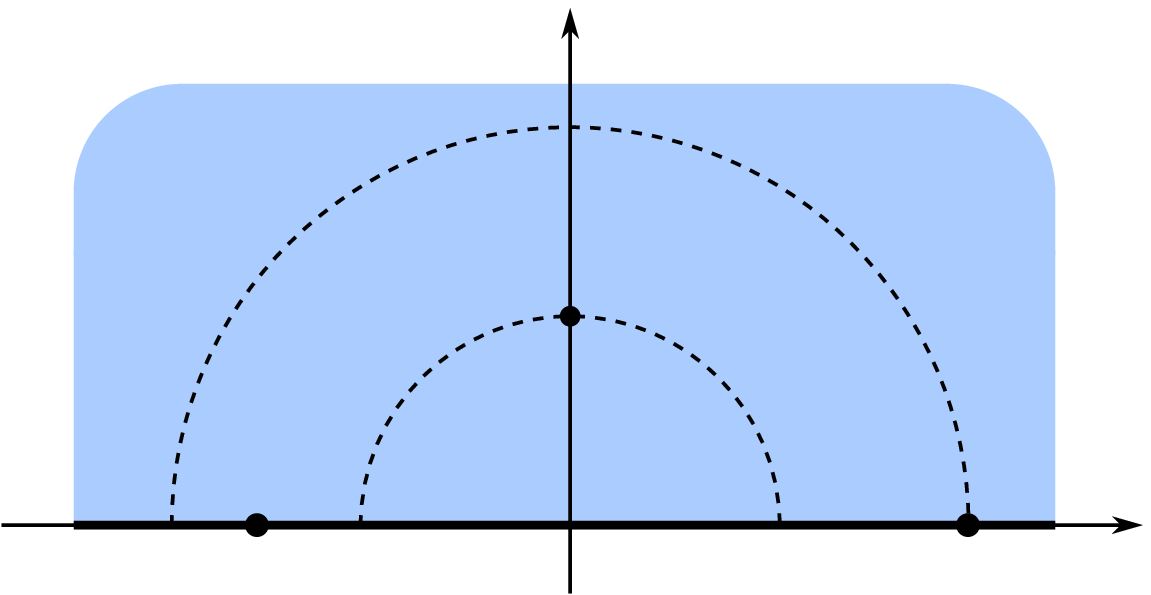}}}
   \put(0,0){
     \setlength{\unitlength}{.50pt}\put(-5,-5){
     \put(153,92)  {\scriptsize $ \phi $}
     \put(176,72)  {\scriptsize $ iy $}
     \put(289,35)  {\scriptsize $ \psi' $}
     \put(77,35)  {\scriptsize $ \psi $}
     \put(280,9)  {\scriptsize $ s $}
     \put(73,9)  {\scriptsize $ x $}
     \put(225,9)  {\scriptsize $ y $}
     \put(40,9)  {\scriptsize $ {-}s $}
     \put(93,9)  {\scriptsize $ {-}y $}
     }\setlength{\unitlength}{1pt}}
  \end{picture}}  
\end{center}
\caption{Geometric setting for the centrality condition. The limit \eqref{eq:U'21-as-lim} defining $U'_{2,1}(s,x,iy,\psi',\psi,\phi)$ is assumed to exist for $y<|x|$. This gives rise to two functions $U_{\pm}(x)$ on the open interval $({-}s,s)$: $U^+(x)$ equals $U'_{2,1}$ for $x \in (y,s)$ as shown in a), while $U^-(x)$ equals $U'_{2,1}$ for $x \in ({-}s,{-}y)$ as shown in b).}
\label{fig:centrality-condition}
\end{figure}

To formulate the centrality condition, we define a candidate correlator $U'_{2,1}$ of two boundary fields $\psi,\psi' \in B$ and one bulk field $\phi \in F'$ via 
\be \label{eq:U'21-as-lim}
   U'_{2,1}(s,x,iy,\psi',\psi,\phi) = \lim_{h \to \infty} U_{3,0}(s,x,0,\psi',\psi,P_h \circ b'_{y}(\phi)) \ ,
\ee
at least for $y<|x|<s$ (we take $s>0$); we assume (as part of the centrality condition) that the limit exists. There are then two disconnected domains for $x$: it can be in $(y,s)$ or in $({-}s,{-}y)$, see figure \ref{fig:centrality-condition} for an illustration. We now try to use the derivative property (B5) in the form 
\be
  \tfrac{d}{dx} U'_{2,1}(s,x,iy,\psi',\psi,\phi) = U'_{2,1}(s,x,iy,\psi',L_{-1} \psi,\phi)
\ee
to extend the function $U'_{2,1}$ to all of $({-s},s)$. Depending on whether we start from $(y,s)$ or in $({-}s,{-}y)$, we a priori obtain two different functions  $U^+(x)$ and $U^-(x)$ on $({-s},s)$. We call $b'$ \textsl{central} if these two extensions coincide: $U^+(x) = U^-(x)$ for all $x \in ({-}s,s)$.

The centrality condition holds automatically in a CFT on the upper half plane (because the correlator $U_{2,1}$ is a smooth function and satisfies the expansion conditions (B3)). The point here, of course, is to impose only a small subset of the conditions a CFT has to satisfy. For example, to define the pairs $(F',b')$ we are only ever looking at candidate correlators with one bulk field and one or two boundary fields.

\medskip

But back to the category of pairs $\Pc$. Now that we have defined its objects, it is easy to give the space of morphisms from $(F',b')$ to $(F'',b'')$. It consists of all $\Vir \oplus \Vir$-intertwiners $f : F' \to F''$ (or more generally $\Vc \otimes_{\Cb} \Vc$-intertwiners) such that the diagram of maps
\be
\raisebox{2em}{\xymatrix@R=1.5em@C=1.5em{
F' \ar[rr]^{f} \ar[dr]_-{b'_y} && F'' \ar[dl]^-{b''_y}
\\
& \overline B
}}
\ee
commutes for all $y>0$.

An object $T$ in a category $\Cc$ is called \textsl{terminal} if for every object $U \in \Cc$ there exists a unique morphism $U \to T$. A category $\Cc$ may or may not have a terminal object, but if one exists, it is unique up to unique isomorphism (take two terminal objects $T$ and $T'$ and play with maps between them). We have now gathered all ingredients to define:
\be \label{eq:terminal-obj-in-P}
  (F(B),b(B))
  ~\text{ is a terminal object in }~ \Pc \ .
\ee  
We want to think of $F(B)$ as the maximal space of bulk fields which can be consistently joined to our prescribed boundary theory $(B,m,\omega^*)$, and of course we take $b(B)$ as the bulk-boundary map. The next lemma makes this interpretation precise.

\begin{lemma} Let $B$ be a non-degenerate boundary theory. Let $(F(B),b(B))$ be a terminal object in $\Pc$ and let $(F',b')$ be an arbitrary object in $\Pc$.
\begin{itemize}
\item[(i)] The kernel of $b'_y : F' \to \overline B$ is independent of $y$.
\item[(ii)]  The map $b(B)_y : F(B) \to \overline B$ is injective for each $y>0$.
\item[(iii)]  If $b'_y : F' \to \overline B$ is injective for $y>0$, then there is an injective $\Vir \oplus \Vir$-intertwiner $\iota : F' \to F(B)$ such that $b'_y = b(B)_y \circ \iota$ for all $y>0$.
\end{itemize}
\end{lemma}

\begin{proof} 
For part (i), let $K(y)$ be the kernel of $b'_y : F' \to \overline B$. By non-degeneracy of the boundary theory, the kernel of $b'_y$ is determined by $U'_{1,1}$.
\\[.3em]
\nxt $K(y)$ is a $\Vir \oplus \Vir$-module: Use the coinvariance property to show that $U'_{1,1}(x,iy,\psi,\phi) = 0$ for all $\psi$ implies $U'_{1,1}(x,iy,\psi,L_m\phi) = 0$ and $U'_{1,1}(x,iy,\psi,\overline L_m\phi) = 0$ for all $\psi$.
\\[.3em]
\nxt $K(y) = K(y')$ for all $y,y'>0$: There exists a global conformal transformation $\Hb \to \Hb$ which leaves a point $x \in \Rb$ invariant and maps $y$ to $y'$. The coinvariance property can be integrated to give $U'_{1,1}(x,iy,\psi,\phi) = U'_{1,1}(x,iy',\psi',\phi')$, where $\psi'$ and $\phi'$ are obtained from $\psi$ and $\phi$ by an appropriate exponential of modes $L_0$, $L_1$, $\overline L_0$, $\overline L_1$. Using the previous point we see that $\phi \in K(y')$ implies $\phi' \in K(y')$ and thus $\phi \in K(y)$. Together with the inverse transformation one finds $K(y) = K(y')$.

To see (ii),  let $K$ be the kernel of $b(B)_y$ and let $e : K \to F(B)$ be the embedding map. The triangle
\be
\raisebox{2em}{\xymatrix@R=1.5em@C=1.5em{
K \ar[rr]^-{e} \ar[dr]_{0} && F(B) \ar[dl]^{b(B)_y}
\\
& \overline B
}}
\ee
commutes for all $y>0$ (since the kernel is independent of $y$). By the terminal object property, the map $K \to F(B)$ which makes the above triangle commute is unique, and therefore $e=0$. Hence also $K = \{ 0 \}$. 

Part (iii) is now trivial. The existence of $\iota$ follows from the terminal object property. Since $b'_y = b(B)_y \circ \iota$ with $b'_y$ and $b(B)_y$ injective, also $\iota$ must be injective.
\end{proof}

\begin{remark}
That the candidate bulk-boundary map $b'$ in a pair $(F',b')$ is injective has the physical interpretation that all bulk fields can be distinguished in upper half plane correlators. If a bulk field $\phi$ from the kernel of the bulk-boundary map is inserted in a correlator on the upper half plane, this correlator vanishes, irrespective of the other field insertions. Thus by the above lemma, the space $F(B)$ is maximal in the sense that any candidate space of bulk fields $(F',b')$, for which all bulk fields can be distinguished in upper half plane correlators, can be embedded in $F(B)$. This embedding is compatible with the candidate bulk-boundary map.
\end{remark}

It remains to address the question of existence of the terminal object $(F(B), b(B))$, to see how the OPE of bulk fields in $F(B)$ is determined, to verify its associativity and commutativity, and to investigate the compatibility of bulk and boundary OPE with the bulk-boundary map $b(B)$. To do so, it is best to leave behind the infinite dimensional vector spaces underlying $F(B)$ and $B$ and the infinite set of coinvariance conditions on the correlators, and to take a fresh look at the problem from the more abstract viewpoint of algebras in braided monoidal categories.\footnote{
  We should also address the non-degeneracy of the 2-point correlator and verify modular invariance. Unfortunately, we currently do not know how to do this at the level of generality used in section \ref{sec:alg-reform}. We can only point to non-logarithmic rational CFTs, where everything works as it should \cite{Fjelstad:2005,Kong:2008ci}, and to the $W_{1,p}$-series and the $W_{2,3}$-model \cite{Gaberdiel:2007jv,Gaberdiel:2010rg}, which give modular invariant torus amplitudes and have a self-contragredient space of bulk fields, $F(B) \cong \overline{F(B)}{}^*$. The latter condition is necessary for the existence of a non-degenerate 2-point correlator.
}

\section{Algebraic reformulation} \label{sec:alg-reform}

Some aspects of the consistency conditions for a CFT are analytic in nature, such as the convergence condition (C3) for the OPE and the differential equations (C4) to be satisfied by correlation functions. Other aspects have a combinatorial counterpart which can be described using the language of algebras in braided monoidal categories. In this section we present these counterparts, and we point out the corresponding concepts from section \ref{sec:bb-corr}. 

The translation is made by fixing a vertex operator algebra $\Vc$ as chiral symmetry of the CFT and considering the category $\Rep(\Vc)$ of representations of $\Vc$. This category is by definition $\Cb$-linear and abelian. Under certain conditions on $\Vc$, one obtains in addition a tensor product and a braiding on $\Rep(\Vc)$ \cite{Huang:1994,Huang:2010}.

\medskip

In this section, $k$ denotes a field of characteristic 0. We will use the notation $\Cc(U,V)$ to denote the set of morphisms from an object $U$ to an object $V$ in a category $\Cc$. The categories $\Cc$ we will consider have the following properties:
\begin{itemize}
\item $\Cc$ is $k$-linear, abelian, and satisfies finiteness conditions detailed in section \ref{sec:deligne-prod}.
\item If $\Cc$ is monoidal, the tensor product functor is $k$-linear and right exact in both arguments.\footnote{
For monoidal $\Cc$, we do not require the tensor unit $\one$ to be simple. Neither do we require it to be absolutely simple, that is, we do not impose that the space of endomorphisms of $\one$ is $k \cdot \id_\one$.}
\end{itemize}
For the algebraic constructions presented in this section, it is irrelevant whether $\Cc$ is realised as representations of some vertex operator algebra $\Vc$ or not.

\medskip

We assume some familiarity with abelian categories, exact functors, monoidal categories, monoidal functors, and braidings; the standard reference is \cite{MacLane-book}. Other notions, such as conjugates, the Deligne product and related constructions, and algebras in monoidal categories, are reviewed in sections \ref{sec:conjugate}--\ref{sec:algebra}. The main point of this section is the notion of the ``full centre'', introduced in section \ref{sec:fullcentre}, which is the algebraic implementation of the construction of a bulk theory form a boundary theory described in section \ref{sec:boundary-to-bulk}. Some properties related to the full centre are discussed in section \ref{sec:leftcentre}.

\subsection{Conjugates} \label{sec:conjugate}

In many of the constructions below we will need that every object $U \in \Cc$ has a \textsl{conjugate object} $U^*$. The extra structure we demand to come along with this conjugation is summarised in

\begin{quote}
{\bf Condition (C)}:
The category $\Cc$ is equipped with an involutive contragredient $k$-linear functor $(-)^* : \Cc \to \Cc$, together with a natural family of isomorphisms $\delta_U : U \to U^{**}$ which satisfy $(\delta_U)^* = (\delta_{U^*})^{-1} : U^{***} \to U^{*}$ for all $U \in \Cc$.
Furthermore, $\Cc$ is equipped with a family of isomorphisms $\pi_{U,V} : \Cc(U,V^*) \to \Cc(U \otimes V, \one^*)$, natural in $U$ and $V$.
\end{quote}

We do \textsl{not} demand that there be maps $\ev_U : U^* \otimes U \to \one$ and $\coev_U : \one \to U \otimes U^*$ which satisfy the properties of a categorical dual. Indeed, this property fails in the $W_{2,3}$-example, see section \ref{sec:W23-reps} below. We do also not demand the (weaker) property that $(U\otimes V)^*$ be isomorphic to $V^* \otimes U^*$ (which also fails in the $W_{2,3}$-example).

\begin{remark}
Condition (C) was introduced in \cite[Sect.\,3.1]{Gaberdiel:2009ug} (there, the condition $(\delta_U)^* = (\delta_{U^*})^{-1}$ was not spelled out). It is motivated by the relation of Hom-spaces and spaces of conformal blocks on the sphere in the case $\Cc = \Rep(\Vc)$ for a suitable vertex operator algebra $\Vc$. Then $R^*$ is the contragredient representation to $R$ (see \cite[Notation\,I:2.36]{Huang:2010}) and $\delta_R$ is the natural isomorphism from a graded vector space with finite-dimensional homogeneous components to its graded double dual, which indeed satisfies $(\delta_R)^* = (\delta_{R^*})^{-1}$. Denote by $\otimes_f$ the fusion-tensor product in $\Rep(\Vc)$. The definition of $\pi$ is motivated by the observation that $\Hom_{\Vc}(R \otimes_f S,T)$ is isomorphic to the space of three-point conformal blocks on the Riemann sphere $\Cb \cup \{\infty\}$ with insertions of $R$ and $S$ at $x$ and $y$, say, and of the \textsl{contragredient} $T^*$ of $T$ at a point $z$. Since the position of the insertion points are arbitrary, this space of conformal blocks is also isomorphic to $\Hom_{\Vc}(R \otimes_f T^*,S^*)$. Furthermore, the space of blocks does not change by inserting the vertex operator algebra $\Vc$ itself. Thus, with $S = \Vc$, $\Hom_{\Vc}(R,T^*) \cong \Hom_{\Vc}(R \otimes_f \Vc,T^*) \cong \Hom_{\Vc}(R \otimes_f T,\Vc^*)$. In the setting of \cite{Huang:2010}, the above reasoning amounts to Proposition~II:3.46.
\end{remark}

\begin{definition} \label{def:non-deg-pair}
A pairing $p : U \otimes V \to \one^*$ is called \textsl{non-degenerate} if the map $\pi_{U,V}^{-1}(p) : U \to V^*$ is an isomorphism.
\end{definition}

An alternative characterisation of non-degeneracy of $p$ is that $p \circ (f \otimes \id_V) = 0$ implies $f = 0$ for all $f : X \to U$, and $p \circ (\id_U \otimes g) = 0$ implies $g = 0$ for all $g : Y \to V$ (see \cite[Lem.\,B.7]{Gaberdiel:2009ug} for a proof). This justifies the name `non-degenerate'. There is a canonical non-degenerate pairing
\be \label{eq:def-beta-pairing}
  \beta_U := \pi_{U,U^*}(\delta_U) ~:~ U \otimes U^* \to \one^* \ ,
\ee
which in particular has the property that $\beta_V \circ (h \otimes \id_{V^*}) = \beta_U \circ (\id_U \otimes h^*)$ for all $h : U \to V$, see lemma \ref{lem:conjugate-prop}.

\subsection{Deligne product} \label{sec:deligne-prod}

The point of this subsection is to gain some familiarity with the Deligne product of $k$-linear abelian categories which will be used extensively below. A reader who deems this too technical (or boring) could maybe have a quick glance at definition \ref{def:deligne-prod}, condition (PF) and corollary \ref{cor:proj-gen-in-prod}, and then continue with section \ref{sec:braiding}.

\medskip

Let $A,B$ be two $k$-algebras. Denote by $A\text{-mod}$ and $B\text{-mod}$ the $k$-linear abelian categories of finitely generated modules over these algebras. We can now ask if we can construct $(A\otimes_k B)\text{-mod}$ directly from the categories $A\text{-mod}$ and $B\text{-mod}$ rather than using the algebras $A$ and $B$. The problem one faces is that in general not every $A\otimes_k B$-module is a direct sum of tensor products of $A$-modules and $B$-modules. 

For example, if $A=B=k[x]/\langle x^2 \rangle$, we have $A \otimes_k B \cong k[x,y] / \langle x^2 , y^2 \rangle$. The $A\otimes_k B$-module $M= k[x,y] / \langle x^2, y^2, x{-}y \rangle$ has dimension 2 and both $x$ and $y$ act non-trivially. Since up to dimension two, the only $A$- (or $B$-) module with non-trivial action is $k[x]/\langle x^2 \rangle$, the module $M$ does not arise as a direct sum of tensor products.

\medskip

The passage from $A\text{-mod} \times B\text{-mod}$ (the category of pairs of objects and morphisms) to $(A\otimes_k B)\text{-mod}$ is a special case of the Deligne product of abelian categories \cite[\S\,5.1]{Deligne:2007}. Given two $k$-linear abelian categories $\Ac,\Bc$, denote by $\Fun_{k,\mathrm{r.ex.}}(\Ac,\Bc)$ the category of $k$-linear right exact functors from $\Ac$ to $\Bc$ and natural transformations between them.

\begin{definition} \label{def:deligne-prod}
Let $\{ \Ac_\sigma \}_{\sigma \in S}$ be a family of $k$-linear abelian categories. The \textsl{Deligne product} of the $\{ \Ac_\sigma \}_{\sigma \in S}$ is a pair $(\Ac_S, \boxtimes_S)$, such that
\begin{itemize}
\item[(i)] $\Ac_S$ is a $k$-linear abelian category, and 
 $\boxtimes_S : \prod_{\sigma \in S} \Ac_\sigma \to \Ac_S$ is a functor which is $k$-linear and right exact in each $\Ac_\sigma$,
\item[(ii)] Let $\Bc$ be a $k$-linear abelian category and denote by $\Fun_{\text{mult},\mathrm{r.ex.}}(\prod_{\sigma \in S} \Ac_\sigma,\Bc)$ the category of all functors which are $k$-linear and right exact in each $\Ac_\sigma$. Then for all $\Bc$, the functor
\be\label{eq:defining-equiv}
  (-) \circ \boxtimes_S : \Fun_{k,\mathrm{r.ex.}}(\Ac_S,\Bc) \longrightarrow \Fun_{\text{mult},\mathrm{r.ex.}}\big({\textstyle \prod_{\sigma \in S}} \Ac_\sigma,\Bc \big)
  ~~, \quad
  F \mapsto F \circ \boxtimes_S  \ ,
\ee
is an equivalence of categories.
\end{itemize}
\end{definition}

We will also write the Deligne product as $\boxtimes_{\sigma \in S} \Ac_\sigma$, or, in case there are only a finite number of factors with index set $S = \{1,2,\dots,n\}$, as $\Ac_1 \boxtimes \Ac_2 \boxtimes \cdots \boxtimes \Ac_n$. The triangle one would like to draw for the universal property in condition (ii) is
\be\label{eq:defining-equiv-diag}
\raisebox{2em}{\xymatrix@R=1.5em@C=1.5em{
{\textstyle \prod_{\sigma \in S}} \Ac_\sigma \ar[rr]^{\boxtimes_S}  \ar[dr]_{f}
&& \boxtimes_{\sigma \in S} \Ac_\sigma \ar@{-->}[dl]^{\exists! \, F}
\\
& \Bc
}} \qquad ,
\ee
and it should be read as follows: for each $f \in \Fun_{\text{mult},\mathrm{r.ex.}}(\prod_{\sigma \in S} \Ac_\sigma,\Bc)$ there exists an $F \in \Fun_{k,\mathrm{r.ex.}}(\Ac_S,\Bc)$ such that $f$ is naturally isomorphic to $F \circ \boxtimes_S$. Any other $F'$ with this property is naturally isomorphic to $F$. However, this captures the equivalence of functor categories required in condition (ii) only on the level of objects.

\begin{remark}
In the algebraic reformulation of the construction in section \ref{sec:boundary-to-bulk}, the Deligne product appears as follows. Let $\Vc$ be a suitable vertex operator algebra. The space of boundary fields will be a representation of $\Vc$, in other words, an object in $\Cc = \Rep(\Vc)$. The space of bulk fields will be a representation of $\Vc \otimes_{\Cb} \Vc$, that is, an object in $\Rep(\Vc \otimes_{\Cb} \Vc)$. In the algebraic description, we will replace\footnote{
  We are not aware of a statement in the vertex operator algebra literature that says $\Rep(\Vc \otimes_{\Cb} \Wc) = \Rep(\Vc) \boxtimes \Rep(\Wc)$, but it seems very natural to us that this property should hold, at least for `sufficiently nice' $\Vc$ and $\Wc$, e.g.\ when their representation categories satisfy condition (PF) below.
} 
$\Rep(\Vc \otimes_{\Cb} \Vc)$ by $\Cc \boxtimes \Cc$ (or rather by $\CoC$, were `rev' refers to the inverse braiding, see section \ref{sec:braiding} below).  
\end{remark}

If it exists, the Deligne product is unique up to an equivalence: Let $(\Ac'_S,\boxtimes'_S)$ be another Deligne product and set $\Bc = \Ac'_S$ and $f = \boxtimes'_S$ in the above triangle. This results in a functor $F: \Ac_S \to \Ac'_S$. The converse procedure gives $G : \Ac'_S \to \Ac$ and their compositions have to be equivalent to the identity. To make general existence statements, we will need the following finiteness condition (cf.\,\cite[\S\,2.12.1]{Deligne:2007}):
\begin{quote}
{\bf Condition (F):} The category is $k$-linear and abelian, each object is of finite length\footnote{
   An object $A$ has \textsl{finite length} if there is a chain of subobjects $0 = A_0 \subset A_1 \subset A_2 \subset \cdots \subset A_{n-1} \subset A_n = A$ such that each $S_i = A_{i}/A_{i-1}$ is non-zero and simple. The $S_i$ are called \textsl{composition factors} and $n$ is the \textsl{composition length}.
}, 
and all Hom-spaces are finite-dimensional over $k$.
\end{quote}
By \cite[Prop.\,5.13]{Deligne:2007}, if each $\Ac_\sigma$ satisfies condition (F) then the Deligne product $\Ac_S \equiv \boxtimes_{\sigma \in S} \Ac_\sigma$ exists and equally satisfies condition (F); for each $X_\sigma,Y_\sigma \in \Ac_\sigma$, the functor $\boxtimes_S$ gives an isomorphism
\be \label{eq:hom-iso}
  \textstyle \bigotimes_{k\,,\,\sigma \in S} \Ac_\sigma(X_\sigma,Y_\sigma) \xrightarrow{~\sim~}  
  \Ac_S\big( \boxtimes_{\sigma \in S} X_\sigma \,,\, \boxtimes_{\sigma \in S} Y_\sigma \,\big) \ .
\ee
A stronger condition than (F) is
\begin{quote}
{\bf Condition (PF):} The category is $k$-linear and abelian, and it has a projective generator $P$ whose endomorphism space is finite-dimensional over $k$.
\end{quote}
That $P$ is a projective generator means that $P$ is projective and for every $U \in \Ac$ there is an $m \in \Nb$ and a surjection $P^{\oplus m} \to U$, i.e.\ every object in $\Ac$ is a quotient of some $m$-fold direct sum of $P$'s. Since $\Ac(P,P)$ is finite-dimensional, so are all other morphism spaces in $\Ac$ (pick projective resolutions). Since $P$ has finite composition series (or $\Ac(P,P)$ would have infinite dimension since $P$ would have non-zero maps to every subobject in the descending chain), so does every object in $\Ac$. Thus (PF)$\Rightarrow$(F). Categories satisfying (PF) have the following convenient description:

\begin{theorem}[{\cite[Cor.\,2.17]{Deligne:2007}}] \label{eq:PF-means-fd-algebra}
$\Ac$ satisfies condition (PF) if and only if there exists a unital finite-dimensional $k$-algebra $A$ such that $\Ac$ is equivalent, as a $k$-linear category, to the category $\Rep_{f.d.}(A)$ of finite dimensional (over $k$) right $A$-modules.
\end{theorem}

The proof is maybe instructive to gain some intuition for the finiteness condition (PF); for the convenience of the reader we include it in appendix \ref{app:more-deligne}. The next theorem confirms the motivation for studying Deligne products which was stated in the beginning of this subsection. It is proved (in greater generality) in \cite[Prop.\,5.3]{Deligne:2007}; we sketch a proof in our simpler situation.

\begin{theorem} \label{thm:deligne-product-of-alg}
Let $A,B$ be finite-dimensional unital $k$-algebras. Then 
\be
  \Rep_{f.d.}(A) \, \boxtimes \, \Rep_{f.d.}(B) ~=~ \Rep_{f.d.}(A \otimes_k B) \ .
\ee
\end{theorem}

\begin{proof}[Sketch of proof] Write $\Ac = \Rep_{f.d.}(A)$, $\Bc = \Rep_{f.d.}(B)$, $\Dc = \Rep_{f.d.}(A \otimes_k B)$. The functor $\boxtimes : \Ac \times \Bc \to \Dc$ is $(M,N) \mapsto M \otimes_k N$, seen as an $A \otimes_k B$ right module, and $(f,g) \mapsto f \otimes_k g$ for module maps $f,g$. (Since $k$ is a field, $\boxtimes$ is actually \textsl{exact} in each argument, not only right exact, cf.\,\cite[Cor.\,5.4]{Deligne:2007}.)

Let $\Ec$ be a $k$-linear abelian category. We need to show that $(-) \circ \boxtimes$ gives an equivalence of functor categories $\Fun_{k,\mathrm{r.ex.}}(\Dc,\Ec) \to \Fun_{\text{mult},\mathrm{r.ex.}}(\Ac \times \Bc,\Ec)$, see \eqref{eq:defining-equiv}. The point is that a $k$-linear, right exact functor $F : \Dc \to \Ec$ is fixed by $F( A \otimes_k B )$, and by $F(f)$ for all right module endomorphisms of $A \otimes_k B$. To see this, just express an arbitrary finite-dimensional $A \otimes_k B$ right module $M$ via the first two terms in a free resolution, $(A \otimes_k B)^{\oplus n} \to(A \otimes_k B)^{\oplus m} \to M\to 0$ for appropriate $m,n \in \Zb_{\ge 0}$. Similarly, a functor $G : \Ac \times \Bc \to \Ec$ which is $k$-linear and right exact in each argument is fixed by $G(A,B)$ and $G(f,g)$ for all right module endomorphisms $f$ of $A$ and $g$ of $B$. From this one derives that $(-) \circ \boxtimes$ is essentially surjective. Natural transformations are equally determined by evaluating them on $A \otimes_k B$, respectively on $(A,B)$, and from this one can deduce that $(-) \circ \boxtimes$ is full and faithful.
\end{proof}

\begin{corollary} \label{cor:proj-gen-in-prod}
If $\Ac$ and $\Bc$ satisfy property (PF), then so does $\Ac \boxtimes \Bc$. If $P$ and $Q$ are projective generators of $\Ac$ and $\Bc$, respectively, then $P \boxtimes Q$ is a projective generator of $\Ac \boxtimes \Bc$.
\end{corollary}

\begin{proof}
By the explicit construction in appendix \ref{app:more-deligne} we have $\Ac \cong  \Rep_{f.d.}(A)$ as $k$-linear abelian categories for the choice  $A = \Ac(P,P)$, and also $\Bc \cong  \Rep_{f.d.}(B)$ for $B = \Bc(Q,Q)$. Then by theorem \ref{thm:deligne-product-of-alg} we may take $\Ac \boxtimes \Bc \equiv \Rep_{f.d.}(A \otimes_k B)$. With this choice, $P \boxtimes Q = A \otimes_k B$, which indeed is a projective generator.
\end{proof}

Natural transformations of right exact functors whose domain is a Deligne product are determined by their action on ``product objects''. We will use this a number of times, so let us give a short proof (the statement holds for $ \boxtimes_{\sigma \in S} \Ac_\sigma$, but for notational simplicity we only give the case with two factors).

\begin{lemma} \label{lem:nat-xfer-unique-on-prod}
Let $\Ac,\Bc$ satisfy property (F). Let $\Cc$ be a $k$-linear abelian category, let $F,G \in \Fun_{k,\text{r.ex.}}(\Ac \boxtimes \Bc,\Cc)$ and let $\alpha,\beta : F \Rightarrow G$ be natural transformations. The following are equivalent:
\begin{itemize}
\item[(i)] $\alpha_X = \beta_X$ for all $X \in \Ac \boxtimes \Bc$,
\item[(ii)] $\alpha_{A \boxtimes B} = \beta_{A \boxtimes B}$ for all $A \in \Ac$, $B \in \Bc$.
\end{itemize}
\end{lemma}

\begin{proof}
We need to check (ii)$\Rightarrow$(i). Write $\hat F = F \circ \boxtimes$ and $\hat G = G \circ \boxtimes$. The functor $(-) \circ \boxtimes$ maps natural transformations $F \Rightarrow G$ to natural transformations $\hat F \Rightarrow \hat G$ via 
\be \label{eq:nat-xfer-unique-on-prod-aux1}
  \{ \eta_X \}_{X \in \Ac \boxtimes \Bc} \longmapsto \{ \eta_{A,B} \}_{(A,B) \in \Ac \times \Bc} \qquad , \quad \text{where}~~ \eta_{A,B} := \eta_{A\boxtimes B} \ .
\ee
By condition (ii) in definition \ref{def:deligne-prod}, the map \eqref{eq:nat-xfer-unique-on-prod-aux1} is an isomorphism and hence $\beta$ is uniquely determined by its values on all $A \boxtimes B$.
\end{proof}

We will be interested in the case that a category $\Cc$ satisfies property (F) and is in addition monoidal with $k$-linear right exact tensor product. Then the tensor product $\otimes_{\Cc}:\Cc \times \Cc \to \Cc$ gives us a right exact functor
\be \label{eq:TC-construction}
  T_{\Cc} : \Cc \boxtimes \Cc \longrightarrow \Cc \ ,
\ee
such that $A \otimes_{\Cc} B = T_{\Cc}(A \boxtimes B)$ and analogously for morphisms. Let now $\Dc$ be another such category. Then $\Cc \boxtimes \Dc$ is monoidal with right exact tensor product given by
\be
   \otimes_{\Cc \boxtimes \Dc} = \Big[ (\Cc \boxtimes \Dc) \times (\Cc \boxtimes \Dc) \xrightarrow{\boxtimes}
   \Cc \boxtimes \Dc \boxtimes \Cc \boxtimes \Dc
  \xrightarrow{\sim}
   \Cc \boxtimes \Cc \boxtimes \Dc \boxtimes \Dc
  \xrightarrow{T_{\Cc} \boxtimes T_{\Dc}}
   \Cc \boxtimes \Dc \Big] \ ,
\ee 
for details see \cite[Sect.\,5.16--17]{Deligne:2007}. The unnamed isomorphism is induced by the functor $\Cc \times \Dc \times \Cc \times \Dc \to \Cc \times \Cc \times \Dc \times \Dc$ which exchanges the middle two factors. In particular, for $A,B \in \Cc$ and $U,V \in \Dc$,
\be \label{eq:CxC-tensor-def}
  (A \boxtimes U) \,\otimes_{\Cc \boxtimes \Dc} (B \boxtimes V) ~=~  (A \otimes_{\Cc} B) \, \boxtimes \, (U \otimes_{\Dc} V) \ .
\ee

\subsection{Braiding} \label{sec:braiding}

For this subsection we fix a braided monoidal $k$-linear abelian category $\Cc$ which satisfies property (F), and which has a $k$-linear right exact tensor product. In the previous subsection we saw that $\Cc \boxtimes \Cc$ is again monoidal with right exact tensor product. We will use the braiding on $\Cc$ for three related constructions:
\begin{itemize}
\item turn the functor $T_{\Cc} : \Cc \boxtimes \Cc \to \Cc$ from above into a tensor functor,
\item equip the category $\Cc \boxtimes \Cc$ with a braiding, 
\item define a `mixed braiding' with one object from $\Cc \boxtimes \Cc$ and one object from $\Cc$.
\end{itemize}

Let us start with the monoidal structure on $T \equiv T_{\Cc}$. The tensor product of $\Cc \boxtimes \Cc$ will be denoted by $\otimes_{\Cc^2}$.
We have to give isomorphisms
\be \label{eq:T2T0-def}
  T_{2;X,Y} : T(X) \otimes_{\Cc} T(Y) \to T(X \otimes_{\Cc^2} Y)
  \quad , \qquad 
  T_0 : \one \to T(\one \boxtimes \one) \equiv \one \otimes_{\Cc} \one \ ,
\ee
where $T_{2;X,Y}$ is natural in $X,Y \in \Cc\boxtimes\Cc$. $T_2$ and $T_0$ are required to satisfy the hexagon and triangle identity (given explicitly in \eqref{eq:lax-monoidal-hexagon} and \eqref{eq:lax-monoidal-triangle} below for a lax monoidal functor). For $T_0$ one takes the inverse unit isomorphism of $\Cc$. For $T_2$, consider first the two functors from $\Cc^{\times 4}$ to $\Cc$ given by
\be
\begin{array}{l}
  (A,B,U,V) ~\mapsto~ T(A \boxtimes B) \otimes_{\Cc} T(U \boxtimes V) ~\equiv~ (A\otimes_{\Cc}B)\otimes_{\Cc}(U\otimes_{\Cc}V)
  ~~\text{and}
  \\[.4em]
  (A,B,U,V) ~\mapsto~ T\big(\,(A\boxtimes B) \otimes_{\Cc^2} (U\boxtimes V)\,\big) ~\equiv~ (A\otimes_{\Cc}U)\otimes_{\Cc}(B\otimes_{\Cc}V) \ .
\end{array}
\ee
These are linked by the natural isomorphism (not writing out $\otimes_{\Cc}$ between objects)\footnote{
  The convention to use $c^{-1}$ and not $c$ for $T_2$ agrees with 
  \cite[Sect.\,2.4]{Kong:2008ci}
  but it is opposite to 
  \cite[Sect.\,7]{Davydov:2009}. This should be taken into account when referring to proofs in \cite{Davydov:2009}. We use the $c^{-1}$ convention to make lemma \ref{lem:varphi-other-def} true in the form given below. In the context of CFT, the inverse braiding convention means that in the graphical notation ``lines corresponding to holomorphic insertions go on top''.}
\be \label{eq:T-tensor-structure}
  \tilde T_{2;(A,B),(U,V)} := 
  \Big[ (AB)(UV) \xrightarrow{\text{assoc.}} (A(BU))V \xrightarrow{(\id_A \otimes c_{U,B}^{-1}) \otimes \id_V} (A(UB))V
  \xrightarrow{\text{assoc.}} (AU)(BV)  \Big] \ .
\ee
The defining isomorphism of the Deligne product between functor categories transports $ \tilde T_2$ to the desired natural isomorphism $T_2$ in \eqref{eq:T2T0-def}. In particular, $T_2$ obeys $T_{2;U \boxtimes V,A \boxtimes B} =  \tilde T_{2;(U,V),(A,B)}$. The hexagon identity for $T_2$ follows if it holds on product objects (lemma \ref{lem:nat-xfer-unique-on-prod}), and for these it reduces to the hexagon of the braiding $c$ of $\Cc$, cf.\,\cite[Prop.\,5.2]{Joyal:1993}.

\medskip

Next we turn to the braiding on $\Cc \boxtimes \Cc$ that we wish to use. This will again be defined by transporting a natural isomorphism, this time between two functors $\Cc^{\times 4} \to \Cc\boxtimes \Cc$
\be
\begin{array}{l}
  (A,B,U,V) ~\mapsto~ (A \boxtimes B) \otimes_{\Cc^2} (U \boxtimes V) ~\equiv~ (A\otimes_{\Cc}U)\boxtimes (B\otimes_{\Cc}V)
  ~~\text{and}
  \\[.4em]
  (A,B,U,V) ~\mapsto~ (U \boxtimes V) \otimes_{\Cc^2} (A \boxtimes B) ~\equiv~ (U\otimes_{\Cc}A)\boxtimes(V\otimes_{\Cc}B) \ .
\end{array}
\ee
The natural isomorphism we choose is $\tilde c_{(A,B),(U,V)} = c_{A,U} \boxtimes c_{V,B}^{-1}$. The defining property of the Deligne product provides a natural isomorphism $c_{X,Y} : X \otimes_{\Cc^2} Y \to Y \otimes_{\Cc^2} X$ which satisfies
\be \label{eq:CxC-braiding-def}
  c_{A \boxtimes B,U \boxtimes V} =
  \Big[ (A \boxtimes B) \otimes_{\Cc^2} (U \boxtimes V) \xrightarrow{~c_{A,U} \,\boxtimes \,c_{V,B}^{-1}~}
  (U \boxtimes V) \otimes_{\Cc^2} (A \boxtimes B) \Big] \ .
\ee
One verifies that the hexagon condition for the braiding on $\Cc$ implies the hexagon for $c$ in $\Cc \boxtimes \Cc$ on product objects; by lemma \ref{lem:nat-xfer-unique-on-prod} it then holds on all of $\Cc \boxtimes \Cc$. We will denote the category $\Cc \boxtimes \Cc$ with tensor product \eqref{eq:CxC-tensor-def} and braiding \eqref{eq:CxC-braiding-def} by
\be
  \CoC \ .
\ee

Finally, we turn to the mixed braiding between $\CoC$ and $\Cc$. The relevant functors $\Cc^{\times 3} \to \Cc$ are $\tilde L(A,B,U) = (A \otimes_{\Cc} B) \otimes_{\Cc} U$ and $\tilde R(A,B,U) = U \otimes_{\Cc} (A \otimes_{\Cc} B)$. Between these we have the natural isomorphism $\tilde \varphi_{A,B,U} : \tilde L \Rightarrow \tilde R$ given by the string diagram (to be read the optimistic way, i.e.\ upwards from bottom to top)
\be \label{eq:mixed-braid-tildephi-braid}
\tilde \varphi_{A,B,U} ~=~ 
  \raisebox{-20pt}{
  \begin{picture}(35,51)
   \put(0,0){\scalebox{1.00}{\includegraphics{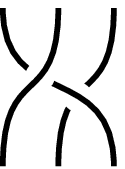}}}
   \put(0,0){
     \setlength{\unitlength}{1.00pt}\put(-6,-5){
     \put(3,-3) {\scriptsize$ A $}
     \put(19,-3) {\scriptsize$ B $}
     \put(34,-3) {\scriptsize$ U $}
     \put(3,55) {\scriptsize$ U $}
     \put(19,55) {\scriptsize$ A $}
     \put(34,55) {\scriptsize$ B $}
     }\setlength{\unitlength}{1pt}}
  \end{picture}}
  \quad .
\ee
In terms of formulas, this translates as\footnote{
  Our convention for associators is $\alpha_{X,Y,Z} : X \otimes (Y \otimes Z) \to (X \otimes Y) \otimes Z$.
  }
\be \label{eq:mixed-braid-tildephi}
  \tilde \varphi_{A,B,U}
  = 
  \alpha_{U,A,B}^{-1}
  \circ
  (c_{A,U} \otimes_{\Cc} \id_B)
  \circ
  \alpha_{A,U,B}
  \circ
  (\id_A \otimes_{\Cc} c_{U,B}^{-1})
  \circ
  \alpha_{A,B,U}^{-1} \ .
\ee
From the Deligne product, we obtain a natural isomorphism $\varphi$ between $L,R : \Cc \boxtimes \Cc \boxtimes \Cc \to \Cc$ such that $\varphi_{A \boxtimes B \boxtimes U} = \tilde \varphi_{A,B,U}$. We will most often use $\varphi$ in the form
\be \label{eq:phi-mixed-braid}
  \varphi_{X,U} := \varphi_{T(X) \boxtimes U} : T(X) \otimes_{\Cc} U \longrightarrow U \otimes_{\Cc} T(X) 
  \quad ; ~ X \in \CoC ~,~ Y \in \Cc \ .
\ee
There is an alternative way to define $\varphi_{X,U}$ by transporting the braiding from $\CoC$ to $\Cc$ with $T$. By the next lemma, these two possibilities give the same result.

\begin{lemma} \label{lem:varphi-other-def}
For $X \in \CoC$ and $U \in \Cc$, the following diagram commutes.
\be \label{eq:varphi-alternate-def}
\raisebox{2em}{\xymatrix{
T(X) \otimes_{\Cc} U \ar[r]^-{\sim} \ar[d]^{\varphi_{X,U}} 
& T(X) \otimes_{\Cc} T(U \boxtimes \one) \ar[r]^{T_2} 
& T(X \otimes_{\Cc^2} (U \boxtimes \one)) \ar[d]_{T(c_{X,U \boxtimes \one})}
\\
U \otimes_{\Cc} T(X) \ar[r]^-{\sim} 
& T(U \boxtimes \one) \otimes_{\Cc} T(X)\ar[r]^{T_2}
& T((U \boxtimes \one) \otimes_{\Cc^2} X)
}}
\ee
\end{lemma}

\begin{proof}
By lemma \ref{lem:nat-xfer-unique-on-prod} it is enough to verify commutativity of the diagram on product objects $X = A \boxtimes B$. Drawing the corresponding string diagrams using \eqref{eq:T-tensor-structure}, \eqref{eq:CxC-braiding-def} and \eqref{eq:mixed-braid-tildephi} one finds the string diagram \eqref{eq:mixed-braid-tildephi-braid} for both paths.
\end{proof}

With the help of the above lemma, it is easy to use identities for the braiding on $\CoC$ to obtain identities for $\varphi$. We will need
\begin{align}\label{eq:phi-properties}
 \varphi_{X \otimes_{\Cc^2} Y,U} = 
&\Big[ T(XY)U \xrightarrow{T_2^{-1} \otimes \id_U} (TX\,TY)U 
\xrightarrow{\sim} TX(TY\,U) \xrightarrow{\id_{TX} \otimes \varphi_{Y,U}} 
TX(U\,TY)\\\nonumber
&\ \xrightarrow{\sim} (TX\,U)TY \xrightarrow{\varphi_{X,U} \otimes \id_{TY}} 
(U\,TX)TY \xrightarrow{\sim} U(TX\,TY) \xrightarrow{\id_U \otimes T_2} 
U\,T(XY) \Big] \ ,
\end{align}
which follows form applying $T$ to the hexagon identity $c_{X \otimes_{\Cc^2} Y, U \boxtimes \one} = ( c_{X , U \boxtimes \one} \otimes_{\Cc^2} \id_Y ) \circ (\id_X \otimes_{\Cc^2} c_{Y , U \boxtimes \one})$ (we have omitted the associators) and rearranging terms via lemma \ref{lem:varphi-other-def}.

Instead of $\varphi_{X,U}$, which takes one argument from $\CoC$ and one from $\Cc$, we can use the diagram \eqref{eq:varphi-alternate-def} to define $\hat \varphi_{X,Y}$, which takes both arguments from $\CoC$ via
\be \label{eq:varphi-hat-def}
\raisebox{2em}{\xymatrix{
T(X) \otimes_{\Cc} T(Y) \ar[r]^{T_2}  \ar@{-->}[d]_{\hat\varphi_{X,Y}}
& T(X \otimes_{\Cc^2} Y) \ar[d]^{T(c_{X,Y})}
\\
T(Y) \otimes_{\Cc} T(X)\ar[r]^{T_2}
& T(Y \otimes_{\Cc^2} X)
}}
\ee
The following observation will be important below.

\begin{lemma} \label{lem:phi-hatphi-rel}
For $X,Y \in \CoC$ we have the identity
\be
  \varphi_{X,T(Y)} = \hat \varphi_{X,Y} ~:~ T(X) \otimes_{\Cc} T(Y) \to T(Y) \otimes_{\Cc} T(X) \ .
\ee
\end{lemma}

\begin{proof}
From \eqref{eq:varphi-alternate-def} and \eqref{eq:varphi-hat-def} we see that we have to establish commutativity of
\be
\raisebox{2em}{\xymatrix{
T(X \otimes_{\Cc^2} Y) \ar[r]^{T_2^{-1}} \ar[d]^{T(c_{X,Y})} 
& T(X) \otimes_{\Cc} T(Y) \ar[r]^-{\sim}  
& T(X) \otimes_{\Cc} T(T(Y) \boxtimes \one) \ar[r]^{T_2} 
& T(X \otimes_{\Cc^2} (T(Y) \boxtimes \one)) \ar[d]_{T(c_{X,T(Y) \boxtimes \one})}
\\
T(Y \otimes_{\Cc^2} X) \ar[r]^{T_2^{-1}} 
&
T(Y) \otimes_{\Cc} T(X) \ar[r]^-{\sim} 
& T(T(Y) \boxtimes \one) \otimes_{\Cc} T(X)\ar[r]^{T_2}
& T((T(Y) \boxtimes \one) \otimes_{\Cc^2} X)
}}
\ee
By lemma \ref{lem:nat-xfer-unique-on-prod}, it is enough to verify this for $X = A \boxtimes B$ and $Y = U \boxtimes V$ for all $A,B,U,V \in \Cc$. In this case, the above diagram reads (not writing $\otimes_\Cc$, brackets between objects, and associators)
\be
\raisebox{2em}{\xymatrix{
AUBV \ar[rr]^{\id_A \otimes c_{U,B} \otimes \id_V} \ar[d]_{c_{A,U} \otimes c_{V,B}^{-1}} 
&& ABUV \ar[r]^{\sim}  
& ABUV\one \ar[rr]^{\id_A \otimes c_{UV,B}^{-1} \otimes \id_\one} 
&& AUVB\one \ar[d]^{c_{A,UV} \otimes c_{\one,B}^{-1}}
\\
UAVB \ar[rr]^{\id_U \otimes c_{A,V} \otimes \id_B} 
&& UVAB \ar[r]^{\sim} 
& UV1AB \ar[rr]^{\id_{UV} \otimes c_{\one,A}^{-1} \otimes \id_B}
&& UVA1B
}}
\quad .
\ee
That this diagram commutes can be checked easily by drawing string diagrams.
\end{proof}

\begin{remark} \label{rem:T-central}
The functor $T$ is \textsl{central} in the sense of \cite[Sect.\,2]{Davydov:2010a}. Namely, it factors through the braided tensor functor $G$ from $\CoC$ to the monoidal centre of $\Cc$ as $T_{\Cc} = \big[ \CoC \xrightarrow{G} \mathcal{Z}(\Cc) \xrightarrow{\text{forget}} \Cc\big]$; we refer to \cite[Sect.\,2]{Davydov:2010a} for details. 
\end{remark}

\subsection{Algebras} \label{sec:algebra}

We recall the definition of algebras in monoidal categories, and of commutative algebras in braided monoidal categories. In the category of vector spaces, these give the usual notions of algebras /  commutative algebras.

\begin{definition}
(i) An \textsl{algebra} in a monoidal category $\Cc$ is an object $A \in \Cc$ together with a morphism $\mu : A \otimes A \to A$ which is associative in the sense that
\be
 \raisebox{2em}{\xymatrix@R=.5em{
 A \otimes (A \otimes A) \ar[r]^-{\id_A \otimes \mu} \ar[dd]_{\alpha_{A,A,A}}
 &
 A \otimes A
 \ar[dr]^\mu
 \\
 && A
 \\
 (A \otimes A) \otimes A \ar[r]^-{\mu \otimes \id_A}
 &
 A \otimes A \ar[ur]_\mu
}}
\ee
commutes. $A$ is called \textsl{unital} if it is equipped with a morphism $\iota : \one \to A$ such that \be
 \raisebox{2em}{\xymatrix{
 \one \otimes A \ar[r]^{\iota \otimes \id_A} \ar[dr]_{\lambda_A}
 & A \otimes A \ar[d]^\mu
 & A \otimes \one \ar[l]_{\id_A \otimes \iota} \ar[dl]^{\rho_A}
 \\
 & A
}} 
\ee
commutes. Here $\alpha$ is the associator of $\Cc$ and $\lambda$, $\rho$ are the unit isomorphisms. An \textsl{algebra homomorphism} from $(A,\mu)$ to $(A',\mu')$ is a morphism $f : A \to A'$ such that $f \circ \mu = \mu' \circ (f \otimes f)$. If $A$ and $A'$ are unital, $f$ is called \textsl{unital} if $f \circ \iota = \iota'$.
\\[.3em]
(ii) An algebra in a braided monoidal category is called \textsl{commutative} if $\mu \circ c_{A,A} = \mu$. 
\end{definition}

The tensor unit $\one \in \Cc$ with multiplication $\mu = \lambda_I = \rho_I$ and unit $\iota = \id_\one$ is always a commutative unital algebra. A similar class of examples are objects $S \in \Cc$ such that $\Cc(S,S) = k \cdot \id_S$ and $S \otimes_{\Cc} S \cong S$. Each isomorphism $S \otimes_{\Cc} S \to S$ is a commutative associative multiplication on $S$ (not necessarily unital), and of course all these multiplications give isomorphic algebras, see appendix \ref{app:idempot-alg}. In the $W_{2,3}$-model treated in section \ref{sec:W23-model}, this will give three examples of algebras (namely the representations $\Wc$, $\Wc^*$ and $\Wc(0)$, see section \ref{sec:W23-model} for details).

Suppose $\Cc$ has property (C). By a \textsl{pairing} on an algebra $A$ in $\Cc$ we mean a morphism $\pi : A \otimes A \to \one^*$. The pairing is called \textsl{invariant} if 
\be
  \pi \circ (\id_A \otimes \mu) = \pi \circ (\mu \otimes \id_A) \circ \alpha_{A,A,A} \ .
\ee
If $A$ is unital, giving an invariant pairing is the same as giving a morphism
   $\tau : A \to \one^*$ via $\pi = \tau \circ \mu$.
The notion of non-degeneracy of a paring on $A$ is that of definition \ref{def:non-deg-pair}. 

\medskip

A brief comparison between these algebraic notions and the discussion of CFT in section \ref{sec:bb-corr} is given in table \ref{tab:alg-cft-table1}.

\begin{table}[bt]
\begin{center}
\begin{tabular}{p{16em}|p{21em}}
Conformal field theory & Algebraic counterpart
\\
\hline
\\[-.5em]
$\Rep\Vc$, for a vertex operator algebra $\Vc$ which is `logarithmic-rational': the tensor product theory of \cite{Huang:2010} should apply and it should only have a finite number of irreducible sectors.
&
A braided monoidal category $\Cc$ which is
$\Cb$-linear, abelian, with right exact tensor product, and which satisfies the finiteness
condition (PF) and has conjugates in the
sense of condition (C).
\\
\\
$(B,m,\omega^*)$, a non-degenerate
boundary theory as in definitions 
\ref{def:theo-on-bnd} and \ref{def:non-deg-bnd}.
&
An algebra $B \in \Cc$ with associative product $m : B \otimes_{\Cc} B \to B$ and a
map $\omega^* : B \to \one^*$ such that the pairing $\omega^* \circ m$ on $B$ is
non-degenerate.
\\
\\
$(F,M,\Omega^*)$, a non-degenerate CFT on $\Cb$
as in definitions \ref{def:CFT-plane} and \ref{def:non-deg-CFT}.
&
An algebra $F \in \CoC$ with associative, commutative product $M : F \otimes_{\Cc^2} F \to F$ and a map $\Omega^* : F \to \one^* \boxtimes \one^*$ such that the pairing $\Omega^* \circ M$ is
non-degenerate.
\\
\\
$(F,M)$, a CFT on $\Cb$ with background states as defined in section \ref{sec:genout-ideal}.
&
An algebra $F \in \CoC$ with associative, commutative product $M : F \otimes_{\Cc^2} F \to F$.
\end{tabular}
\end{center}
\caption{Relation between the algebraic notions of this section and the discussion of CFT in section \ref{sec:bb-corr}. These relations have been proved for non-logarithmic rational CFTs (see \cite{Fuchs:2001am,Fjelstad:2005,Huang:2005,Kong2006b}). In general the table should be understood as `similarity in structure'. This table continues after some preparation with table \ref{tab:alg-cft-table2} below.}
\label{tab:alg-cft-table1}
\end{table}

\medskip

It is not surprising that a monoidal functor between two monoidal categories transports algebras to algebras. However, also the weaker notion of a lax monoidal functor is sufficient for this purpose.

\begin{definition}\label{def:lax-monoidal-def}
Let $\Ac$ and $\Bc$ be two monoidal categories and let $F : \Ac \to \Bc$ be a functor. Then $F$ is called \textsl{lax monoidal} if it is equipped 
with morphisms $F_0 : \one_{\Bc} \to F(\one_{\Ac})$ and $F_{2;U,V} : F(U) \otimes_{\Bc} F(V) \to F(U \otimes_{\Ac} V)$, the latter natural in $U,V$, such that for all $U,V,W \in \Ac$,
\be\label{eq:lax-monoidal-hexagon}
\raisebox{4em}{\xymatrix{
F(U)\otimes_{\Bc} \big(F(V)\otimes_{\Bc} F(W)\big) \ar[rr]^{\alpha^{\Bc}_{FU,FV,FW}}
\ar[d]_{\id_{FU}\otimes F_{2;V,W}}  
&&   \big(F(U)\otimes_{\Bc} F(V)\big)\otimes_{\Bc} F(W) \ar[d]^{F_{2;U,V} \otimes \id_{FW} }  
\\
F(U)\otimes_{\Bc} F(V\otimes_{\Ac} W) \ar[d]_{F_{2;U,VW}} 
&& F(U\otimes_{\Ac} V)\otimes_{\Bc} F(W) \ar[d]^{F_{2;UV,W}} 
\\
F(U\otimes_{\Ac} (V\otimes_{\Ac} W)) \ar[rr]^{F(\alpha^{\Ac}_{U,V,W})}  
&& 
F((U\otimes_{\Ac} V)\otimes_{\Ac} W)}}   
\quad ,
\ee
and
\be\label{eq:lax-monoidal-triangle}
\raisebox{2em}{\xymatrix{
\one_{\Bc} \otimes_{\Bc} F(U) \ar[d]_{F_0 \otimes \id_{FU}}
\ar[r]^{\lambda^{\Bc}_{FU}} & F(U)  \\
F(\one_{\Ac})\otimes_{\Bc} F(U) \ar[r]^-{F_{2;\one,U}} &  
F(\one_{\Ac} \otimes_{\Ac} U) \ar[u]_{F(\lambda^{\Ac}_U)} }}
\quad , \quad
\raisebox{2em}{\xymatrix{
F(U)\otimes_{\Bc} \one_{\Bc} \ar[d]_{ \id_{FU}\otimes F_0}
\ar[r]^{\rho^{\Bc}_{FU}} & F(U) \\
F(U)\otimes_{\Bc} F(\one_{\Ac}) \ar[r]^-{F_{2;U,\one}} &
F(U\otimes_{\Ac} \one_{\Ac})\ar[u]_{F(\rho^{\Ac}_U)} }}  
\ee
commute. If $F_0$ and $F_2$ are isomorphisms, $F$ is called {\em strong monoidal} (or just {\em monoidal}).
\end{definition}

Let $\Ac$, $\Bc$ be monoidal categories and let $F : \Ac \to \Bc$ be a lax monoidal functor. If $(A,\mu)$ is an algebra in $\Ac$, then the image object $F(A)$ also carries the structure of an algebra, with associative multiplication given by
\be \label{eq:F(A)-product}
\mu_{F(A)} = \Big[
F(A) \otimes_{\Bc} F(A) \xrightarrow{F_{2;A,A}} F(A \otimes_{\Ac} A) \xrightarrow{F(\mu)} F(A) \Big] \ .
\ee
If $A$ is unital with unit $\iota$, then so is $F(A)$ with unit $F(\iota) \circ F_0$. If $f : A \to B$ is a homomorphism of algebras in $\Ac$, then $F(f) : F(A) \to F(B)$ is a homomorphism of algebras in $\Bc$. See \cite[Sect.\,5]{Joyal:1993} for details.

\subsection{The full centre in $\CoC$} \label{sec:fullcentre}

In this subsection, $\Cc$ is assumed to be a braided monoidal $k$-linear abelian category with conjugates as in (C), which satisfies the finiteness condition (PF), and which has a $k$-linear right exact tensor product functor. The assumptions (PF) and (C) will guarantee existence of the full centre of an algebra in $\Cc$, to be defined now (though much weaker conditions should be sufficient, too). Recall the definition of the functor $T : \CoC \to \Cc$ from sections \ref{sec:deligne-prod} and \ref{sec:braiding}, as well as the mixed braiding $\varphi_{X,A}$ from \eqref{eq:phi-mixed-braid}.

\begin{definition} \label{eq:full-centre-in-C2}
Let $(A,\mu_A)$ be an algebra in $\Cc$. The \textsl{full centre in} $\CoC$ is an object $Z(A) \in \CoC$ together with a morphism $z : T(Z(A)) \to A$ in $\Cc$ such that the following universal property holds: For all pairs $(X,x)$ with $X \in \CoC$ and $x : T(X) \to A$ such that the diagram
\be \label{eq:full-cent-univprop}
\raisebox{2.4em}{\xymatrix@R=.7em@C=2.5em{
T(X) \otimes_\Cc A \ar[dd]^{\varphi_{X,A}} \ar[rr]^{x \otimes_\Cc \id_A} && A \otimes_\Cc A \ar[dr]^{\mu_A}
\\
&&& A
\\
A \otimes_\Cc T(X)  \ar[rr]^{\id_A \otimes_\Cc x} && A \otimes_\Cc A \ar[ur]_{\mu_A}
}}
\ee
in $\Cc$ commutes, there exists a unique morphism $\zeta_{(X,x)} : X \to Z(A)$ such that
\be \label{eq:full-cent-pairmap}
\raisebox{2em}{\xymatrix@R=1.5em@C=1.5em{
T(X) \ar[rr]^-{T(\zeta_{(X,x)})} \ar[dr]_{x} && T(Z(A)) \ar[dl]^z
\\
& A
}}
\ee
commutes.
\end{definition}

The existence of the full centre will be proved in theorem \ref{thm:full-centre-via-left-centre} below.

\begin{table}[bt]
\begin{center}
\begin{tabular}{p{18em}|p{19em}}
Conformal field theory & Algebraic counterpart
\\
\hline
\\[-.5em]
$b'$, a `candidate bulk-boundary map' from a `candidate space of bulk fields' $F'$ to the space of boundary fields $B$ of a boundary theory, satisfying in particular the centrality condition from section \ref{sec:boundary-to-bulk}.
&
An object $F' \in \CoC$ and an algebra $B \in \Cc$ together with a morphism $b' : T(F') \to B$ such that $b' \in \Cent(F',B)$.
\\
\\
$\Pc$, the category of pairs from section \ref{sec:boundary-to-bulk} and the terminal object $(F(B),b(B))$ in it from \eqref{eq:terminal-obj-in-P}, interpreted as the maximal bulk theory compatible with the given boundary theory $B$.
&
The category $\Cc_{\text{full\,center}}(B)$ for a given algebra $B \in \Cc$ and the terminal object $(Z,z)$ in it, where $Z \in \CoC$ is the full centre of $B$ and $z$ the corresponding map $T(Z) \to B$.
\\
\\
$(F,M,\Omega^*;B,m,\omega^*;b)$, a CFT on the upper half plane as in definition \ref{def:CFT-uhp}, for which the CFT on $\Cb$ and the boundary theory are non-degenerate.
&
A commutative algebra $(F,M)$ in $\CoC$ with non-degenerate pairing $\Omega^* \circ M$, a not necessarily commutative algebra $(B,m)$ in $\Cc$ with non-degenerate pairing $\omega^* \circ m$, and an algebra map $b : T(F) \to B$, such that $b \in \Cent(F,B)$.
\end{tabular}
\end{center}
\caption{Continuation of table \ref{tab:alg-cft-table1}.}
\label{tab:alg-cft-table2}
\end{table}

For later use we give a name to the space of maps for which the diagram \eqref{eq:full-cent-univprop} commutes. For $A$ an algebra in $\Cc$ and $X \in \CoC$
\be
  \Cent(X,A) := \big\{\, x : T(X) \to A \,\big|\, \text{  \eqref{eq:full-cent-univprop} commutes } \,\} 
\ee
 (`Cent' for centrality condition, cf.\,table \ref{tab:alg-cft-table2}).

\begin{remark} \label{rem:fullcentre}
(i) The above definition can be recast into describing the full centre in $\CoC$ as a terminal object. Namely, consider the category $\Cc_{\text{full\,center}}(A)$ whose objects are pairs $(X,x)$ with $X \in \CoC$ and $x \in \Cent(X,A)$, and whose morphisms are maps $f : X \to Y$ in $\CoC$ such that
\be \label{eq:morph-in-cat-of-pairs}
\raisebox{1.4em}{\xymatrix@R=1.5em@C=1.5em{
T(X) \ar[rr]^{T(f)} \ar[dr]_{x} && T(Y) \ar[dl]^y
\\
& A
}}
\ee
commutes. By definition, the full centre $(Z,z)$ of an algebra $A$ is a terminal object in $\Cc_{\text{full\,center}}(A)$. 
\\[.3em]
(ii) The full centre was introduced in \cite{Runkel:2005qw} and \cite[Def.\,4.9]{Fjelstad:2006aw} in the case that $\Cc$ is a modular category and is (in particular) an object in $\CoC$. The notion was then generalised to arbitrary monoidal categories (not necessarily braided or abelian) in \cite[Sect.\,4]{Davydov:2009}, where the full centre, if it exists, is an object in the monoidal centre $\Zc(\Cc)$ of $\Cc$. If $\Cc$ is modular, $\Zc(\Cc) \cong \CoC$ by \cite[Thm.\,7.10]{Muger2001b}, and the two definitions agree (cf.\,\cite[Sect.\,8]{Davydov:2009}). In general, $\Zc(\Cc)$ and $\CoC$ may not be equivalent (but $T$ factors through $\Zc(\Cc)$, cf.\ remark \ref{rem:T-central}). For this reason, we added the suffix ``in $\CoC$'' to the name ``full centre'' in definition \ref{eq:full-centre-in-C2}.
However, because we will only ever use the full centre in $\CoC$ and never the full centre in $\Zc(\Cc)$, we will drop the suffix ``in $\CoC$'' from now on.
\end{remark}

Let $(Z,z)$ be the full centre of an algebra $A$ in $\Cc$ as in definition \ref{eq:full-centre-in-C2}. Suppose we are given a morphism $\mu_Z : Z \otimes_{\Cc^2} Z \to Z$; this will later be an associative, commutative product, but let us not demand that yet. Eqn.\,\eqref{eq:F(A)-product} defines a product $\mu_{T(Z)}$ on $T(Z)$. Suppose further that $z$ intertwines $\mu_{T(Z)}$ and $\mu_A$, i.e.\
\be \label{eq:z-is-algebramap}
\raisebox{2em}{\xymatrix@R=1em@C=1.5em{
T(Z) \otimes_{\Cc} T(Z) \ar[rr]^{\mu_{T(Z)}} \ar[dd]_{z \otimes_{\Cc} z} 
&& T(Z) \ar[dd]^z
\\
& T(Z \otimes_{\Cc^2} Z) \ar[ur]_-{T(\mu_Z)} \ar[ul]^{T_{2;Z,Z}^{-1}}
\\
A \otimes_{\Cc} A \ar[rr]^{\mu_A} && A
}}
\ee
commutes (we included also the definition of $\mu_{T(Z)}$). This diagram can be read in a second way: Starting from $T(Z \otimes_{\Cc^2} Z)$ and following the two paths to $A$, we see that it is an instance of \eqref{eq:full-cent-pairmap}. If we can establish \eqref{eq:full-cent-univprop} for the left path, the universal property of $(Z,z)$ provides us with a unique choice for $\mu_Z$, which, as we will see, is automatically associative and commutative. This is done in the next statement, which is just \cite[Prop.\,4.1]{Davydov:2009} with $\Zc(\Cc)$ replaced by $\CoC$. Even the proof works in the same way. Still, as the full centre is one of the main players in this paper we include parts of the proof in appendix \ref{app:fullcentre-alg-assoc+comm}.

\begin{theorem} \label{thm:fullcentre-alg-assoc+comm}
Let $(Z,z)$ be the full centre of an algebra $A \in \Cc$. 
There exists a unique product $\mu_Z : Z \otimes Z \to Z$ such that \eqref{eq:z-is-algebramap} commutes. This product is associative and commutative. If $A$ has a unit $\iota_A$, then there exists a unique map $\iota_Z : \one \to Z$ such that 
\be
\raisebox{2em}{\xymatrix{
T(\one) \ar[r]^{T(\iota_Z)}  \ar[d]_{(F_0)^{-1}}
& T(Z) \ar[d]^{z}
\\
\one \ar[r]^{\iota_A} & A
}}
\ee 
commutes. This map is a unit for the product $\mu_Z$. 
\end{theorem}

In particular, $z : T(Z) \to A$ is an algebra map. It is unital if $A$ is unital.

\subsection{The right adjoint $R$ of $T$} \label{sec:computeR1*}

As in the previous subsection, $\Cc$ is assumed to be $k$-linear abelian and braided monoidal, to satisfy (PF) and (C), and to have a $k$-linear right exact tensor product.
The aim of this section is to show the existence of the right adjoint $R  : \Cc \to \CoC$ of the functor $T : \CoC \to \Cc$ and give an explicit expression for it. In the next subsection, the adjoint $R$ will be used to give an explicit description of the full centre and thereby prove its existence.

\medskip

Consider the functor $\Cc(T(-),\one^*)$ from $\CoC$ to $\vect$. Define $R_{\one^*}$ (if it exists) to be the representing object of this functor. That is, there is a natural (in $X$) isomorphism of functors $\CoC \to \vect$,
\be \label{eq:R1*-chi}
  \chi_X : \Cc(T(X),\one^*) \overset{\sim}\longrightarrow \Cc^2(X,R_{\one^*}) \ .
\ee
Here and below, $\Cc^2(X,Y)$ denotes the space of morphisms from $X$ to $Y$ in $\CoC$.
We will now show that $R_{\one^*}$ may be written as the cokernel of a morphism between two projective objects in $\CoC$; in particular, $R_{\one^*}$ exists.

Let $P$ be a projective generator of $\Cc$ (which exists by (PF)). By corollary \ref{cor:proj-gen-in-prod}, also $\CoC$ satisfies property (PF) and $P \boxtimes P$ is a projective generator of $\CoC$. Define the linear subspace $N \subset \Cc^2(P \boxtimes P, P \boxtimes P^*)$ to consist of all $f : P \boxtimes P \to P \boxtimes P^*$ such that
\be \label{eq:N-def-cond}
  \Big[ P \otimes_{\Cc} P \xrightarrow{~T(f)~} P \otimes_{\Cc} P^* \xrightarrow{~\beta_P~} \one^* \Big] ~=~ 0 \ ,
\ee
where $\beta_P$ is the non-degenerate pairing defined in \eqref{eq:def-beta-pairing}. 
Let $\{u_1, \dots, u_{|N|}\}$ be a basis of $N$ (the space is finite-dimensional by (PF)). Define the map $n : (P \boxtimes P)^{\oplus |N|} \to P \boxtimes P^*$ as $n = \sum_{i=1}^{|N|} u_i \circ \pi_i$, with $\pi_i$ the projection to the $i$'th direct summand. Define $R'$ to be the cokernel of $n$, so that we have the exact sequence
\be \label{eq:R'-def}
  (P \boxtimes P)^{\oplus |N|} \xrightarrow{~n~} P \boxtimes P^* \xrightarrow{\cok(n)} R' \longrightarrow 0 \ .
\ee
Now consider the diagram
\be \label{eq:r'-def}
\raisebox{2em}{\xymatrix@C=3em{
T((P \boxtimes P)^{\oplus |N|})  \ar[r]^-{T(n)}
& P \otimes P^* \ar[d]_{\beta_P} \ar[r]^-{T(\cok(n))}
& T(R') \ar@{-->}[dl]^{\exists ! r'} \ar[r] & 0
\\
& \one^*
}}
\qquad .
\ee
Since by its construction in \eqref{eq:TC-construction}, $T$ is right exact, the top row of the diagram is exact, i.e.\ $T(\cok(n))$ is the cokernel of $T(n)$.
Because $\beta_P \circ T(n) = \sum_{i=1}^{|N|} \beta_P \circ  T(u_i) \circ T(\pi_i)=0$ (by definition of the $u_i$), from the universal property of the cokernel we obtain the arrow $r' : T(R') \to \one^*$. The next theorem, whose proof can be found in appendix \ref{app:more-on-R1}, states that $R'$ is the object we are looking for.

\begin{theorem} \label{thm:R1*-via-projectives}
The object $R'$ just constructed represents the functor $\Cc(T(-),\one^*)$, i.e.\ one may take $R_{\one^*} = R'$.
\end{theorem}

We will soon use the object $R_{\one^*}$ to construct the entire adjoint functor $R$, but first we would like to state one important property of $R_{\one^*}$. Given a natural transformation $(\nu_U : U \to U)_{U \in \Cc}$  of the identity functor on $\Cc$, set
\be
  \tilde \nu_U = \Big[ U \xrightarrow{\delta_U} U^{**} \xrightarrow{(\nu_{U^*})^*} U^{**} \xrightarrow{\delta_U^{-1}} U \Big] \ ,
\ee
where $\delta$ is the natural isomorphism $\Id \Rightarrow (-)^{**}$ from condition (C). Then $\tilde \nu$ is again natural in $U$. In particular, both $\boxtimes \circ (\nu \times \id)$ and $\boxtimes \circ (\id \times \tilde\nu)$ are natural transformations of $\boxtimes : \CxC \to \CoC$. Via the defining equivalence \eqref{eq:defining-equiv}, these give two natural transformations $\nu \boxtimes \id$ and $\id \boxtimes \tilde\nu$ of the identity functor on $\CoC$. 

\begin{theorem} \label{thm:trivial-twist}
Let $R_{\one^*} \in \CoC$ be as above and let $\nu : \Id_{\Cc} \Rightarrow \Id_{\Cc}$ be a natural transformation. Then $(\nu \boxtimes \id)_{R_{\one^*}} = (\id \boxtimes \tilde\nu)_{R_{\one^*}}$.
\end{theorem}

The theorem is proved in appendix \ref{app:more-on-R1}.

\begin{remark}\label{rem:trivial-twist}
Let $\Vc$ be a vertex operator algebra such that $\Cc = \Rep(\Vc)$ satisfies the conditions set out in the beginning of this subsection. Then $\exp(2 \pi i L_0)$ acting on some representation $A \in \Cc$ is an example of a natural transformation of the identity functor (it commutes with all modes of all fields in the VOA, and it can be moved past all intertwiners $f : A \to B$, i.e.\ it is natural in $A$). Theorem \ref{thm:trivial-twist} states in this case that $\exp(2 \pi i \cdot L_0 \otimes_{\Cb} \id)$ and $\exp(2 \pi i \cdot \id \otimes_{\Cb} L_0)$ act in the same way on the $\Vc \otimes_{\Cb} \Vc$-module $R_{\one^*}$. In other words,
\be
  \exp\!\Big\{ 2 \pi i \big( L_0 \otimes_{\Cb} \id - \id \otimes_{\Cb} L_0\big) \Big\}\Big|_{R_{\one^*}} = \id_{R_{\one^*}} \ .
\ee
In CFT terms this means that in a situation where $R_{\one^*}$ is the space of bulk fields\footnote{
     For this to be possible, we must have $Z(\one^*) = R(\one^*)$ (we will see in \eqref{eq:R(U)-expression} that $R(\one^*) = R_{\one^*}$). By lemma \ref{lem:transparent-Z=R} below this is true if $\one^* \cong \one$. We expect that $Z(\one^*) = R(\one^*)$ also holds in the $W_{2,3}$-model (where $\one^* \not\cong \one$), see section \ref{sec:W23-model}.},     
the partition function is invariant under the T-transformation $\tau \mapsto \tau+1$.
\end{remark}

We now turn to the right adjoint $R$. The involution $(-)^*$ on $\Cc$ induces an involution on $\CoC$, which we also denote by $(-)^*$, and which also satisfies condition (C) (see appendix \ref{app:R-adj-T}). We can use $R_{\one^*}$ and the involution $(-)^*$ to define a functor $R : \Cc \to \CoC$ as 
\be \label{eq:R(U)-expression}
  R(U) = 
  \Big( \,(U^* \boxtimes \one) \,\otimes_{\Cc^2} \,(R_{\one^*})^* \,\Big)^{\!*}
  \quad , \quad
  R(f) = 
  \Big( \, (f^* \boxtimes \id_\one)\,  \otimes_{\Cc^2} \, \id_{(R_{\one^*})^*} \,\Big)^{\!*} \ . 
\ee
Note that $R(\one^*) \cong R_{\one^*}$. 

\begin{theorem} \label{thm:R-adj-to-T}
The functor $R$ is a right adjoint for $T$.
\end{theorem}

The proof and the adjunction isomorphisms are given in appendix \ref{app:R-adj-T}. On general grounds, the functor $R$, being adjoint to a monoidal functor, is lax monoidal (see, e.g., \cite[Lem.\,2.7]{Kong:2008ci}). The structure maps $R_0$ and $R_2$ can equally be found in appendix \ref{app:R-adj-T}. Thus, for an algebra $A \in \Cc$, $R(A)$ is an algebra in $\CoC$ with multiplication \eqref{eq:F(A)-product}.

\begin{remark}
(i)~Since $R$ is a right adjoint functor, it is left exact. This can also be seen explicitly from \eqref{eq:R(U)-expression}, namely
\be
  R ~=~
  \big[\, \Cc^{\mathrm{op}} \xrightarrow{(-)^*} \Cc \xrightarrow{(-) \boxtimes \one} \CoC \xrightarrow{(-) \otimes_{\Cc^2} (R_{\one^*})^*} \CoC \xrightarrow{(-)^*} (\CoC)^{\mathrm{op}} \,\big] \ ,
\ee
where $(-)^*$ is exact, $(-) \boxtimes \one$ is exact (see the beginning of the proof of theorem \ref{thm:deligne-product-of-alg}), and $(-) \otimes_{\Cc^2} (R_{\one^*})^*$ is right exact. Thus $R$ is a right exact functor $\Cc^{\mathrm{op}} \to (\CoC)^{\mathrm{op}}$ which is the same as a left exact functor $\Cc \to \CoC$.
\\[.3em]
(ii)~Suppose there are isomorphisms $(U \otimes_{\Cc} V)^* \to V^* \otimes_{\Cc} U^*$, natural in $U$ and $V$. Then, firstly, $\one \cong (\one^*)^* \cong ( \one \otimes_{\Cc} \one^* )^* \cong \one \otimes_{\Cc} \one^* \cong \one^*$. Secondly, the above formula for $R$ simplifies to $R(U) = R_{\one} \,\otimes_{\Cc^2} \,(U \boxtimes \one)$, and analogously for $R(f)$. In this formulation, $R$ is clearly right exact, so that together with (i) we see that $R$ is exact. Similarly, the natural isomorphism $(-)^* \circ T \circ (-)^* \cong T$ shows that $T$ is exact.
\end{remark}

\subsection{Left centre and full centre} \label{sec:leftcentre}

In this subsection, we will express the full centre of an algebra $A$ as the `left centre' -- to be defined momentarily -- of the adjoint functor $R$ applied to $A$. As before, $\Cc$ is assumed to be $k$-linear abelian and braided monoidal, to satisfy (PF) and (C), and to have a $k$-linear right exact tensor product.

\medskip

In the braided setting, one distinguishes three different notions of the centre of an algebra: the left centre, the right centre and the full centre. From these, the left and right centre are subobjects of the algebra itself, while -- as we have seen in definition \ref{eq:full-centre-in-C2} -- the full centre lives in a different category. The left and right centres were introduced in \cite{Oystaeyen:1998} and appeared in various incarnations in \cite{Ostrik:2001,Frohlich:2003hm,Davydov:2009}. The following definition is taken from \cite[Sect.\,5]{Davydov:2009}.

\begin{definition} \label{def:left-centre}
Let $B$ be an algebra in a braided monoidal category $\Bc$. The \textsl{left centre} of $B$ is an object $C_l(B)$ in $\Bc$ together with a morphism $e : C_l(B) \to B$ such that the following universal property holds: For every $U \in \Bc$ and arrow $u : U \to B$ such that
\be \label{eq:left-centre-diagram}
\raisebox{2.3em}{\xymatrix@R=0.7em{
U \otimes_\Bc B \ar[dd]_{c_{U,B}} \ar[rr]^{u \,\otimes_\Bc \,\id_B} && B \otimes_\Bc B \ar[dr]^{\mu_B}
\\
&&& B
\\
B \otimes_\Bc U  \ar[rr]^{\id_B\, \otimes_\Bc\, u} && B \otimes_\Bc B \ar[ur]_{\mu_B}
}}
\ee
commutes, there exists a unique arrow $\tilde u : U \to C_l(B)$ such that $e \circ \tilde u = u$.
\end{definition}

\begin{remark} \label{rem:left-right-centre}
(i) There is an analogous definition for the right centre, see \cite{Oystaeyen:1998,Ostrik:2001}.
\\[.3em]
(ii) If the category $\Bc$ is in addition abelian and has conjugates as in (C), the left centre of an algebra $B$ can be expressed as a kernel. This shows at the same time that the left centre exists and that $e : C_l(B) \to B$ is injective, i.e.\ $C_l(B)$ is a subobject of $B$. See appendix \ref{app:full-centre-via-left-centre} for details. 
\\[.3em]
(iii) $C_l(B)$ carries a unique algebra structure such that $e : C_l(B) \to B$ is an algebra map. This algebra structure on $C_l(B)$ is commutative. If $A$ is unital, so is $C_l(A)$. See \cite{Oystaeyen:1998,Ostrik:2001,Davydov:2009} for details.
\end{remark}

In the previous subsection we gave the direct definition of the full centre as first formulated in \cite[Sect.\,4]{Davydov:2009}. The original definition in \cite[Def.\,4.9]{Fjelstad:2006aw} proceeds in two steps: first, one applies the adjoint $R$ of $T$ to the algebra $A$ and second, one finds the left centre of $R(A)$. The same works in the present setting, as we now show. The proof is the same as in \cite[Thm.\,5.4]{Davydov:2009}, we reproduce an adapted version in appendix \ref{app:full-centre-via-left-centre}. Denote the adjunction natural transformation $TR \Rightarrow \Id$ by $\eps$, cf.\,\eqref{eq:alpha-beta-def}.

\begin{theorem} \label{thm:full-centre-via-left-centre}
Let $A$ be an algebra in $\Cc$. The pair $(Z,z)$ with
\be
Z = C_l(R(A))
\quad , \qquad
z = \big[ \,T(C_l(R(A))) \xrightarrow{T(e)} T(R(A)) \xrightarrow{\eps_A} A \,\big]
\ee
is the full centre of $A$.
\end{theorem}

In particular, since $R$ exists by theorem \ref{thm:R-adj-to-T} and the left centre exists by remark \ref{rem:left-right-centre}\,(ii), the full centre of an algebra exists under the assumptions set out in the beginning of this subsection.

\medskip

Even if $A$ is a commutative algebra, $R(A)$ need not be commutative and one still needs to take the left centre to arrive at $Z(A)$. However, the next lemma gives a simple condition in addition to commutativity which guarantees $Z(A) = R(A)$; this will be useful in section \ref{sec:W23-model}. An object $S \in \Cc$ is called \textsl{transparent} if $c_{U,S} \circ c_{S,U} = \id_{S \otimes U}$ for all $U \in \Cc$. 

\begin{lemma} \label{lem:transparent-Z=R}
If $(S,\mu_S)$ is a commutative algebra in $\Cc$ and $S$ is transparent in $\Cc$, then we can take $(Z(S),z) = (R(S),\eps_S)$. In particular, $Z(\one) = R(\one)$.
\end{lemma}

\begin{proof}
From \eqref{eq:mixed-braid-tildephi} one checks that for transparent $S$ we have $\tilde \varphi_{A,B,S} = c_{A\otimes B,S}$. Thus also $\varphi_{X,S} = c_{T(X),S}$. Condition \eqref{eq:full-cent-univprop} is then true for all $x : T(X) \to S$ as by commutativity of $S$ we have $\mu_S \circ c_{S,S} = \mu_S$. But then the universal property of the full centre reduces to that in lemma \ref{lem:repres+terminal}\,(ii) with $U=S$, $R'=Z(S)$ and $r'=z$. By part (i) of that lemma, $R'=R(S)$ and $r' = \xi_{R(S),S}^{-1}(\id_{R(S)}) = \eps_S$, see \eqref{eq:xi_XU-def} and \eqref{eq:alpha-beta-def}.
\end{proof}

\begin{remark}\label{rem:FSS-Hopf}
A different approach to finding algebraic counterparts to logarithmic CFTs on $\Cb$ is taken in \cite{Fuchs:2011mg}. There, the category $\Cc$ is chosen to be $H\text{-Mod}$ for a certain Hopf algebra $H$ (in more detail, finite dimensional representations of a finite-dimensional factorisable ribbon Hopf algebra) and $\CoC \cong H\text{-Bimod}$ (see \cite[App.\,A.3]{Fuchs:2011mg}). In $H\text{-Bimod}$ the coregular bimodule $H^*$ is studied and shown to be a commutative Frobenius algebra \cite[Prop.\,2.10\,\&\,3.1]{Fuchs:2011mg}. In addition, $H^*$ satisfies certain modular invariance properties \cite[Thm.\,5.6]{Fuchs:2011mg}. In \cite[App.\,B]{Fuchs:2011mg}, the bimodule $H^*$ is proposed to be a candidate bulk theory for a logarithmic CFT in case $H\text{-Mod} \cong \Rep\Vc$ for the vertex operator algebra $\Vc$ encoding the chiral symmetry. In the setting of the present paper, $H^*$ corresponds to $R(\one)$. 
\end{remark}

\section{The $W_{2,3}$-model with $c=0$} \label{sec:W23-model}

In this section we look more closely at one particular class of examples, namely conformal field theories built from representations of the $W_{2,3}$ vertex operator algebra. This symmetry algebra was chosen because it demonstrates that the level of generality assumed in section \ref{sec:alg-reform} is indeed needed in the treatment of interesting examples.

We start in section \ref{sec:W23-reps} with a brief collection of what is known or expected about the representation theory of the $W_{2,3}$ vertex operator algebra. In section \ref{sec:RW(0)-example} it is shown how the formalism of finding the maximal bulk theory for a given boundary theory can produce the (trivial) $c=0$ Virasoro minimal model. A non-trivial bulk theory is discussed in sections \ref{sec:RW*-socle-dec}--\ref{sec:R1-OPE}; 
this bulk theory is logarithmic and can be understood as a `refinement' of the $c=0$ minimal model.

\subsection{The $W$-algebra and some of its representations} \label{sec:W23-reps}

Let $\mathrm{Ver}(h{=}0,c{=}0)$ be the Virasoro Verma module generated by the state $\Omega$ with $L_0 \Omega = C \Omega = 0$. It has a maximal proper submodule which is generated by the two vectors
\be
  n_1 = L_{-1} \Omega
  \quad , \quad
  n_2 = \big( L_{-2} - \tfrac32 L_{-1} L_{-1} \big) \Omega \ .
\ee
Since words in $L_{-1}$ and $L_{-2}$ acting on $\Omega$ span $\mathrm{Ver}(0,0)$, the quotient  $\mathrm{Ver}(0,0)/ \langle n_1,n_2 \rangle$ is just $\Cb \Omega$ with trivial $\Vir$-action. This describes the vacuum representation of the Virasoro minimal model with $c=0$, which is trivial in the sense that it is a two-dimensional topological field theory for the commutative algebra $\Cb$, cf.\ remark \ref{rem:F-trivial}. The module $\Vc \equiv \mathrm{Ver}(0,0) / \langle n_1 \rangle$ is infinite dimensional and carries the structure of a vertex operator algebra with Virasoro element $T = L_{-2} \Omega \neq 0$ (note that $\mathrm{Ver}(0,0)$ is not itself a vertex operator algebra because the vacuum $\Omega$ is not annihilated by the translation operator $L_{-1}$). The VOA $\Vc$ has an infinite number of distinct irreducible representations (see \cite[Thm.\,4.4]{Frenkel:1992} and \cite[Sect.\,2.3]{Milas:2001}). To be able to apply the discussion in section \ref{sec:alg-reform}, we can pass to a larger VOA $\Wc \supset \Vc$, which is the chiral symmetry algebra for the $W_{2,3}$-model \cite{Feigin:2006iv,Adamovic:2009a}, and is obtained as an extension of $\Vc$ by two fields of weight 15. Its character reads
\begin{eqnarray}
  \chi_{\Wc}(q) &=& 1 + q^2 + q^3 + 2q^4 +2q^5 +4q^6 +4q^7 +7q^8 +8q^9 +12q^{10} 
  \nonumber \\
  && \hspace{1.4em} + 14q^{11} + 21q^{12} + 24q^{13} + 34q^{14} + 44q^{15} + 58q^{16} + \dots
  \quad .
\end{eqnarray}
It turns out that $\chi_{\Wc}(q)$ differs from the character of $\Vc$ only starting from $q^{15}$, namely $\chi_{\Vc}(q) = {\dots} +  41q^{15} + 55q^{16} + {\dots}$ (so e.g.\ there are three new fields at weight 15). The VOA $\Wc$ is $C_2$-cofinite \cite{Adamovic:2009a} and has 13 irreducible representations \cite{Feigin:2006iv,Adamovic:2009a}, which we label by their lowest $L_0$-weight:
\be \label{eq:W23-irreps}
  \begin{tabular}{c|ccc}
    & $s=1$ & $s=2$ & $s=3$ \\
    \hline
    \\[-.8em]
    $r=1$ & $0,\,2,\,7$~&~ $0,\,1,\,5$~
    &$\frac{1}{3},\,\frac{10}{3}$\\[.5em]
    $r=2$ & $\frac58,\,\frac{33}{8}$& $\frac18,\,\frac{21}{8}$& $\frac{-1}{24},\,\frac{35}{24}$
  \end{tabular}
\ee
Here, the two entries `0' refer to the same irreducible representation $\Wc(0) \equiv \Cb \Omega$. We write $\Wc(h)$ for the irreducible $\Wc$-representation of lowest $L_0$-weight $h$. Their characters (from \cite{Feigin:2006iv}) are listed in our notation in \cite[App.\,A.1]{Gaberdiel:2009ug}.

At this point we note the first oddity of the $W_{2,3}$-model: the vertex operator algebra $\Wc$ is not one of the irreducible representations: it is indecomposable but not irreducible. Indeed, $\Omega$ is a cyclic vector (hence indecomposability) and the stress tensor $T = L_{-2} \Omega$ generates a $\Wc$-subrepresentation (on the level of Virasoro modes, this follows since $L_n T = 0$ for $n>0$, where $L_2 T = 0$ is a special feature of $c=0$). Specifically, $\Wc$ is the middle term in a non-split exact sequence
\be \label{eq:W-sequence}
  0 \longrightarrow \Wc(2) \longrightarrow \Wc \longrightarrow \Wc(0) \longrightarrow 0 \ .
\ee

This brings us to the second oddity. Denote by $R^*$ the contragredient representation of $R$ (see, e.g.\ \cite[Def.\,I:2.35]{Huang:2010}) of a representation $R$. Then $\Wc^* \not\cong \Wc$, as can be seen for example from their socle filtration\footnote{
  The \textsl{socle} $\mathrm{soc}(M)$ of a module $M$ is the largest semi-simple submodule contained in $M$. The \textsl{socle filtration} of $M$ is the unique filtration $\{0\} = M_0 \subset M_1 \subset M_2 \subset \cdots \subset M_n = M$ where $M_1$ is the socle of $M$, $M_2/M_1$ is the socle of $M/M_1$ and in general $M_{i+1}/M_i$ is semi-simple and equals the socle of $M/M_i$. The socle filtration is a unique version of the composition series. In the latter, one iteratively picks a simple submodule and quotients by it. The composition series hence involves choices.
}
\be \label{eq:WW*-compos}
  \Wc ~:~ \raisebox{1.3em}{\xymatrix@R=1em{ 0 \ar@{-}[d] \\ 2 }}
  \qquad , \qquad
  \Wc^* ~:~ \raisebox{1.3em}{\xymatrix@R=1em{ 2 \ar@{-}[d] \\ 0 }} 
  \quad .
\ee
The above diagrams show the semi-simple quotients of successive submodules in the socle filtration. For example, the largest semi-simple subrepresentation of $\Wc$ is $\Wc(2)$. Quotienting by $\Wc(2)$, one obtains a representation whose largest semi-simple subrepresentation is $\Wc(0)$, and this accounts for all of $\Wc$ (this is just the statement of the sequence \eqref{eq:W-sequence} and the fact that it is non-split). For $\Wc^*$, the largest semi-simple subrepresentation is $\Wc(0)$ and the quotient is isomorphic to $\Wc(2)$.

\medskip

We have now pretty much reached the frontier of established mathematical truth regarding the $W_{2,3}$-model. Hence it is time for the following
\begin{quote}
{\bf Disclaimer:} The statements concerning the structure of the $W_{2,3}$-model in the remainder of section \ref{sec:W23-model} should be treated as conjectures, even if we refrain from writing `conjecturally' in every sentence.
\end{quote}
The first statement under the umbrella of the above disclaimer is: The tensor product theory of \cite{Huang:2010} turns $\Cc \equiv \Rep(\Wc)$ into a braided monoidal category
\begin{itemize}\setlength{\leftskip}{-1em}
\item[+] which has property (PF) from section \ref{sec:deligne-prod},
\item[+] whose tensor product functor is right exact in each argument,
\item[+] whose contragredient functor $(-)^*$ has property (C) from section \ref{sec:conjugate}. 
\end{itemize}
The category $\Rep(\Wc \otimes_{\Cb} \Wc)$ (with inverse convention for the braiding in the second factor) is just the Deligne product $\CoC$. Every irreducible $\Wc(h)$ has a projective cover, which we denote by $\Pc(h)$. The fusion rules of the representations generated from the 13 irreducibles and from $\Wc^*$ (and from two representations $\Qc$, $\Qc^*$ which have the socle filtration \eqref{eq:WW*-compos} with $2$ replaced by $1$) are listed in \cite[App.\,4]{Gaberdiel:2009ug} (see also \cite{Eberle:2006zn,Feigin:2006iv,Rasmussen:2008ii}); some of the fusion rules of the projective cover $\Pc(0)$ are given in \cite[App.\,B.1]{Gaberdiel:2010rg}. The fusion-tensor product of $\Cc$ will be denoted by $\otimes_f$.

The properties marked `$+$' above allow one to apply the formalism in section \ref{sec:alg-reform}. However, there are many other convenient properties which $\Cc$ does not have:
\begin{itemize}\setlength{\leftskip}{-1em}
\item[$-$] $\Cc$ is not semi-simple (e.g.\ the sequence \eqref{eq:W-sequence} is not split).
\item[$-$] The tensor unit $\one \equiv \Wc$ in $\Cc$ is not simple (cf.\,\eqref{eq:W-sequence}).
\item[$-$] $\one \equiv \Wc$ is not isomorphic to its conjugate $\one^* \equiv \Wc^*$ (cf.\,\eqref{eq:WW*-compos}).
\item[$-$] The involution $(-)^*$ is not monoidal, e.g.\ $(\Wc \otimes_f \Wc(0))^* = \Wc(0)$ and $\Wc^* \otimes_f \Wc(0)^* = 0$.
\item[$-$] The tensor product of $\Cc$ is not exact. For example, the functor $\Wc(0) \otimes_f (-)$ transports the exact sequence $0 \rightarrow \Wc(0) \rightarrow \Wc^* \rightarrow \Wc(2) \rightarrow 0$ to $0 \rightarrow \Wc(0) \rightarrow 0 \rightarrow 0 \rightarrow 0$, which is not exact.
\item[$-$] $\Cc$ is not rigid, i.e.\ not every object has a dual (in the categorical sense -- not to be confused with the contragredient representation, which always exists). Examples are the irreducibles $\Wc(0)$, $\Wc(1)$, $\Wc(2)$, $\Wc(5)$, $\Wc(7)$, the contragredient $\Wc^*$ of the VOA, and the projective cover $\Pc(0)$; we refer to  \cite[Sect.\,1.1.1]{Gaberdiel:2009ug} and \cite[App.\,B.1]{Gaberdiel:2010rg} for details.
\item[$-$] Even if $U \in \Cc$ has a dual $U^\vee$, it may happen that $U^\vee \not\cong U^*$, e.g.\ $\Wc^\vee = \Wc$ (the tensor unit is self-dual in any monoidal category), but $\Wc^* \not\cong \Wc$.
\end{itemize}

For the rest of this subsection we take a look at the most intricate\footnote{
   Also the $\Pc(h)$ with $h \in \{1,2,5,7\}$ have a socle filtration with 5 levels (see \cite[App.\,A.1]{Gaberdiel:2010rg}). $\Pc(0)$ is `most intricate' in the sense that it does not occur in the representations generated by fusion from the 13 irreducibles and its structure has only been found by indirect reasoning.
}
of the $\Wc$-representations, the projective cover $\Pc(0)$ of $\Wc(0)$. It has the socle filtration (as argued for in \cite[App.\,B.2]{Gaberdiel:2010rg}):
\be
\label{eq:P0-diagram}
  \begin{tikzpicture}
    [baseline=(lev33),scale=.8, vert/.style={inner sep=1,
      text centered,circle,anchor=base}]
    \node at (0,1) {$\Pc(0)$};
    \path
    (0,0) node[vert] (lev11) {0}
    ++(-0.5,-1) node[vert] (lev21) {1} edge[-] (lev11)
    ++(1,0) node[vert] (lev22) {2} edge[-] (lev11)
    ++(1.5,-1) node[vert] (lev35) {7} edge[-] (lev21) edge[-] (lev22)
    ++(-1,0) node[vert] (lev34) {7} edge[-] (lev21) edge[-] (lev22)
    ++(-1,0) node[vert] (lev33) {0} edge[-] (lev21) edge[-] (lev22)
    ++(-1,0) node[vert] (lev32) {5} edge[-] (lev21) edge[-] (lev22)
    ++(-1,0) node[vert] (lev31) {5} edge[-] (lev21) edge[-] (lev22)
    ++(1.5,-1) node[vert] (lev41) {1} edge[-] (lev31) edge[-] (lev32) edge[-] (lev33) edge[-] (lev34) edge[-] (lev35)
    ++(1,0) node[vert] (lev42) {2} edge[-] (lev31) edge[-] (lev32) edge[-] (lev33) edge[-] (lev34) edge[-] (lev35)
    ++(-0.5,-1) node[vert] (lev51) {0} edge[-] (lev41) edge[-] (lev42);
    \path
    (-3,1) node[vert] {Level}
    ++(0,-1) node[vert] {0}
    ++(0,-1) node[vert] {1}
    ++(0,-1) node[vert] {2}
    ++(0,-1) node[vert] {3}
    ++(0,-1) node[vert] {4};
  \end{tikzpicture}
\ee
As before, the numbers in each row give the simple summands in the quotient of two consecutive layers of the socle filtration. The lines indicate the action of the $W$-modes; for $\Pc(0)$ they merely state that a vector at a given level can be transported into any of the lower lying submodules (this is not so for $\Pc(h)$ with $h=1,2,5,7$, see \cite[App.\,A.1]{Gaberdiel:2010rg}).

There are three quasi-primary states of generalised $L_0$-weight 0, no such states at weight 1, 
two states at weight 2, and infinitely more at higher weights.\footnote{The character of $\Pc(0)$ starts as
 $3 + 2 q + 4 q^2 + \dots$. The two states at weight 1 are $L_{-1}$-descendents, as are two of the
four states at weight 2.  }
It seems natural to us that the quasi-primary states up to weight 2 organise themselves under the Virasoro action as in the following diagram (the action is given up to constants, see below for the full expressions)
\be\label{5fold}
\raisebox{2.1em}{\xymatrix@R=2.7em@C=2.7em{
\text{$L_0$-weight $2$\hspace{-1em}}& & t \ar[rr]^{L_0-2} \ar[dr]_{L_2} \ar@/_1.2em/@{..>}[drrr]!<-8pt,4pt>
&& \Tc \ar[dr]^{L_2} 
\\
\text{$L_0$-weight $0$\hspace{-1em}}& \eta \ar[ur]^{\!L_{-2}+..} \ar[rr]_{L_0}
&& \omega \ar[ur]^(0.55){\!L_{-2}+..} \ar[rr]_{L_0}
&& \Omega
}}
\qquad ,
\ee
where $L_{-2} + \dots$ stands for the operators defined by 
\be \label{eq:tT-def}
 t := \big( \,L_{-2} - \tfrac32 L_{-1} L_{-1}  + \tfrac95 (L_{-2} + \tfrac16 L_{-1}L_{-1}) L_0 \,
\big) \,\eta ~ , ~~
 \Tc := \big( L_{-2} - \tfrac32 L_{-1} L_{-1} \big) \,\omega \ .
\ee
In more detail, let us assume that we are given a Virasoro representation with the following properties:
\begin{enumerate}
\item It allows for a non-degenerate symmetric pairing such that $\langle v, L_m w \rangle = \langle L_{-m} v, w \rangle$.
\item It has a cyclic vector $\eta$ which is primary (i.e.\ the $\Vir$-action on $\eta$ generates the entire representation and $L_m \eta = 0$ for all $m>0$).
\item $\eta$ generates a rank three Jordan cell for $L_0$ of generalised eigenvalue 0; we set 
\be\label{eq:om-Om-def}
  \omega := L_0 \eta \quad , \quad \Omega := L_0 \omega \ ,
\ee
\item $\Cb \Omega$ is the trivial Virasoro representation: $L_m \Omega = 0$ for all $m \in \Zb$.
\end{enumerate}
Let us draw some conclusions from these assumptions. Firstly, $\langle \omega , \Omega \rangle = \langle L_0 \eta, (L_0)^2 \eta \rangle = \langle \eta, (L_0)^3 \eta \rangle = 0$ and similarly $\langle \Omega , \Omega \rangle = 0$, so that we must have $\langle \eta, \Omega \rangle  \neq 0$ by non-degeneracy. Hence also $\langle \omega, \omega \rangle = \langle L_0 \eta, L_0 \eta \rangle = \langle \eta, \Omega \rangle  \neq 0$. Suppose that $\langle \eta , \omega \rangle \neq 0$. Then we can replace $\eta \leadsto \eta' := \eta + \text{(const)} \omega$ such that $\langle \eta' , L_0 \eta' \rangle = 0$. Next, if $\langle \eta' , \eta' \rangle \neq 0$ we replace $\eta' \leadsto \eta'' := \eta' + \text{(const)} \Omega$ such that $\langle \eta'' , \eta'' \rangle = 0$. We will henceforth assume that both has been done. Altogether, the pairing on the states of generalised $L_0$-eigenvalue $0$ is, for some normalisation constant $N \neq 0$:
\be \label{eq:M-weight0-IP}
\begin{tabular}{c|ccc}
$\langle\,,\,\rangle$ & $\eta$ & $\omega$ & $\Omega$
\\
\hline
$\eta$ & 0 & 0 & $N$ \\
$\omega$ & 0 & $N$ & 0\\
$\Omega$ & $N$ & 0 & 0
\end{tabular}
\ee
Secondly, using points 3.\ and 4.\ above,  one verifies with a little patience that for $t$ and $\Tc$ as defined in \eqref{eq:tT-def}
\be \label{eq:M-modeactions}
  \begin{array}{rclrcl}
  L_1 \,t &=& 0 \ , &
  L_1 \,\Tc &=& 0 \ ,
  \\[.3em]
  L_2 \,t &=& - 5 \omega 
     + 9 \Omega
     \ , &
  L_2 \,\Tc &=& - 5 \Omega \ ,
  \\[.3em]
  (L_0{-}2)\, t &=& \Tc \ , &
  (L_0{-}2) \,\Tc &=& 0 \ .
  \end{array}
\ee
It is then easy to compute the 
pairing on the weight 2 states (the pairing of a quasi-primary with an $L_{-1}$-descendent vanishes;
we give the pairing restricted to $t$, $\Tc$). For example, using invariance of the pairing and the relations \eqref{eq:M-modeactions} gives
\be
  \langle \Tc , t \rangle  = \langle \omega , (L_{2} - \tfrac32 L_{1} L_{1})  t \rangle 
= -5  \langle \omega,\omega \rangle \ .
\ee
Altogether, the pairing takes the form:
\be \label{eq:M-weight1-IP}
\begin{tabular}{c|cc}
$\langle\,,\,\rangle$ & $t$ & $\Tc$ 
\\
\hline
$t$ & 0 & $-5N$ \\
$\Tc$ & $-5N$ & 0
\end{tabular}
\ee
The fact that $\langle t,t\rangle=0$ is the motivation for the complicated choice of $t$ in \eqref{eq:tT-def}.

\medskip

In summary, if  the $\Vir$-submodule of $\Pc(0)$ generated by a state $\eta$ representing the top 0 in the socle filtration \eqref{eq:P0-diagram} indeed has properties 1.--4., then we have quasi-primary states $\omega$, $\Omega$, $t$, $\Tc$ defined as in \eqref{eq:tT-def} and \eqref{eq:om-Om-def} with the properties \eqref{eq:M-weight0-IP}--\eqref{eq:M-weight1-IP}. We will return to this in the discussion of OPEs.

\subsection{Computation of $R(\Wc(0))$} \label{sec:RW(0)-example}

An instance where we can compute the value of the adjoint functor $R$ directly is the one-dimensional $\Wc$-module $\Wc(0)$. Namely, as we will explain in the second half of this short subsection,
\be \label{eq:R(W0)}
  R(\Wc(0)) = \Wc(0) \boxtimes \Wc(0) \ .
\ee
The object $\Wc(0)$ is transparent because it is a quotient of $\Wc$ and the tensor unit is always transparent (recall that $\otimes_f$ is right exact and hence preserves surjections). Furthermore, $\Wc(0) \otimes_f \Wc(0) \cong \Wc(0)$ so that lemma \ref{lem:SS=S-algebra} implies that $\Wc(0)$ is a commutative associative algebra. Lemma \ref{lem:transparent-Z=R} now tells us that the full centre is
\be
  Z(\Wc(0)) = \Wc(0) \boxtimes \Wc(0) \ .
\ee
This result has an evident CFT interpretation. The algebra $\Wc(0)$ is a non-degenerate boundary theory in the sense of definition \ref{def:non-deg-bnd}. In fact, $\Wc(0)$ is nothing but the chiral symmetry algebra of the $c=0$ Virasoro minimal model. According to the discussion in sections \ref{sec:boundary-to-bulk} and \ref{sec:fullcentre}, $Z(\Wc(0))$ is the largest bulk theory that can be consistently and non-degenerately joined to the boundary theory $\Wc(0)$. It is then not surprising that this bulk theory is the $c=0$ Virasoro minimal model, i.e.\ the trivial theory with one-dimensional state space.

\medskip

The derivation of \eqref{eq:R(W0)} is as follows. We first remark that the functor $\Wc(0) \otimes_f (-)$ from $\Cc$ to $\Cc$ is monoidal (combine $\Wc(0) \otimes_f \Wc(0) \cong \Wc(0)$ with lemma \ref{lem:SS=S-algebra}). The image of $\Wc(0) \otimes_f (-)$ lies in the full subcategory of $\Cc_0 \subset \Cc$ of objects isomorphic to direct sums of $\Wc(0)$ (thus $\Cc_0$ is a tensor-ideal). But $\Cc_0 \cong \vect$ as monoidal categories via $N \mapsto \Cc(\Wc(0),N)$. We shall need $N \mapsto \Cc(N,\Wc(0))$ instead, which is a monoidal equivalence $\Cc_0^\mathrm{opp} \cong \vect$. Now
\be
  \Cc(\Wc(0) \otimes_f U, \Wc(0)) \cong \Cc(\Wc(0) \otimes_f U \otimes_f \Wc(0) , \Wc^*) \cong \Cc(U,\Wc(0)) \ ,
\ee
so that the composition $\Cc \xrightarrow{\Wc(0) \otimes_f (-)} \Cc_0 \xrightarrow{\sim} \vect$ is just $\Cc(-,\Wc(0))$. Since both functors are monoidal, we conclude that $\Cc(-,\Wc(0))$ is a monoidal functor $\Cc^\mathrm{opp} \to \vect$. Finally, for all $U,V \in \Cc$ we have 
\be\begin{array}{ll}
  \Cc(T(U \boxtimes V),\Wc(0)) &=~ \Cc(U \otimes_f V,\Wc(0)) ~\overset{(1)}{\cong}~ \Cc(U,\Wc(0)) \otimes_{\Cb} \Cc(V,\Wc(0)) \\[.2em]
  & \overset{(2)}{\cong}~ \Cc^2(U \boxtimes V, \Wc(0) \boxtimes \Wc(0)) \ ,
\end{array}
\ee  
where (1) is monoidality of $\Cc(-,\Wc(0))$ and (2) is \eqref{eq:hom-iso}. By lemma \ref{lem:nat-xfer-unique-on-prod}, this shows that for all $X \in \CoC$ we have $\Cc(T(X),\Wc(0))  \cong \Cc^2(X, \Wc(0) \boxtimes \Wc(0))$. Thus by definition of the adjoint, $R(\Wc(0)) = \Wc(0) \boxtimes \Wc(0)$. 

\subsection{Computation of $R(\Wc^*)$} \label{sec:RW*-socle-dec}

We now want to implement the general construction of section \ref{sec:boundary-to-bulk} for more interesting boundary theories than the $\Wc(0)$-example treated in the previous subsection. That is, we should fix an associative algebra $A \neq \Wc(0)$ in $\Cc$ which has a non-degenerate pairing and compute its full centre $Z(A) \in \CoC$ according to definition \ref{eq:full-centre-in-C2}. 

A point to stress is that neither $\Wc$ nor $\Wc^*$ are non-degenerate boundary theories as in definition \ref{def:non-deg-bnd}. They are both associative (and commutative) algebras (cf.\ lemma \ref{lem:SS=S-algebra}), but neither allows for a non-degenerate pairing. This is evident from the socle filtration \eqref{eq:WW*-compos}, as a necessary condition for a non-degenerate pairing on an algebra $A$ is that $A^* \cong A$ (see definition \ref{def:non-deg-pair}).

According to \cite[Thm.\,3.10]{Gaberdiel:2009ug}, one way to produce such an algebra $A \in \Cc$ is to take an object $U \in \Cc$ for which $U^*$ is the categorical dual and set $A = U \otimes_f U^*$. There are (recall the above disclaimer) many such objects to choose from. The original idea was to choose $A$ small in order to simplify the analysis, and one convenient choice which produces a particularly small $A$ is $U = \Wc(\tfrac58)$ (which is self-contragredient). From \cite[App.\,A.3\,\&\,A.4]{Gaberdiel:2009ug} we read off that $A = \Wc(\tfrac58) \otimes_f \Wc(\tfrac58)$ has socle filtration
\be \label{eq:small-bnd-object}
A~:~
\raisebox{2.4em}{\xymatrix@C=1em@R=1em{
  & 2 \ar@{-}[dr] \ar@{-}[dl] \ar@{-}[d]
\\
7 \ar@{-}[dr] &0\ar@{-}[d]&7\ar@{-}[dl]  
\\
& 2
}}
\qquad .
\ee
The next step would be to use expression \eqref{eq:R(U)-expression} and theorem \ref{thm:full-centre-via-left-centre} to obtain $Z(A)$ as the subobject $C_l( R(A) )$ of $R(A) = ((A^* \boxtimes \Wc) \otimes_f R(\Wc^*)^*)^*$. Unfortunately, we do not control the tensor product and braiding on $\Cc$ well enough to carry out this computation. 

\medskip

Instead, let us have a closer look at $\Wc^*$. As we already remarked, $\Wc^*$ is not a non-degenerate boundary theory, but it is still a (non-unital) associative algebra and hence a boundary theory with background states as alluded to in remark \ref{rem:bnd-bgstate}\,(ii).
The full centre $Z(\Wc^*)$ is a commutative associative algebra and provides a bulk theory with background states as defined in section \ref{sec:genout-ideal}. Indeed, it is by construction the maximal such theory that can be non-degenerately joined to the boundary theory $\Wc^*$ (this follows from \eqref{eq:terminal-obj-in-P}, table \ref{tab:alg-cft-table2} and remark \ref{rem:fullcentre}\,(i)). However, since $Z(\Wc^*)$ is obtained from a `non-standard' boundary theory, it is maybe not surprising that it will show some `non-standard' features itself; this will be discussed in section \ref{sec:R1-OPE} below.

\medskip

The first step towards $Z(\Wc^*)$ is to determine $R(\Wc^*)$. The method for this given in section \ref{sec:computeR1*} has been carried out (recall the above disclaimer) in \cite[Sect.\,2.2]{Gaberdiel:2010rg}. The result is as follows. As a $\Wc \otimes_{\Cb} \Wc$-representation, $R(\Wc^*)$ splits into 5 indecomposable summands,
\be
  R(\Wc^*) = {\cal H}_0 \oplus  {\cal H}_{\frac{1}{8}} \oplus {\cal H}_{\frac{5}{8}}
\oplus {\cal H}_{\frac{1}{3}} \oplus {\cal H}_{\frac{-1}{24}}  \oplus {\cal H}_{\frac{35}{24}}\ ,
\end{equation}
where we have labelled the individual blocks ${\cal H}_h$ by the conformal weight of the
lowest state. The blocks $\Hc_{\frac{-1}{24}}$ and $\Hc_{\frac{35}{24}}$ are irreducible and given by
\be 
 \Hc_{\frac{-1}{24}} = \Wc(\tfrac{-1}{24}) \otimes_\Cb \Wc(\tfrac{-1}{24}) \ , \qquad
 \Hc_{\frac{35}{24}} = \Wc(\tfrac{35}{24}) \otimes_\Cb \Wc(\tfrac{35}{24}) \ .
\ee
The remaining blocks are not irreducible. The socle filtration of ${\cal H}_{\frac{1}{8}}$ reads
\be
 {\cal H}_{\frac{1}{8}} ~ : \quad
 \begin{array}{rcl} 
   \Wc(\tfrac{1}{8}) \otimes_\Cb \Wc(\tfrac{1}{8})
 &\oplus& \Wc(\tfrac{33}{8}) \otimes_\Cb \Wc(\tfrac{33}{8})
\\
 &\downarrow&
\\
 2 \cdot \Wc(\tfrac{1}{8}) \otimes_\Cb \Wc(\tfrac{33}{8})
 &\oplus& 2 \cdot \Wc(\tfrac{33}{8}) \otimes_\Cb \Wc(\tfrac{1}{8})
\\
 &\downarrow&
\\
 \Wc(\tfrac{1}{8}) \otimes_\Cb \Wc(\tfrac{1}{8})
 &\oplus& \Wc(\tfrac{33}{8}) \otimes_\Cb \Wc(\tfrac{33}{8})\ .
\end{array}
\quad 
\ee  
This can be organised in a more transparent fashion if we 
replace each direct sum by a little table where we indicate the multiplicity of each term as in
\be\label{eq:3level_compseries} \small
\begin{tabular}{c|c|c|}
\raisebox{-7pt}{\rule{0pt}{20pt}} & $\tfrac{1}{8}$ & $\tfrac{33}{8}$ \\ 
\hline
\raisebox{-7pt}{\rule{0pt}{20pt}}  $\tfrac{1}{8}$  & 1 & 0 \\
\hline
\raisebox{-7pt}{\rule{0pt}{20pt}}$\tfrac{33}{8}$ & 0 & 1 \\
\hline
\end{tabular}
\quad \longrightarrow \quad
\begin{tabular}{c|c|c|}
\raisebox{-7pt}{\rule{0pt}{20pt}}& $\tfrac{1}{8}$ & $\tfrac{33}{8}$ \\ 
\hline
\raisebox{-7pt}{\rule{0pt}{20pt}} $\tfrac{1}{8}$  & 0 & 2 \\
\hline
\raisebox{-7pt}{\rule{0pt}{20pt}} $\tfrac{33}{8}$ & 2 & 0 \\
\hline
\end{tabular}
\quad \longrightarrow \quad
\begin{tabular}{c|c|c|}
\raisebox{-7pt}{\rule{0pt}{20pt}}& $\tfrac{1}{8}$ & $\tfrac{33}{8}$ \\ 
\hline
\raisebox{-7pt}{\rule{0pt}{20pt}} $\tfrac{1}{8}$  & 1 & 0 \\
\hline
\raisebox{-7pt}{\rule{0pt}{20pt}} $\tfrac{33}{8}$ & 0 & 1 \\
\hline
\end{tabular} \quad .
\ee
The socle filtrations of  ${\cal H}_{{5}/{8}}$ and ${\cal H}_{1/3}$ are  the same, but with $\{ \tfrac18, \tfrac{33}8 \}$ replaced by $\{ \tfrac58, \tfrac{21}8 \}$ and $\{ \tfrac13, \tfrac{10}3 \}$, respectively. The sector $\Hc_0$ is the most interesting, its socle filtration is
(all empty entries are equal to `0') 
\def\stab{\hspace*{-.6em}&\hspace*{-.4em}}
\begin{align}\label{eq:0_compseries}
\begin{tikzpicture}[>=latex,baseline=(F)]
  \node (A) at (0,0) {\scriptsize
  \begin{tabular}{c|c|c|c|c|c|}
  \multicolumn{6}{c}{Level 4}\\[.5em]
 &\hspace*{-.6em} 0 \stab 1 \stab 2 \stab 5 \stab 7 \hspace*{-.6em}\\
\hline
0 &\hspace*{-.6em}   1 \stab   \stab   \stab   \stab   \hspace*{-.6em}\\
\hline
1 &\hspace*{-.6em}   \stab 1 \stab   \stab   \stab   \hspace*{-.6em}\\
\hline
2 &\hspace*{-.6em}   \stab   \stab 1 \stab   \stab   \hspace*{-.6em}\\
\hline
5 &\hspace*{-.6em}   \stab   \stab   \stab 1 \stab   \hspace*{-.6em}\\
\hline
7 &\hspace*{-.6em}   \stab   \stab   \stab   \stab 1 \hspace*{-.6em}\\
\hline
\end{tabular}};
  \node (B) at (-5,0) {\scriptsize\begin{tabular}{c|c|c|c|c|c|}
  \multicolumn{6}{c}{Level 3}\\[.5em]
 &\hspace*{-.6em} 0 \stab 1 \stab 2 \stab 5 \stab 7 \hspace*{-.6em}\\
\hline
0 &\hspace*{-.6em}    \stab 1  \stab 1  \stab   \stab   \hspace*{-.6em}\\
\hline
1 &\hspace*{-.6em} 1 \stab  \stab   \stab  2 \stab 2  \hspace*{-.6em}\\
\hline
2 &\hspace*{-.6em} 1  \stab   \stab  \stab 2  \stab 2  \hspace*{-.6em}\\
\hline
5 &\hspace*{-.6em}   \stab  2 \stab 2  \stab  \stab   \hspace*{-.6em}\\
\hline
7 &\hspace*{-.6em}   \stab 2  \stab 2  \stab   \stab  \hspace*{-.6em}\\
\hline
\end{tabular}};
  \node (C) at (2.5,4) {\scriptsize\begin{tabular}{c|c|c|c|c|c|}
  \multicolumn{6}{c}{Level 2}\\[.5em]
 &\hspace*{-.6em} 0 \stab 1 \stab 2 \stab 5 \stab 7 \hspace*{-.6em}\\
\hline
0 &\hspace*{-.6em} 1 \stab   \stab   \stab 2 \stab 2 \hspace*{-.6em}\\
\hline
1 &\hspace*{-.6em}  \stab 2 \stab 4 \stab   \stab   \hspace*{-.6em}\\
\hline
2 &\hspace*{-.6em}  \stab 4 \stab 2 \stab  \stab  \hspace*{-.6em}\\
\hline
5 &\hspace*{-.6em} 2 \stab  \stab  \stab 2 \stab 4 \hspace*{-.6em}\\
\hline
7 &\hspace*{-.6em} 2 \stab  \stab  \stab 4 \stab 2 \hspace*{-.6em}\\
\hline
\end{tabular}};
  \node (D) at (-2.5,4) {\scriptsize\begin{tabular}{c|c|c|c|c|c|}
  \multicolumn{6}{c}{Level 1}\\[.5em]
 &\hspace*{-.6em} 0 \stab 1 \stab 2 \stab 5 \stab 7 \hspace*{-.6em}\\
\hline
0 &\hspace*{-.6em}    \stab 1  \stab 1  \stab   \stab   \hspace*{-.6em}\\
\hline
1 &\hspace*{-.6em} 1 \stab  \stab   \stab  2 \stab 2  \hspace*{-.6em}\\
\hline
2 &\hspace*{-.6em} 1  \stab   \stab  \stab 2  \stab 2  \hspace*{-.6em}\\
\hline
5 &\hspace*{-.6em}   \stab  2 \stab 2  \stab  \stab   \hspace*{-.6em}\\
\hline
7 &\hspace*{-.6em}   \stab 2  \stab 2  \stab   \stab  \hspace*{-.6em}\\
\hline
\end{tabular}};
  \node (E) at (-7.5,4) {\scriptsize\begin{tabular}{c|c|c|c|c|c|}
  \multicolumn{6}{c}{Level 0}\\[.5em]
 &\hspace*{-.6em} 0 \stab 1 \stab 2 \stab 5 \stab 7 \hspace*{-.6em}\\
\hline
0 &\hspace*{-.6em}  1  \stab   \stab   \stab   \stab   \hspace*{-.6em}\\
\hline
1 &\hspace*{-.6em}   \stab 1 \stab   \stab   \stab   \hspace*{-.6em}\\
\hline
2 &\hspace*{-.6em}   \stab   \stab 1 \stab   \stab   \hspace*{-.6em}\\
\hline
5 &\hspace*{-.6em}   \stab   \stab   \stab 1 \stab   \hspace*{-.6em}\\
\hline
7 &\hspace*{-.6em}   \stab   \stab   \stab   \stab 1 \hspace*{-.6em}\\
\hline
\end{tabular}};
\node (F) at (-5,2) {};
\node (G) at (2.5,2) {};
\draw[<-] (A) -- (B);
\draw[<-] (B) -- (F.center) -- (G.center) -- (C);
\draw[<-] (C) -- (D);
\draw[<-] (D) -- (E);
\node at (1.5,0) {\qquad .};
\end{tikzpicture}
\end{align}

Via \eqref{eq:F(A)-product}, $R(\Wc^*)$ inherits the structure of an associative algebra from $\Wc^*$. To find the full centre we should compute $Z(\Wc^*) = C_l(R(\Wc^*)) \subset R(\Wc^*)$. Again, the lack of detailed knowledge of the braiding means we currently cannot do this. However, we know from remark \ref{rem:trivial-twist} that $\exp\!\big( 2 \pi i ( L_0 -  \overline L_0) \big)$ acts as the identity on $R(\Wc^*)$. This implies that $m_{R} \circ c_{R,R} \circ c_{R,R} = m_R$ (abbreviating $R \equiv R(\Wc^*)$), i.e.\ taking one field all the way around another does not produce a monodromy. Our guess is that in fact $m_{R} \circ c_{R,R} = m_R$, i.e.\ $Z(\Wc^*) = R(\Wc^*)$, but as we already said, we cannot check this.

\medskip

Finally, recall that the functor $R(-)$ is lax monoidal (cf.\ section \ref{sec:computeR1*}) and so the algebra map $\Wc^* \to \Wc(0)$ gives an algebra map $\pi : R(\Wc^*) \to R(\Wc(0))$, which we expect to be non-zero. As a consequence, there is an OPE-preserving surjection from the tentative bulk theory $Z(\Wc^*) = R(\Wc^*)$ to the $c=0$ minimal model $\Wc(0) \boxtimes \Wc(0)$. In this sense, $R(\Wc^*)$ is a `refinement' of the minimal model.

\subsection{Modular invariance}\label{sec:W23-modinv}

A second interesting feature of this construction is that it leads to a modular invariant partition function. It is a straightforward exercise to write down the vector space of modular invariant bilinear combinations of the 13 characters of irreducible $\Wc$-representations. Namely, make the general ansatz
\be
  \xi(M,\tau) := \sum_{a,b} M_{ab} \, \chi_{\Wc(a)}(q) \, \chi_{\Wc(b)}(\bar q) 
  \qquad ; \quad q = e^{2\pi i \tau} \ ,
\ee
where $M$ is a $13{\times}13$-matrix and $a,b$ run over the lowest $L_0$-weights of the 13 irreducibles. Then the condition $\xi(M,\tau+1) = \xi(M,\tau)$ already forces most entries of $M$ to be zero. The known modular properties of the characters (see \cite{Feigin:2006iv}, or \cite[App.\,A.2]{Gaberdiel:2010rg} for the notation used here) turn $\xi(M,-1/\tau) = \xi(M,\tau)$ into a linear equation for $M$. In this way one finds that $\xi(M,\tau+1) = \xi(M,\tau) = \xi(M,-1/\tau)$ has a two-dimensional space of solutions given by (zeros are not written)
\be
M ~= \quad 
{\tiny
\begin{tabular}{c|c@{~}c@{~}c@{~}c@{~}c@{~}c@{~}c@{~}c@{~}c@{~}c@{~}c@{~}c@{~}c}
& $\Wc(0)$ & $\Wc(1)$ & $\Wc(2)$ & $\Wc(5)$ & $\Wc(7)$ & $\Wc(\tfrac{1}{3})$ & $\Wc(\tfrac{10}{3})$ & $\Wc(\tfrac{5}{8})$ & $\Wc(\tfrac{21}{8})$ & $\Wc(\tfrac{1}{8})$  & $\Wc(\tfrac{33}{8})$ & $\Wc(\tfrac{-1}{24})$ & $\Wc(\tfrac{35}{24})$ 
\\
\hline
\\[-.7em]
$\Wc(0)$ & $\alpha$ & $2\beta$ & $2\beta$ & $2\beta$ & $2\beta$ 
\\[.5em]
$\Wc(1)$ & $2\beta$ & $4\beta$ & $4\beta$ & $4\beta$ & $4\beta$
\\[.5em]
$\Wc(2)$ & $2\beta$ & $4\beta$ & $4\beta$ & $4\beta$ & $4\beta$
\\[.5em]
$\Wc(5)$ & $2\beta$ & $4\beta$ & $4\beta$ & $4\beta$ & $4\beta$
\\[.5em]
$\Wc(7)$ & $2\beta$ & $4\beta$ & $4\beta$ & $4\beta$ & $4\beta$
\\[.5em]
$\Wc(\tfrac{1}{3})$ & & & & & & $2\beta$ & $2\beta$
\\[.5em]
$\Wc(\tfrac{10}{3})$ & & & & & & $2\beta$ & $2\beta$
\\[.5em]
$\Wc(\tfrac{5}{8})$ & & & & & & & & $2\beta$ & $2\beta$ 
\\[.5em]
$\Wc(\tfrac{21}{8})$ & & & & & & & & $2\beta$ & $2\beta$ 
\\[.5em]
$\Wc(\tfrac{1}{8})$ & & & & & & & & & & $2\beta$ & $2\beta$ 
\\[.5em]
$\Wc(\tfrac{33}{8})$ & & & & & & & & & & $2\beta$ & $2\beta$ 
\\[.5em]
$\Wc(\tfrac{-1}{24})$ & & & & & & & & & & & & $\beta$ 
\\[.5em]
$\Wc(\tfrac{35}{24})$ & & & & & & & & & & & & & $\beta$
\end{tabular}
}
\ee
with $\alpha,\beta \in \Cb$. Summing up the entries of the tables in \eqref{eq:3level_compseries} and \eqref{eq:0_compseries} level by level, one quickly checks that the above space of solutions is spanned by the characters of $R(\Wc(0))$ and $R(\Wc^*)$. In particular, we see that the character $\chi_{R(\Wc^*)}(q,\bar q)$ is modular invariant.

\medskip

To relate the character $\chi_{R(\Wc^*)}(q,\bar q)$ to the partition function of $R(\Wc^*)$ we appeal to remark \ref{rem:sphere2}\,(iv) and section \ref{sec:mod-inv-pf}: The composition series \eqref{eq:3level_compseries} and \eqref{eq:0_compseries} suggest that $R(\Wc^*) \cong R(\Wc^*)^*$, i.e.\ that $R(\Wc^*)$ is self-conjugate. Therefore, assuming inversion invariance of $R(\Wc^*)$, the construction in remark \ref{rem:sphere2}\,(iv) provides us with non-degenerate two-point correlators on the Riemann sphere. According to section \ref{sec:mod-inv-pf} this allows one to express the torus amplitude as a trace over the space of states.

\medskip

The partition function of $R(\Wc^*)$ follows a pattern also observed in supergroup WZW models and the $W_{1,p}$-models \cite{Quella:2007hr,Gaberdiel:2007jv}, as well as in the study of modular properties of Hopf algebra modules \cite{Fuchs:2011mg} (cf.\ remark \ref{rem:FSS-Hopf}). Namely, despite the complicated submodule structure of $R(\Wc^*)$ as given in \eqref{eq:0_compseries}, in terms of characters we simply have
\be
  \chi_{R(\Wc^*)}(q,\bar q) ~=~ \sum_h \chi_{\Wc(h)}(q) \,\chi_{\Pc(h)}(\bar q) \ ,
\ee
where the sum is over the weights of the 13 irreducibles. 

In \cite{Pearce:2010pa} it has been argued that this bilinear combination of characters is modular invariant for all $W_{p,q}$-models.
Furthermore, it turns out that the function $\chi_{R(\Wc^*)}(q,\bar q)$ can -- up to a constant -- be written as a linear combination of modular invariant free boson partition functions at $c=1$ \cite{Pearce:2010pa}. In this form, $\chi_{R(\Wc^*)}(q,\bar q)$ has already appeared in the context of a model for dilute polymers \cite{Saleur:1991hk}.\footnote{
  More precisely, $\chi_{R(\Wc^*)}(q,\bar q) = Z_c[\frac32,1] + 3$, where for $Z_c[\frac32,1]$ we refer to 
  \cite[Eqn.\,(38)]{Saleur:1991hk} and for the relation to polymers to \cite[Sect.\,4.1.2]{Saleur:1991hk}. 
  We thank Hubert Saleur for a discussion on this point.
}

\medskip

\subsection{Correlators and OPEs in $R(\Wc^*)$} \label{sec:R1-OPE}

Finally, we want to explain the non-standard features of the putative bulk theory $R(\Wc^*)$ 
in more detail. In particular, we want to show that it does not have an identity field, nor
a stress energy tensor. (However, the correlation functions are still invariant under infinitesimal conformal transformations.)

In order to understand these features let us study the OPEs of the low-lying fields. It follows from the socle filtration in (\ref{eq:0_compseries}) that there are three states of generalised conformal dimension $(0,0)$. These are mapped into one another
under the action of the zero modes. Denoting the relevant states again by $\eta$, 
$\omega$ and $\Omega$, one would expect (as is also assumed in (\ref{5fold})) that the
relevant zero mode can be taken to be $L_0$ or $\overline{L}_0$. Since locality requires that 
$L_0-\overline{L}_0$ must be diagonalisable (cf.\ remark \ref{rem:trivial-twist}), we then conclude that 
\be
L_0 \eta = \overline L_0 \eta = \omega 
~~,\quad
L_0 \omega = \overline L_0 \omega = \Omega 
~~,\quad
L_0 \Omega = \overline L_0 \Omega = 0 \ .
\ee
We can again define quasiprimary states $t$ and $\Tc$ by  \eqref{eq:tT-def}, and likewise for $\overline{t}$ and $\overline{\Tc}$. It follows from \eqref{eq:0_compseries} that $\Tc$ is a holomorphic field since there is no primary field of generalised dimension $(2,1)$ in the
third or fourth level of the socle filtration and hence $\overline{L}_{-1} \Tc=0$.  By the same argument 
we also see that $\Omega$ is annihilated by all $L_n$ and $\overline{L}_n$ modes. On the other hand,
we cannot conclude that $t$ is holomorphic, since there is a $(2,1)$ state in level 2
of the socle filtration \eqref{eq:0_compseries}; this is indeed expected since 
the diagram (\ref{5fold}) still applies, and hence $t$ is the `logarithmic partner' of $\Tc$ (and 
thus should depend on both $z$ and $\bar{z}$). 

\subsubsection*{Some OPEs}

The derivation of the OPEs and correlators presented below can be found in appendix \ref{app:OPEs}, here we merely list the results. The simplest set of OPEs are those involving $\Omega$:
\be \label{eq:Om-phi-OPE}
  \Omega(z) \phi(w) = \pi(\phi) \cdot \Omega(w)
  \quad , ~~ \text{for all} ~ \phi \in F \ ,
\ee
and the OPE does not contain subleading terms.
Here $\pi$ is the intertwiner $R(\Wc^*) \to R(\Wc(0)) \equiv \Cb$ introduced in the previous subsection. The map $\pi$ is an algebra homomorphism, i.e.\ it is compatible with the OPE, and it is non-vanishing on the level 0 state $\eta$. We can normalise $\eta$ such that $\pi(\eta) = 1$. 
In particular,
\be
  \Omega(z)\Omega(w) = \Omega(z)\omega(w) = 0
  \quad , \quad
  \Omega(z) \eta(w) = \Omega(w) \ .
\ee
Since $\Omega$ is the only $sl(2,\Cb)$-invariant field in $R(\Wc^*)$, this shows that $R(\Wc^*)$ has no identity field. Next we list some OPEs involving $\Tc$:
\be \label{eq:RW*-OPEs-T}
\begin{array}{rcl}
\Tc(z)\, \omega(w) &=& \Oc\big((z{-}w)^0\big)
\quad , \quad \Tc(z) \Tc(0) ~=~  \Oc\big((z{-}w)^0\big) \ ,
\\[1em]
\Tc(z) \,\eta(w) &=&  \displaystyle
A \cdot \Big(~ \frac{\Omega(w)}{(z-w)^{2}}  + \frac{\tfrac{\partial}{\partial w}\omega(w)}{z-w} ~ \Big) ~+~ \Oc\big((z{-}w)^0\big) \ ,
\\[1em]
t(z) \, \Tc(w) &=& \displaystyle
(A+1) \cdot \Big(~
  \frac{-5\,\Omega(w)}{(z-w)^4} 
  + \frac{ 2 \Tc(w) }{(z{-}w)^2} 
  + \frac{\tfrac{\partial}{\partial w}\Tc(w)}{z-w} ~\Big) ~+~ \Oc\big((z{-}w)^0\big) \ ,
\end{array}
\ee
where $A \in \Cb$ is a so far undetermined constant.
From this we see that $\Tc$ -- the only holomorphic field of weight $(2,0)$ in the space of fields $F$ -- does not behave as the stress tensor. For example, it has regular OPE with itself. However, a glance at \eqref{eq:M-modeactions} shows that the OPE of the field $\hat{t}= \frac{1}{A+1} t$ with $\Tc$ can be written as
\be
\hat{t}(z) \, \Tc(w) ~=~ \sum_{n=-1}^{2} \frac{(L_n \Tc)(w)}{(z-w)^{-n-2} }
~+~ \Oc\big((z-w)^0\big) \ .
\ee
So in this OPE, $\hat t$ behaves as the stress tensor (but it is not the stress tensor as it is not holomorphic).

Finally, we give two more OPEs for fields of generalised weight $(0,0)$:
\be \label{eq:gen(00)-OPEs}
\begin{array}{rcl}
  \omega(z) \, \omega(w) &=& B \cdot \Omega(w) + \dots \quad ,
  \\[.5em]
  \omega(z) \, \eta(w) &=& \big( 2(A{-}B) \ln|z{-}w|^2 + C \big) \cdot \Omega(w) ~+~ \big(1{-}B{+}2A\big) \cdot \omega(w) + \dots
   \ ,
\end{array}
\ee
where $B,C \in \Cb$ are new constants which remain to be determined. The dots stand for terms which vanish for $|z-w| \to 0$ and which have no component of generalised weight $(0,0)$.

\subsubsection*{Some correlators}

Recall the intertwiner $\pi : R(\Wc^*) \to R(\Wc(0)) \equiv \Cb$ from above. By our normalisation $\pi(\eta) = 1$ and by \eqref{eq:project-to-c=0-minmod} we have
\be
     \langle \eta(z_1) \cdots \eta(z_n) \rangle = 1 \ .
\ee
These are the correlators of the $c=0$ minimal model. If a state from the kernel of $\pi$  is inserted, the correlator vanishes.

To obtain non-trivial correlators we have to allow background states as in section \ref{sec:genout-ideal}. For example, the normalisation condition ${}^\eta \langle \Omega(0) \rangle = 1$ and the OPE \eqref{eq:Om-phi-OPE} imply the correlators
\be
  {}^\eta\langle \phi_1(z_1) \cdots \phi_n(z_n) \Omega(w) \rangle
  = \pi(\phi_1) \cdot \pi(\phi_2) \cdots \pi(\phi_n)
\ee
for all $\phi_i \in F$, independent of the insertion points $z_i$ and $w$. Another example is\footnote{
  The constant $-5$ found here is reminiscent of the $b$-value in the correlator of the stress tensor and its logarithmic partner (but recall that $\Tc$ is not a stress tensor). The value $b=-5$ has recently been observed in certain logarithmic bulk theories with $c=0$  \cite{Vasseur:2011ud}.
  }
\be
  {}^t \langle \Tc(0) \rangle =   {}^{\Tc} \langle t(0) \rangle = -5 \ ,
\ee
which follows immediately from \eqref{eq:M-weight1-IP} together with ${}^\eta \langle \Omega(0) \rangle = 1$ which fixes $N=1$.
Finally, from the OPEs \eqref{eq:gen(00)-OPEs} we can directly read off the two-point correlators
\be
\begin{array}{ll}
{}^\omega\langle \omega(z) \omega(w) \rangle = 0 ~~,~~~ &{}^\omega\langle \eta (z) \omega(w) \rangle = 1{-}B{+}2A \ ,
\\[.5em]
{}^\eta\langle \omega(z) \omega(w) \rangle = B 
~~,~~~~ &
{}^\eta\langle \eta (z) \omega(w) \rangle =  2(A{-}B) \ln|z{-}w|^2 + C \ .
\end{array}
\ee

In summary, we have seen that $R(\Wc^*)$ does not have an identity field or a stress tensor. Consequently, $R(\Wc^*)$ does not allow for an OPE-preserving embedding 
$\Wc \otimes_{\Cb} \Wc \to R(\Wc^*)$ as one might have expected from a 
$\Wc \otimes_{\Cb} \Wc$-symmetric theory. Nonetheless, by definition the $n$-point correlators are $\Vir \oplus \Vir$-coinvariants (and also $\Wc \otimes_{\Cb} \Wc$-coinvariants). The above problems are closely related to the fact that the boundary theory $\Wc^\ast$ from which this construction starts only defines a boundary theory with background states, see  remark \ref{rem:bnd-bgstate}\,(ii). If one were to consider instead a usual non-degenerate boundary theory $A$ with identity field as in \eqref{eq:small-bnd-object}, one would expect that the corresponding full center $Z(A)$ is better behaved. In particular, the unit condition in theorem \ref{thm:fullcentre-alg-assoc+comm} gives then a non-zero OPE-preserving map $\Wc \otimes_{\Cb} \Wc \to Z(A)$ which we expect to be an embedding, so that $Z(A)$ would have an identity field and a stress tensor.

\bigskip

\noindent
{\bf Acknowledgements:} This paper is based on talks given by IR at the conferences \textsl{Conformal field theories and tensor categories} (13--17 June 2011, Beijing) and
\textsl{Logarithmic CFT and representation theory} (3--7 October 2011, Paris).
IR thanks the Beijing International Center for Mathematical Research at Peking University and the Centre Emile Borel at the Institut Henri Poincare in Paris for hospitality during one month stays around these conferences. MRG and SW thank the Centre Emile Borel for hospitality during the second conference.
The authors thank
Alexei Davydov,
J\"urgen Fuchs,
Azat Gainutdinov,
J\'er\^ome Germoni,
Antun Milas,
Victor Ostrik,
Hubert Saleur,
Christoph Schweigert
and
Alexei Semikhatov
for
helpful discussions.
The research of MRG is supported in part by the Swiss National Science Foundation.
IR is partially supported by the Collaborative Research Centre 676 `Particles, Strings and the Early Universe -- the Structure of Matter and Space-Time'.
SW is supported by the SNSF scholarship for prospective researchers and by 
the World Premier International Research Center Initiative (WPI Initiative), 
MEXT, Japan.

\appendix

\section{Conditions B1--B5} \label{app:B1-B5}

In this appendix we write out conditions (B1)--(B5) referred to in definition \ref{def:CFT-uhp}. Let $(F,M,\Omega^*;B,m,\omega^*;b)$ be as in that definition.

Below, we always take  $(x_1,\dots,x_m) \in \Rb^m{\setminus}\text{diag}$, $(z_1,\dots,z_n) \in \Hb^n{\setminus}\text{diag}$ and $\psi_i \in B$, $\phi_j \in F$. The integers $m,n$ are to be chosen such that all $U_{m,n}$ in the statement are defined (there has to be at least one field insertion; this field insertion can be a boundary field or a bulk field, i.e.\ $m,n \in \Zb_{\ge 0}$, $m+n >0$).
\begin{itemize}
\item[(B1)] $U_{m,n}$ is smooth in each argument from $\Rb$ and $\Hb$, and linear in each argument from $B$ and $F$.
\item[(B2)] $U_{m,n}$ is invariant under joint permutation of $\Rb^m$ and $B^m$ and $\Hb^n$ and $F^n$. Namely, for all $\sigma \in S_m$ and $\tau \in S_n$,
\be
\begin{array}{l}
  U_{m,n}(x_1,\dots,x_m,z_1,\dots,z_n,\psi_1,\dots,\psi_m,\phi_1,\dots,\phi_n) 
  \\[.5em]
  \hspace{2em} = 
    U_{m,n}(x_{\sigma(1)},\dots,x_{\sigma(m)},z_{\tau(1)},\dots,z_{\tau(n)},\psi_{\sigma(1)},\dots,\psi_{\sigma(m)},\phi_{\tau(1)},\dots,\phi_{\tau(n)}) \ .
\end{array}
\ee
\end{itemize}

Because there are three maps describing a short distance expansion in the defining data, namely $M, m, b$, there are three ways to link the $U_{k,l}$ for different $k,l$. These are listed in the next three conditions. We denote the canonical projection $\overline F \to \bigoplus_{d \le \Delta} F^{(d)}$ by $P_\Delta$ and the canonical projection $\overline B \to \bigoplus_{d \le h} B^{(d)}$ by $P_h$.
\begin{itemize}
\item[(B3a)] (\textsl{bulk OPE})
Suppose that $n \ge 2$ and that $|z_1-z_2| < |z_i-z_2|$ for all $i>2$ and $|z_1-z_2| < |x_j-z_2|$ for all $j$. Then we can take the OPE of $\phi_1(z_1)$ and $\phi_2(z_2)$, reducing the number of bulk fields by one:
\be\label{eq:BB-OPE-condition}
\begin{array}{l}
\displaystyle
  U_{m,n}\big(\dots,z_1,z_2,\dots,\phi_1,\phi_2,\dots \big) 
  \\[.5em]
\displaystyle
  = \lim_{\Delta \to \infty}
  U_{m,n-1}\big(\dots,z_2,\dots,P_\Delta \circ M_{z_1-z_2}(\phi_1 \otimes \phi_2),\dots \big) 
\end{array}
\ee
\item[(B3b)] (\textsl{boundary OPE})
Suppose that $m \ge 2$ and that $x_1>x_2$, and $|x_1-x_2| < |x_i-x_2|$ for all $i>2$ and $|x_1-x_2| < |z_j-x_2|$ for all $j$. Then we can take the OPE of $\psi_1(x_1)$ and $\psi_2(x_2)$, reducing the number of boundary fields by one:
\be\label{eq:bb-OPE-condition}
\begin{array}{l}
\displaystyle
  U_{m,n}\big(x_1,x_2,\dots,\psi_1,\psi_2,\dots \big) 
  \\[.5em]
\displaystyle
  = \lim_{h \to \infty}
  U_{m-1,n}\big(x_2,\dots,P_h \circ m_{x_1-x_2}(\psi_1 \otimes \psi_2),\dots \big) 
\end{array}
\ee
\item[(B3c)] (\textsl{bulk-boundary map})
Suppose that $n \ge 1$. Write $z_1 = x + iy$. Suppose further that $|x_i-x|>y$ for all $i$ and $|z_j-x|>y$ for all $j>0$. Then we can expand $\phi_1(z_1)$ in terms of boundary fields at $x$, exchanging one bulk field for one boundary field:
\be\label{eq:Bb-OPE-condition}
\begin{array}{l}
\displaystyle
  U_{m,n}\big(x_1,\dots,x_m,z_1,\dots,z_n,\psi_1,\dots,\psi_m,\phi_1,\dots,\phi_n\big) 
  \\[.5em]
\displaystyle
  = \lim_{h \to \infty}
  U_{m+1,n-1}\big(x,x_1,\dots,x_m,z_2,\dots,z_n,P_h \circ b_{y}(\phi_1), \psi_1, \dots,\psi_m,\phi_2,\dots,\phi_n\big) 
\end{array}
\ee
\end{itemize}

The relation between derivatives and $L_{-1}$ is as before,
\begin{itemize}
\item[(B4)] The $U_{m,n}$ satisfy
\be \begin{array}{l}
\frac{d}{dz_1} U_{m,n}\big(\dots,z_1,\dots,\phi_1,\dots \big)
\,=\,   U_{m,n}\big(\dots,z_1,\dots,L_{-1} \phi_1,\dots \big) \ ,
\\[.5em]
\frac{d}{d\bar z_1} U_{m,n}\big(\dots,z_1,\dots,\phi_1,\dots \big)
\,=\,   U_{m,n}\big(\dots,z_1,\dots, \overline L_{-1} \phi_1,\dots \big) \ , 
\\[.5em]
\frac{d}{dx_1} U_{m,n}\big(x_1,\dots,\psi_1,\dots \big)
\,=\, U_{m,n}\big(x_1,\dots, L_{-1} \psi_1,\dots \big) \ ,
\end{array}
\ee
where $\frac{d}{dz_1}$ and $\frac{d}{d\bar z_1}$ are complex derivatives, and $\frac{d}{dx_1}$ is a real derivative.
\end{itemize}

Let $f$ be a rational function on $\Cb \cup \{\infty\}$
which has poles at most in the set $\{ x_1, \dots , x_m$, $z_1,\dots,z_n$, $\bar z_1,\dots, \bar z_n \}$ and $\infty$, and which satisfies the growth condition $\lim_{\zeta \to \infty} \zeta^{-3} f(\zeta) = 0$. The expansion around each of these points is
\be
  f(\zeta) 
  = 
  \sum_{p=-\infty}^\infty f^{k}_p \cdot (\zeta - x_k)^{p+1}
  = 
  \sum_{p=-\infty}^\infty g^{+,k}_p \cdot (\zeta - z_l)^{p+1}
  = 
  \sum_{p=-\infty}^\infty g^{-,k}_p \cdot (\zeta - \bar z_l)^{p+1} \ .
\ee

\begin{itemize} 
\item[(B5)] 
For all $f$ as above,
\be \label{eq:uhp-mode-moving}
\begin{array}{ll}
\displaystyle
  \sum_{p=-\infty}^\infty  \Big\{
  \sum_{k=1}^m f^{k}_p
  U_{m,n}(\dots,\psi_1,\dots,\,L_p \psi_k\,,\dots,\psi_m,\phi_1,\dots,\phi_n) 
\\[.5em]
\displaystyle
  \hspace{4em} +
  \sum_{l=1}^n
  U_{m,n}(\dots,\psi_1,\dots,\psi_m,\phi_1,\dots,\,(g^{+,l}_p L_p +  g^{-,l}_p  \overline L_p ) \phi_l\,,\dots,\phi_n) 
  \Big\} ~=~ 0 \ .
\end{array}
\ee
As in (C5), only a finite number of summands in the sum over $p$ are non-zero. There is a corresponding condition with $L_p$ and $\overline L_p$ exchanged in the sum over bulk insertions.
\end{itemize}

The complicated looking set of conditions (B5) is obtained following the original argument in \cite{Cardy:1984bb}: The fact that the boundary condition preserves conformal symmetry means that the correlator on the UHP satisfies the same conditions as the `holomorphic part' of a bulk correlator with insertions at $\{ x_1,\dots,x_m, z_1,\dots,z_n , \bar z_1,\dots, \bar z_n \}$. In other words, an insertion at $z$ in the upper half plane is duplicated to an insertion at $z$ and $\bar z$. This prescription arises again from contour integration, as noted in the following remark.

\begin{remark} \label{rem:bnd-stress}
As in remark \ref{rem:bulk-stress}, one can replace (B5) by the stronger requirement that there should exist a stress tensor, that is, a field $T^\mathrm{bnd} \in B^{(2)}$ such that $m_x(T \otimes \psi) = \sum_{m=-\infty}^\infty x^{-m-2} \, L_m \psi$. The CFT on the complex plane $(F,M,\Omega^*)$ is then equally required to be equipped with a stress tensor $T,\overline T \in F^{(2)}$. The statement `the boundary condition respects conformal symmetry' means that the two components of the stress tensor in the bulk agree with the stress tensor on the boundary in the sense that
\be
\lim_{y \to 0} \langle T(x{+}iy) \cdots \rangle = \langle T^{\mathrm{bnd}}(x) \cdots \rangle = \lim_{y \to 0} \langle \overline T(x{+}iy) \cdots \rangle \ .
\ee
holds in all correlators. Define the meromorphic function $u(\zeta)$ on the complex plane as follows:
\be
  u(\zeta) = \begin{cases}
  \big\langle T(\zeta) \psi_1(x_1) \cdots \phi_1(z_1) \cdots \big\rangle & ; ~\mathrm{Im}(\zeta) \ge 0
  \\
  \big\langle \overline T(\bar \zeta) \psi_1(x_1) \cdots \phi_1(z_1) \cdots  \big\rangle & ; ~\mathrm{Im}(\zeta) < 0
  \end{cases}
\ee
The conditions \eqref{eq:uhp-mode-moving} arise from the contour integral $\frac{1}{2 \pi i} \oint  f(\zeta) \, u(\zeta)\, d\zeta = 0$, where the contour is a big circle enclosing $\{ x_1,\dots,x_m, z_1,\dots,z_n , \bar z_1,\dots, \bar z_n \}$. Then the contour is deformed to a union of small circles around each of the $x_i$, $z_i$, $\bar z_i$ and the OPEs of the stress tensor are substituted.
\end{remark}

\section{Further details on the algebraic reformulation}

\subsection{Proof of theorem \ref{eq:PF-means-fd-algebra}} \label{app:more-deligne}

\begin{proof}[Proof of theorem \ref{eq:PF-means-fd-algebra}]
If $\Ac \cong \Rep_{f.d.}(A)$, condition (PF) follows by taking $P=A$, seen as a right module over itself. For the converse,
pick a projective generator $P$. We can choose $A := \Ac(P,P)$ and define the functor $H: \Ac \to \Rep_{f.g.}(A)$ on objects and morphisms by 
\be
  U \mapsto \Ac(P,U) \quad , \quad \big[ U \xrightarrow{f} V \big]  \mapsto \big[P \xrightarrow{(-)} U \xrightarrow{f} V\big] \ .
\ee
The right action of $a \in A$ is given by $f \mapsto f \circ a$. Note that $H(P)=A$. The functor $H$ is
\\[.3em]
\nxt faithful: there exists a surjection $P^{\oplus m} \to U$ for some $m$, and so, given $g : U \to V$, if $\big[P \xrightarrow{s} U \xrightarrow{g} V\big] = 0$ for all $s$, then also $g=0$. 
\\[.3em]
\nxt full: We need to show that every linear map $\varphi : \Ac(P,U) \to \Ac(P,V)$ such that $\varphi(f) \circ a = \varphi(f \circ a)$ for all $a \in A$, is of the form $\varphi(f) = \psi \circ f$ for some $\psi : U \to V$. Let $P^{\oplus k} \xrightarrow{K} P^{\oplus m} \xrightarrow{s} U$ the first two steps of a projective resolution (thus $s$ is surjective and the image of $K$ is the kernel of $s$). Let $s_1,\dots,s_m$ be the restriction of $s$ to each summand. The pullback property along surjections,
\be
\raisebox{2em}{\xymatrix{
& P \ar@{-->}[dl]_{\exists\,a} \ar[d]^f
\\
P^{\oplus m} \ar[r]^{s} & U
}}
\ee
shows that all $f$ can be written as $\sum_i s_i \circ a_i$ for some $a_i \in A$. Thus $s_1,\dots,s_m$ generates $\Ac(P,U)$ as an $A$-module. Next, consider the cokernel diagram
\be
\raisebox{2em}{\xymatrix{
P^{\oplus k}  \ar[r]^K 
& P^{\oplus m} \ar[r]^s \ar[d]_{\sum_i \varphi(s_i)}
& U \ar@{-->}[dl]^{\exists!\,\psi}
\\ 
& V
}}
\qquad .
\ee
The cokernel property can be applied because, denoting by $K_{ij} : P \to P$ the components of $K$,  $\sum_i \varphi(s_i) \circ K_{ij} = \sum_i \varphi(s_i \circ K_{ij}) = 0$ as $\sum_i s_i \circ K_{ij} = 0$ for all $j$. The diagram then shows that $\varphi(s_i) = \psi \circ s_i$ for some $\psi : U \to V$. Since the $s_i$ generate $\Ac(P,U)$, this fixes $\varphi$ uniquely. 
\\[.3em]
\nxt essentially surjective: Let $M$ be a finite-dimensional $A$-module. Let $A^{\oplus k} \to A^{\oplus m} \to M$ be the first two steps of a projective (in fact: free) resolution. In other words, $\Ac(P,P^{\oplus k}) \xrightarrow{\varphi} \Ac(P,P^{\oplus m}) \to M \to 0$ is exact for some $A$-module map $\varphi$. By fullness, there is a $\psi : P^{\oplus k} \to P^{\oplus m}$ such that $\varphi = \psi \circ (-)$. Since $P$ is projective, the functor $\Ac(P,-)$ is exact, and so the exact sequence $P^{\oplus k} \xrightarrow{\psi} P^{\oplus m} \xrightarrow{\cok(\psi)} U \to 0$ gets mapped to the exact sequence $\Ac(P,P^{\oplus k}) \xrightarrow{\psi \circ (-)} \Ac(P,P^{\oplus m}) \to \Ac(P,U) \to 0$. Thus $M \cong \Ac(P,U)$ for some $U$.
\end{proof}

\subsection{Idempotent absolutely simple objects are algebras}\label{app:idempot-alg}

\begin{lemma} \label{lem:SS=S-algebra}
Let $\Mc$ be a $k$-linear monoidal category and let $S \in \Mc$ be such that $S \otimes S \cong S$ and $\Cc(S,S) = k\cdot \id_S$. Pick an isomorphism $m : S \otimes S \to S$.
\\[.3em]
(i) The associator $\alpha_{S,S,S} : S\otimes(S\otimes S) \to (S\otimes S)\otimes S$ is $\alpha_{S,S,S} = (m^{-1} \otimes \id_S) \circ (\id_S \otimes m)$.
\\[.3em]
(ii) If $\Mc$ is in addition braided, the braiding on $S$ is $c_{S,S} = \id_{S \otimes S}$.
\\[.3em]
(iii) The pair $(S,m)$ is an associative algebra in $\Mc$. If $\Mc$ is braided, this algebra is commutative.
\\[.3em]
(iv) If $n : S \otimes S \to S$ is an isomorphism, then $(S,m)$ and $(S,n)$ are isomorphic as algebras.
\end{lemma}

\begin{proof} We will omit `$\otimes$' between objects for better readability.
\\[.3em]
(i) The space $\Mc(S(S S),(S S) S)$ is isomorphic to $\Mc(S,S)$ and hence one-dimensional. Therefore, there has to exist a $\lambda \in k^\times$ such that
\be
  \alpha_{S,S,S} = \lambda \cdot 
  \Big[ S (S S) \xrightarrow{\id_S \otimes m} SS \xrightarrow{m^{-1} \otimes \id_S} (SS)S \Big]  \ .
\ee
Naturality of the associator implies
\be
  \Big[ U (V W) \xrightarrow{f \otimes (g \otimes h)} S (SS) \xrightarrow{\alpha_{S,S,S}} (SS)S \Big]
  =
  \Big[ U (V W) \xrightarrow{\alpha_{U,V,W}}  (U V)W \xrightarrow{(f \otimes g) \otimes h}  (SS)S \Big]
\ee
Applying this to $f = m$, $g=h=\id_S$, etc., allows one to solve for $\alpha$ with one entry being $S S$. The result is
\be
\begin{array}{l}
  \alpha_{SS,S,S} = \lambda \cdot (m^{-1} \otimes \id_S \otimes m) = \alpha_{S,S,SS} \ ,
  \\[.5em]
  \alpha_{S,SS,S} = \lambda \cdot \big\{ (\id_S \otimes m^{-1}) \circ m^{-1} \big\} \otimes \big\{ m \circ (m \otimes \id_S) \big\} \ .
\end{array}
\ee
The pentagon with all four objects set to $S$ reads
\be
\alpha_{SS,S,S} \circ \alpha_{S,S,SS}
=
(\alpha_{S,S,S} \otimes \id_S) \circ \alpha_{S,SS,S} \circ (\id_S \otimes \alpha_{S,S,S}) \ .
\ee
Substituting the expressions in terms of $\lambda$ and $m$ one quickly checks that the above identity simplifies to $\lambda^2 \cdot u = \lambda^3 \cdot u$, with $u = \big\{ (m^{-1} \otimes \id_S) \circ m^{-1} \big\} \otimes \big\{ m \circ (\id_S \otimes m) \big\} \neq 0$. Thus, $\lambda = 1$.
\\[.3em]
(ii) By assumption $\Mc(SS,SS)$ is one-dimensional, and hence there has to be an $\omega \in k^\times$ such that $c_{S,S} = \omega \cdot \id_{SS}$. By naturality,
\be
  \Big[ U V \xrightarrow{f \otimes g} S S \xrightarrow{c_{S,S}} S S \Big]
  =
  \Big[ U V \xrightarrow{c_{U,V}} V U \xrightarrow{g \otimes f} S S \Big] \ ,
\ee
and applying this to $f=m$ and $g=\id_S$ we can solve for $c_{SS,S}$. The result is $c_{SS,S} = \omega \cdot (m^{-1} \otimes \id_S) \circ (\id_S \otimes m)$. One of the two hexagons with all objects set to $S$ reads
\be
  \alpha_{S,S,S} \circ c_{SS,S} \circ \alpha_{S,S,S} = (c_{S,S} \otimes \id_S) \circ \alpha_{S,S,S} \circ (\id_S \otimes c_{S,S})
\ee
Substituting the expressions for $\alpha_{S,S,S}$ from (i) and $c_{S,S}$, $c_{SS,S}$ as above, this reduces to $\omega \cdot v = \omega^2 \cdot v$ with $v = (m^{-1} \otimes\id_S) \circ (\id_S \otimes m)$. Thus $\omega=1$.
\\[.3em]
(iii) Associativity is $m \circ (\id_S \otimes m) = m \circ (m \otimes \id_S) \circ \alpha_{S,S,S}$, which holds by (i), and commutativity is trivial as $c_{S,S} = \id_{SS}$ by (ii).
\\[.3em]
(iv) Since $\Mc(SS,S)$ is one-dimensional, we have $n = \lambda m$ for some $\lambda \in k^\times$. Take $f = \lambda \,\id_S$. Then $f \circ n = m \circ (f \otimes f)$.
\end{proof}

\subsection{Proof of theorem \ref{thm:fullcentre-alg-assoc+comm}} \label{app:fullcentre-alg-assoc+comm}

\begin{proof}[Proof of theorem \ref{thm:fullcentre-alg-assoc+comm}]
\textsl{Existence:}
Consider the composition (the left path in \eqref{eq:z-is-algebramap})
\be
  w := \Big[ T(Z \otimes_{\Cc^2} Z) \xrightarrow{T_{2;Z,Z}^{-1}} T(Z) \otimes_{\Cc} T(Z)
  \xrightarrow{z \otimes_{\Cc} z} A \otimes_{\Cc} A \xrightarrow{\mu_A} A \Big] \ .
\ee
We need to check that the pair $(Z \otimes_{\Cc^2} Z,w)$ satisfies condition \eqref{eq:full-cent-univprop}, i.e.\ that it is an object in $\Cc_{\text{full\,center}}(A)$. This amounts to commutativity of (brackets, associators and `$\otimes_{\Cc}$' are not written)
\be \label{eq:TZZ-in-Cent}
\raisebox{4em}{\xymatrix@C=3.5em{
T(Z \otimes_{\Cc^2} Z) \, A \ar[dd]^{\varphi_{Z \otimes Z,A}} \ar[r]^{T_2^{-1} \otimes \id_A}
& T(Z) \, T(Z) \, A \ar[r]^-{z \otimes z \otimes \id_A}  \ar[d]^{\id_{T(Z)} \otimes \varphi_{Z,A}}
& A \, A \, A \ar[r]^{\mu_A \otimes \id_A}
& A \, A \ar[dr]^{\mu_A}
\\
& T(Z) \, A \, T(Z) \ar[d]^{\varphi_{Z,A} \otimes  \id_{T(Z)} } &&& A
\\
A \, T(Z \otimes_{\Cc^2} Z) \ar[r]^{\id_A \otimes T_2^{-1}}
& A \,T(Z) \, T(Z) \ar[r]^-{\id_A \otimes z \otimes z} 
& A \, A \, A \ar[r]^{\id_A \otimes \mu_A}
& A \, A \ar[ur]_{\mu_A}
}}
\quad .
\ee
The left subdiagram is just \eqref{eq:phi-properties}, while the details for the right subdiagram are obtained by copying out the corresponding diagram in the proof of \cite[Prop.\,4.1]{Davydov:2009} in the present setting; we omit the details. 

By the universal property of $(Z,z)$, there exists a unique morphism $Z \otimes_{\Cc^2} Z \to Z$ such that \eqref{eq:full-cent-pairmap} commutes. We define this morphism to be $\mu_Z$. 
\\[.3em]
\textsl{Commutativity:} 
We will show below that $c_{Z,Z}$ is an arrow from $(Z \otimes_{\Cc^2} Z,w)$ to itself in $\Cc_{\text{full\,center}}(A)$. This provides us with two arrows from $(Z \otimes_{\Cc^2} Z,w)$ to $(Z,z)$ in $\Cc_{\text{full\,center}}(A)$, namely $\mu_Z$ and $\mu_Z \circ c_{Z,Z}$. By uniqueness, they have to be equal, establishing commutativity.
\\
That $c_{Z,Z}$ is an endomorphism of $(Z \otimes_{\Cc^2} Z,w)$ amounts to commutativity of the diagram
\be
\raisebox{2.4em}{\xymatrix@R=.7em@C=2.5em{
  T(Z \otimes_{\Cc^2} Z)  \ar[dd]_{T(c_{Z,Z})} \ar[r]^{T_2^{-1}}
  & T(Z) \otimes_{\Cc} T(Z) \ar[r]^-{\id_{T(Z)} \otimes z} \ar[dd]_{\hat\varphi_{Z,Z} = \varphi_{Z,T(Z)}}
  & T(Z) \otimes_{\Cc} A \ar[r]^-{z \otimes \id_A} \ar[dd]_{\varphi_{Z,A}}
  & A \otimes A \ar[dr]^-{\mu_A}
\\
&&&&A
\\
  T(Z \otimes_{\Cc^2} Z)  \ar[r]^{T_2^{-1}}
  & T(Z) \otimes_{\Cc} T(Z) \ar[r]^-{z \otimes \id_{T(Z)}}
  & A \otimes_{\Cc} T(Z) \ar[r]^-{\id_A \otimes z}
  & A \otimes A\ar[ur]_-{\mu_A}
}}
\quad .
\ee
Starting from the left, the first square commutes by definition \eqref{eq:varphi-hat-def} of $\hat \varphi_{Z,Z}$. By lemma \ref{lem:phi-hatphi-rel}, this is equal to $\varphi_{Z,T(Z)}$. The second square is then just naturality of $\varphi_{Z,T(Z)}$. The third square is property \eqref{eq:full-cent-univprop} for $z$.
\\[.3em]
\textsl{Associativity:} In the proof of associativity, we will not write out tensor product symbols and brackets between objects, and we omit all associators. We will show the equality of the two maps $a = \mu_Z \circ (\mu_Z \otimes \id_Z)$ and $b = \mu_Z \circ (\id_Z \otimes \mu_Z)$ from $ZZZ$ to $Z$ via the terminal object property. Define the map
\be
   y := \Big[ T(Z\,Z\,Z) \xrightarrow{~\sim~} T(Z)\,T(Z)\,T(Z) \xrightarrow{z \otimes z \otimes z}  A\,A\,A \xrightarrow{~\text{mult.}~} A \Big] \ ,
\ee
where the first isomorphism is constructed from $T_2$ and associators, and `mult.' stands for any order of multiplying the three factors via $\mu_A$. That $y \in \Cent(ZZZ,A)$ is checked by an analogous argument as that giving commutativity of \eqref{eq:TZZ-in-Cent}. We now need to verify that $a$ and $b$ are maps from $(ZZZ,y)$ to $(Z,z)$. This will imply $a=b$ and hence associativity of $\mu_Z$. That $T(a) : T(ZZZ) \to T(Z)$ makes \eqref{eq:full-cent-univprop} commute amounts to commutativity of 
\be
\raisebox{3.5em}{\xymatrix{
T(Z\,Z\,Z) \ar[rr]^{T(\mu_Z \otimes \id_Z)} \ar[d]^\sim
&& T(Z\,Z) \ar[r]^{T(\mu_Z)} \ar[d]^\sim
& T(Z) \ar@{=}[d]
\\
T(Z)\,T(Z)\,T(Z) \ar[rr]^-{\mu_{TZ} \otimes \id_{TZ}} \ar[d]^{z \otimes z \otimes z}
&& T(Z)\,T(Z) \ar[r]^-{\mu_{TZ}} \ar[d]^{z \otimes z}
& T(Z) \ar[d]^z
\\
A\,A\,A \ar[rr]^{\mu_A \otimes \id_A}
&& A\,A \ar[r]^{\mu_A}
& A
}}
\quad \raisebox{-3.5em}{.}
\ee
The top two squares commute by definition of $\mu_{TZ}$ in \eqref{eq:F(A)-product}, the bottom two squares commute because $z$ is an algebra map (since it satisfies \eqref{eq:z-is-algebramap}). The argument for $T(b)$ is similar.
\\[.3em]
\textsl{Unitality:} The construction of the unit for $Z$ rests on the observation that
\be \label{eq:varphi_1U}
  \varphi_{\one,U} = \Big[ T(\one) \, U \xrightarrow{T_0^{-1} \otimes \id_U} \one \, U \xrightarrow{\lambda_U} U \xrightarrow{\rho_U^{-1}} U \, \one \xrightarrow{\id_U \otimes T_0} U\,T(\one) \Big] \ ,
\ee
which can be checked directly from \eqref{eq:mixed-braid-tildephi}. Define the map
\be
  u := \Big[ T(\one) \xrightarrow{T_0^{-1}} \one \xrightarrow{\iota_A} A \Big] \ .
\ee
To see that $u \in \Cent(\one,A)$, we need to establish commutativity of
\be
\raisebox{3em}{\xymatrix@R=1.5em{
T(\one)\,A \ar[rr]^{T_0^{-1}\otimes\id_A} \ar[dd]_{\varphi_{\one,A}}
&& \one\,A \ar[r]^{\iota_A \otimes \id_A} \ar[dr]_{\lambda_A}
& A \, A \ar[d]^{\mu_A}
\\
&&& A 
\\
A\,T(\one) \ar[rr]_{\id_A \otimes T_0^{-1}}
&& A \, \one \ar[r]_{\id_A \otimes \iota_A} \ar[ur]^{\rho_A}
& A\,A \ar[u]_{\mu_A}
}}
\qquad .
\ee
The pentagon is \eqref{eq:varphi_1U} and the remaining triangles amount to the unit property of $\iota_A$. Thus there exists a unique $\iota_Z : \one \to Z$ such that $u = z \circ T(\iota_Z)$. The unit property of $\iota_Z$ follows by verifying that $\mu \circ (\id_Z \otimes \iota_Z) \circ \rho_Z^{-1}$, $\mu \circ (\iota_Z \otimes \id_Z) \circ \lambda_Z^{-1}$ and $\id_Z$ are morphisms $Z \to Z$ in the category $\Cc_{\text{full\,center}}(A)$ and hence are all equal. We refer to \cite[Prop.\,4.1]{Davydov:2009} for details.
\end{proof}

\subsection{Proofs for theorems \ref{thm:R1*-via-projectives} and \ref{thm:trivial-twist}} \label{app:more-on-R1}

The proof of theorem \ref{thm:R1*-via-projectives} requires three lemmas. The first one gives an alternative characterisation of a representing object.

\begin{lemma} \label{lem:repres+terminal}
Let $U \in \Cc$, $R' \in \CoC$ and $r' : T(R') \to U$. The following are equivalent:
\begin{itemize}
\item[(i)] The object $R'$ represents the functor $\Cc(T(-),U)$ such that the natural isomorphism $\Cc(T(-),U) \to \Cc^2(-,R')$ maps $r'$ to $\id_{R'}$.
\item[(ii)] The pair $(R',r')$ satisfies the following universal property: For all pairs $(X,x)$ with $X \in \CoC$ and $x : T(X) \to U$, there exists a unique morphism $\tilde x : X \to R'$ such that the diagram
\be \label{eq:R'-terminal-ob-prop}
\raisebox{2em}{\xymatrix@C=1em@R=1em{
T(X) \ar[rr]^{T(\tilde x)} \ar[dr]_{x} && T(R') \ar[dl]^{r'}
\\
& U
}}
\ee
commutes.
\end{itemize}
\end{lemma}

\begin{proof}
(i)$\,\Rightarrow\,$(ii): Denote the natural isomorphism by $\chi_- : \Cc(T(-),U) \to \Cc^2(-,R')$. Naturality amounts to the following two equivalent identities, for all $f : X \to Y$, $y : T(Y) \to U$, and for $b = \chi_Y(y)$,
\be \label{eq:repres+terminal-aux1}
  \chi_X(y \circ T(f)) = \chi_Y(y) \circ f 
  \quad , \quad
  \chi_Y^{-1}(b) \circ T(f) = \chi^{-1}_X(b \circ f)  \ .
\ee
Suppose we are given $(X,x)$. We need to show existence and uniqueness of $\tilde x$.
\\[.3em]
\textsl{Existence:} Choose $\tilde x = \chi_X(x)$. Commutativity of \eqref{eq:R'-terminal-ob-prop} follows since $r' \circ T(\tilde x) = \chi^{-1}_{R'}(\id_{R'}) \circ T(\chi_X(x)) = \chi_{X}^{-1}(\id_{R'} \circ \chi_X(x)) = x$.
\\[.3em]
\textsl{Uniqueness:} Suppose \eqref{eq:R'-terminal-ob-prop} holds for some $a : X \to R'$ in place of $\tilde x$, i.e.\ $r' \circ T(a) = x$. 
By naturality, $r' \circ T(a) = \chi^{-1}_{R'}(\id_{R'}) \circ T(a) = \chi_X^{-1}(a)$. Thus $\chi_X^{-1}(a) = x$, which is equivalent to $a = \chi_X(x)$.
\\[.3em]
(ii)$\,\Rightarrow\,$(i): Given $x : T(X) \to U$, we define the map $\chi_X : \Cc(T(X),U) \to \Cc^2(X,R')$ to be $\chi_X(x) = \tilde x$. By uniqueness of $\tilde x$, this is well-defined. Since for $(X,x) = (R',r')$ we can choose $\tilde x = \id_{R'}$, the collection of maps $\chi_-$ satisfies $\chi_{R'}(r') = \id_{R'}$, as required. It remains to see that $\chi_X$ is a bijection for each $X$ and that it is natural in $X$.
\\[.3em]
\textsl{Naturality:} 
We will check the first identity in \eqref{eq:repres+terminal-aux1}. By uniqueness of $\tilde x$ in \eqref{eq:R'-terminal-ob-prop} it is enough to check that also $\chi_Y(y) \circ f$ provides an arrow from $(X,y \circ T(f))$ to $(R',r')$, i.e.\ that the diagram
\be
\raisebox{2.5em}{\xymatrix@R=1em@C=1.2em{
T(X) \ar[dr]_{T(f)} \ar[rr]^{T(f)} && T(Y) \ar[rr]^{T(\chi_Y(y))} \ar@{=}[dl] && T(R') \ar[ddll]^{r'} 
\\
& T(Y) \ar[dr]_{y}
\\
&& U
}}
\ee
commutes, which it does by definition of $\chi_Y$.
\\[.3em]
\textsl{Surjectivity:} Given $a : X \to R'$, by naturality and $\chi_{R'}(r') = \id_{R'}$ one has $\chi_X(r' \circ T(a)) = \chi_X(r') \circ a = a$.
\\[.3em]
\textsl{Injectivity:} Suppose $\chi_X(x) = 0$. Then by definition also $x = r' \circ T(0) = 0$.
\end{proof}

The second lemma allows one to rewrite any pairing in terms of the canonical non-degenerate pairings defined in \eqref{eq:def-beta-pairing}.

\begin{lemma} \label{lem:conjugate-prop}
Let $p : U \otimes V \to \one^*$. 
\\[.3em]
(i) There exist unique maps $f : U \to V^*$ and $g : V \to U^*$ such that $p = \beta_{U} \circ (\id_U \otimes g)$ and $p = \beta_{V^*} \circ (f \otimes \delta_V)$.
\\[.3em]
(ii) For all $h : U \to V$ we have $\beta_V \circ (h \otimes \id_{V^*}) = \beta_U \circ (\id_U \otimes h^*)$.
\end{lemma}

\begin{proof}
Both parts follow from naturality of $\delta$ and $\pi$ in condition (C). The latter amounts to the statement that for all $a :  X \to U$, $b : Y \to V$ and $q : U \to V^*$,
\be \label{eq:pi-natural-aux}
  \pi_{X,V}(q \circ a) = \pi_{U,V}(q) \circ (a \otimes \id_V)
  \quad , \quad
  \pi_{U,Y}(b^* \circ q) = \pi_{U,V}(q) \circ (\id_U \otimes b) \ .
\ee
For part (ii) we compute 
$\pi_{V,V^*}(\delta_V)\circ (h \otimes \id_{V^*})
= \pi_{U,V^*}(\delta_V \circ h)
= \pi_{U,V^*}(h^{**} \circ \delta_U)
= \pi_{U,U^*}(\delta_U) \circ (\id_U \otimes h^*)$. For part (i) set $f = \pi_{U,V}^{-1}(p)$ and $g = f^* \circ \delta_V$. Then
\be
\begin{array}{ll}
\beta_{U} \circ (\id_U \otimes g)
&= \pi_{U,U^*}(\delta_U) \circ (\id_U \otimes (f^* \circ \delta_V))
\overset{(1)}{=} \pi_{U,V}((f^* \circ \delta_V)^* \circ \delta_U) 
\\
&= \pi_{U,V}((\delta_V)^* \circ f^{**} \circ \delta_U) 
\overset{(2)}{=} \pi_{U,V}((\delta_V)^* \circ \delta_{V^*} \circ f) 
\\
&\overset{(3)}{=} \pi_{U,V}(f) = p \ ,
\end{array}
\ee
where (1) is naturality of $\pi$, (2) is naturality of $\delta$ and (3) is $(\delta_V)^* = (\delta_{V^*})^{-1} $, which is required by condition (C). The identity $\beta_{V^*} \circ (f \otimes \delta_V) = p$ is checked along the same lines (use naturality to move $f$ and $\delta$ inside $\pi$). Uniqueness of $f$ and $g$ is implied by non-degeneracy of $\beta_U$ and $\beta_{V^*}$, see the text below definition \ref{def:non-deg-pair}.
\end{proof}

For the third lemma, recall the space $N$, the basis $\{ u_1, \dots ,u_{|N|} \}$ of $N$ and the map $n = \sum_i u_i \circ \pi_i$ defined in section \ref{sec:computeR1*}. 

\begin{lemma} \label{lem:factor-through-n}
For any $f : (P \boxtimes P)^{\oplus m} \to P \boxtimes P^*$ such that $\beta_P \circ T(f) = 0$, there exists a (typically non-unique) $\varphi : (P \boxtimes P)^{\oplus m} \to (P \boxtimes P)^{\oplus |N|}$ such that
\be
\raisebox{2em}{\xymatrix{
(P \boxtimes P)^{\oplus m} \ar@{-->}[dr]_{\exists \varphi} \ar[rr]^f && P \boxtimes P^* 
\\
& (P \boxtimes P)^{\oplus |N|} \ar[ur]_n
}}
\ee
commutes.
\end{lemma}

\begin{proof}
Denote by
\be
\raisebox{.35em}{\xymatrix{
(P \boxtimes P)^{\oplus m} \ar@<1ex>[r]^-{p_i} & P \boxtimes P \ar@<1ex>[l]^-{e_i}
}}
\quad , \quad
\raisebox{.35em}{\xymatrix{
(P \boxtimes P)^{\oplus |N|} \ar@<1ex>[r]^-{\pi_i} & P \boxtimes P \ar@<1ex>[l]^-{\iota_i}
}}
\ee
the embedding and projection maps of the two direct sums. Let $f_j = f \circ e_j$. By assumption, $f_j \in N$. Thus we can write $f_j = \sum_{i=1}^{|N|} A_{ij} \, u_i$ for some $A_{ij} \in k$. Define
\be
  \varphi = \sum_{i=1}^{|N|} \sum_{j=1}^m A_{ij} \cdot \iota_i \circ p_j ~:~ (P \boxtimes P)^{\oplus m} \to (P \boxtimes P)^{\oplus |N|}  \ .
\ee
Then indeed $n \circ \varphi = \sum_{i,j,k} A_{ij} \,u_k \circ \pi_k \circ \iota_i \circ p_j = \sum_{i,j} A_{ij} \,u_i \circ p_j = \sum_j f_j \circ p_j = f$.
\end{proof}

\begin{proof}[Proof of theorem \ref{thm:R1*-via-projectives}]
We will show that $R'$ satisfies condition (ii) in lemma \ref{lem:repres+terminal} (with $U=\one^*$). Namely, suppose we are given a pair $(X,x)$ with $X \in \CoC$ and $x : T(X) \to \one^*$. We need to show that there exists a unique $\tilde x : X \to R'$ such that $x = r' \circ T(\tilde x)$.

\medskip\noindent
\nxt \textsl{Existence:} Let
\be
  (P \boxtimes P)^{\oplus k} \xrightarrow{~K~} (P \boxtimes P)^{\oplus m} \xrightarrow{~\cok(K)~} X
\ee
be the first two steps of a projective resolution of $X$. That is, we have a surjection $s : (P \boxtimes P)^{\oplus m} \to X$ whose kernel is the image of $K : (P \boxtimes P)^{\oplus k} \to (P \boxtimes P)^{\oplus m}$ (and hence $s = \cok(K)$). Define
\be
  p = \Big[ T((P \boxtimes P)^{\oplus m}) \xrightarrow{~T(\cok(K))~} TX \xrightarrow{~x~} \one^* \Big] \ .
\ee
Let $\pi_i : (P \boxtimes P)^{\oplus m} \to P \boxtimes P$ be the projection to the $i$'th summand. Then $T(\pi_i) : T((P \boxtimes P)^{\oplus m})  \to P \otimes_{\Cc} P$ and if we can define $p_i$ via
\be
  p = \sum_i \Big[ T((P \boxtimes P)^{\oplus m}) \xrightarrow{~T(\pi_i)~} P \otimes_{\Cc} P \xrightarrow{~p_i~} \one^* \Big] \ .
\ee
By lemma \ref{lem:conjugate-prop}\,(i), there exists a $q_i : P \to P^*$ such that $p_i  = \big[ P \otimes_{\Cc} P \xrightarrow{ \id \otimes q_i } P \otimes_{\Cc} P^* \xrightarrow{\beta_P} \one^* \big]$. Define $\tilde p := \sum_i (\id_P \boxtimes q_i) \circ \pi_i$. Then
\be \label{eq:projectives-aux2}
\raisebox{2em}{\xymatrix{
T((P \boxtimes P)^{\oplus m}) \ar[r]^-{T(\tilde p)} \ar[d]_{T(\cok(K))} \ar[dr]^p
& P \boxtimes P^* \ar[d]^{\beta_P}
\\
TX \ar[r]^x & \one^*
}}
\ee
commutes by construction. It follows that $\beta_P \circ T(\tilde p \circ K) = x \circ T(\cok(K) \circ K) = 0$. From lemma \ref{lem:factor-through-n} we get a map $u$ such that subdiagram (1) in the following diagram commutes:
\be \label{eq:projectives-aux1}
\raisebox{2em}{\xymatrix{
(P \boxtimes P)^{\oplus k} \ar[r]^{K} \ar[d]_{u} \ar@{}[dr]|{(1)}
& (P \boxtimes P)^{\oplus m} \ar[d]^{\tilde p}  \ar[rr]^{\cok(K)} \ar@{}[drr]|{(2)}
&& X \ar@{-->}[d]^{\exists! \, \tilde x}
\\
(P \boxtimes P)^{\oplus |N|} \ar[r]^{n}  
& P \boxtimes P^* \ar[rr]^{\cok(n)}
&& R'
}}
\ee
The existence of $u$ implies that $\cok(n) \circ \tilde p \circ K = 0$, so that by the universal property of $\cok(K)$ there exists a unique $\tilde x : X \to R'$ such that subdiagram (2) commutes. This is the $\tilde x$ we are looking for. It remains to show that $x = r' \circ T(\tilde x)$. Since $\cok(K)$ is a surjection and since $T$ is right exact, also $T(\cok(K))$ is a surjection, and it is sufficient to verify $x \circ T(\cok(K)) = r' \circ T(\tilde x) \circ T(\cok(K))$, i.e.\ commutativity of
\be
\raisebox{2em}{\xymatrix{
T((P \boxtimes P)^{\oplus m}) \ar[r]^-{T(\cok(K))} \ar[dd]_{T(\cok(K))} \ar[dr]^{T(\tilde p)}
& T(X) \ar[r]^{T(\tilde x)}
& T(R') \ar[dd]^{r'}
\\
& T(P \boxtimes P^*) \ar[ur]_{T(\cok(n))} \ar[dr]^{\beta_P}
\\
T(X) \ar[rr]^{x}
&& \one^*
}}
\ee
Commutativity of the top square is $T$ applied to square (2) in \eqref{eq:projectives-aux1}; the right triangle is the definition of $r'$ in \eqref{eq:r'-def}; finally, the bottom left square is \eqref{eq:projectives-aux2}.

\medskip\noindent
\nxt \textsl{Uniqueness:}
We will show that if a map $f : X \to R'$ satisfies $r' \circ T(f) = 0$, then $f=0$. This implies that the $\tilde x$ constructed above is unique. Write $g = f \circ \cok(K)$. It is enough to show that $g=0$. Consider the diagram
\be \label{eq:R1*-proof-aux1}
\raisebox{2em}{\xymatrix{
(P \boxtimes P)^{\oplus m} \ar[r]^-{g} \ar@{-->}[d]_{\exists\,v} \ar@{-->}[dr]^{\exists\,h}
& R'
\\
(P \boxtimes P)^{\oplus |N|} \ar[r]^-{n} 
& P \boxtimes P^* \ar[u]_{\cok(n)}
}}
\ee
Since $(P \boxtimes P)^{\oplus m}$ is projective, we can pull back $g$ along the surjection $\cok(n)$, giving us the existence of $h$. By \eqref{eq:r'-def} we have $\beta_P = r' \circ T(\cok(n))$, so that  $\beta_P \circ T(h) = r' \circ T(\cok(n)) \circ T(h) = r' \circ T(g) = r' \circ T(f) \circ T(\cok(K)) = 0$ by assumption on $f$. Hence we can apply lemma \ref{lem:factor-through-n} to obtain the map $v$ in \eqref{eq:R1*-proof-aux1}. Altogether, $g = \cok(n) \circ n \circ v = 0$.
\end{proof}

\begin{remark}
Because of the finiteness assumption (PF), there is a finite number of isomorphism classes of simple objects in $\Cc$. Let $\{ U_i | i \in \Ic \}$ be a choice of representatives. Furthermore, each $U_i$ has a projective cover $P_i$. For the projective generator, we can choose $P = \bigoplus_{i \in \Ic} P_i$, so that $R'$ arises as a quotient of $P \boxtimes P^*$. In fact, one can choose a `smaller' starting point, namely 
\be
  Q := \bigoplus_{i \in \Ic} P_i \boxtimes P_i^* 
\ee
(but then the above proof would have involved more indices). To describe the map whose cokernel to take, define the subspace
\be
\textstyle
  M = \big\{  f : P \boxtimes P \to Q \,\big|\, \sum_{i\in\Ic} \beta_{P_i} \circ f = 0 \big\} 
  ~~ \subset ~~ \Cc^2(P \boxtimes P , Q)
  \ .
\ee
Denote by $\iota_i : P_i \to P$ and $\pi_i : P \to P_i$ the embedding and restriction map of the direct sum. Pick a basis $\{ v_j \}$ of $M$ and define $m : (P \boxtimes P)^{\oplus |M|} \to Q$ as $m = \sum_{l=1}^{|M|} v_l \circ p_l$, with $p_l$ the $l$'th projection $(P \boxtimes P)^{\oplus |M|} \to P \boxtimes P$. Set $R'' = \cok(m)$. Then in fact
\be
  R' \cong R'' \ ,
\ee
with $R'$ defined as in \eqref{eq:R'-def}. To see this, define
$\pi : P \boxtimes P^* \to Q$, $\pi = \bigoplus_{i \in \Ic} \pi_i \boxtimes \iota_i^*$
and
$\iota : Q \to P \boxtimes P^*$, $\iota = \bigoplus_{i \in \Ic} \iota_i \boxtimes \pi_i^*$. These maps make the two diagrams contained in
\be \label{eq:PiPi-aux1}
\raisebox{2em}{\xymatrix@R=1em@C=3em{
P \otimes_{\Cc} P^* \ar@<-1ex>[dd]_{T(\pi)} \ar[dr]^{\beta_P}
\\
 & \one^*
\\
\bigoplus_{i \in\Ic} P_i \otimes_{\Cc} P_i^*  \ar@<-1ex>[uu]_{T(\iota)} \ar[ru]_-{\sum_{i \in \Ic} \beta_{P_i} }
}}
\ee
commute. For example, $\beta_P 
= \sum_{i \in \Ic} \beta_P \circ (\iota_i \otimes_{\Cc} \id_{P^*}) \circ (\pi_i \otimes_{\Cc} \id_{P^*}) 
= \sum_{i \in \Ic} \beta_{P_i} \circ (\pi_i \otimes_{\Cc} \iota_i^*) 
= \sum_{i \in \Ic} \beta_{P_i} \circ T(\pi)$. We can now construct maps between the two cokernels using their universal properties. Consider the diagram
\be 
\raisebox{2em}{\xymatrix{
(P \boxtimes P)^{\oplus |N|}  \ar[r]^{n} 
& P \boxtimes P^* \ar[r]^{\cok(n)}  \ar@<-1ex>[d]_{\pi}
& R' \ar@<-1ex>@{-->}[d]_{\exists !}
\\
(P \boxtimes P)^{\oplus |M|}  \ar[r]^{m} 
& Q \ar[r]^{\cok(m)}  \ar@<-1ex>[u]_{\iota}
& R'' \ar@<-1ex>@{-->}[u]_{\exists !}
}}
\ee
The diagram \eqref{eq:PiPi-aux1} tells us that $(\sum_{i \in \Ic} \beta_{P_i}) \circ T(\pi \circ n) = \beta_P \circ T(n) = 0$. Thus the image of $\pi \circ n$ lies in the image of $m$ (by an argument analogous to the one in lemma \ref{lem:factor-through-n}), so that $\cok(m) \circ \pi \circ n = 0$. The universal property gives a unique map $R' \to R''$. Similarly one checks that $\cok(n) \circ \iota \circ m = 0$, giving the map $R'' \to R'$. By uniqueness, these are inverse to each other.
\end{remark}

The next lemma prepares the proof of theorem \ref{thm:trivial-twist}.

\begin{lemma} \label{lem:kernel-of-cokn}
For all $u \in \Cc(P,P)$ we have $\cok(n) \circ (u \boxtimes \id - \id \boxtimes u^*) = 0$.
\end{lemma}

\begin{proof}
Pick an $m \in \Zb_{>0}$ such that there is a surjection $s : (P \boxtimes P)^{\oplus m} \to P \boxtimes P^*$. Let $f = (u \boxtimes \id - \id \boxtimes u^*) \circ s$. Then the statement $\cok(n) \circ (u \boxtimes \id - \id \boxtimes u^*) = 0$ is equivalent to $\cok(n) \circ f = 0$. We will show the latter. By lemma \ref{lem:conjugate-prop}\,(ii), we have
\be
  \Big[ T((P \boxtimes P)^{\oplus m}) \xrightarrow{~T(s)~} P \otimes P^* \xrightarrow{u \otimes \id - \id \otimes u^*} P \otimes P^* \xrightarrow{~\beta_P~} \one^* \Big] = 0 \ .
\ee
We can thus apply lemma \ref{lem:factor-through-n} and obtain a map $\tilde f : (P \boxtimes P)^{\oplus m} \to (P \boxtimes P)^{\oplus |N|}$ such that $f = n \circ \tilde f$. Hence $\cok(n) \circ f = \cok(n) \circ n \circ \tilde f = 0$.
\end{proof}

\begin{proof}[Proof of theorem \ref{thm:trivial-twist}]
Since $\nu \boxtimes \id - \id \boxtimes \tilde\nu$ is a natural transformation of the identity functor on $\CoC$, the diagram
\be \label{eq:proof-trivial-twist-aux1}
\raisebox{2em}{\xymatrix{
P \boxtimes P^* \ar[r]^{\cok(n)} \ar[d]_{(\nu \boxtimes \id - \id \boxtimes \tilde\nu)_{P \boxtimes P^*}}
& R' \ar[d]^{(\nu \boxtimes \id - \id \boxtimes \tilde\nu)_{R'}}
\\
P \boxtimes P^* \ar[r]^{\cok(n)} & R'
}}
\ee
commutes. Now note that for all $U\in\Cc$,
\be
  \tilde \nu_{U^*} 
  = (\delta_{U^*})^{-1} \circ (\nu_{U^{**}})^* \circ \delta_{U^*} 
  = (\delta_{U})^* \circ (\nu_{U^{**}})^* \circ (\delta_{U}^{-1})^*
  = (\nu_U)^* \ .
\ee
Thus $(\nu \boxtimes \id - \id \boxtimes \tilde\nu)_{P \boxtimes P^*} = \nu_P \boxtimes \id - \id \boxtimes (\nu_P)^*$. By lemma \ref{lem:kernel-of-cokn}, the lower path in the above diagram is zero. Since $\cok(n)$ is surjective, this implies $(\nu \boxtimes \id - \id \boxtimes \tilde\nu)_{R'}=0$.
\end{proof}

\subsection{Adjoint to the tensor product} \label{app:R-adj-T}

We first need to establish the compatibility of condition (C) and the Deligne product. We do this under the assumption that we are given two categories $\Cc$, $\Dc$ which satisfy condition (PF) from section \ref{sec:deligne-prod} (rather than (F) for we need to invoke \cite[Prop.\,5.5]{Deligne:2007}), which are monoidal with $k$-linear right exact tensor product, and both have conjugates according to condition (C).

Since $(-)^*$ is an equivalence, it is exact. Thus $\Cc^\mathrm{opp} \times \Dc^\mathrm{opp} \xrightarrow{(-)^* \times (-)^*} \Cc \times \Dc \xrightarrow{\boxtimes} \Cc \boxtimes \Dc$ factors through a functor $\Cc^\mathrm{opp} \boxtimes \Dc^\mathrm{opp} \to \Cc \boxtimes \Dc$. By \cite[Prop.\,5.5]{Deligne:2007}, we may take $\Cc^\mathrm{opp} \boxtimes \Dc^\mathrm{opp} = (\Cc \boxtimes \Dc)^\mathrm{opp}$. Altogether, we get a contragredient involutive functor on $\Cc \boxtimes \Dc$, which we also denote by $(-)^*$. By definition, on `factorised objects' it satisfies
\be
  (C \boxtimes D)^* ~=~ C^* \boxtimes D^* \ .
\ee

\begin{lemma} \label{lem:CxD-prop-C}
$(-)^* : \Cc \boxtimes \Dc \to \Cc \boxtimes \Dc$ satisfies property (C).
\end{lemma}

\begin{proof}
The natural isomorphisms $\delta^{\Cc}$ and $\delta^{\Dc}$ between the exact functors $\boxtimes$ and $\boxtimes \circ \{ (-)^{**} \times (-)^{**} \}$ from $\Cc \times \Dc$ to $\Cc \boxtimes \Dc$ provide a natural isomorphism $\delta$ from $\Id$ to $(-)^{**}$ on $\Cc \boxtimes \Dc$ with the required property $(\delta_X)^* = (\delta_{X^*})^{-1}$. 

For the existence of $\pi$, we stress again \cite[Prop.\,5.5]{Deligne:2007}: $\Cc^\mathrm{opp} \boxtimes \Dc^\mathrm{opp} = (\Cc \boxtimes \Dc)^\mathrm{opp}$. Thus there is an equivalence of functor categories between $k$-linear right exact functors in each argument $\Cc^\mathrm{opp} \times \Dc^\mathrm{opp} \to \Ec^\text{opp}$ and right exact functors $(\Cc \boxtimes \Dc)^\mathrm{opp} \to \Ec^\text{opp}$. But this is the same as saying that there is an equivalence of functor categories between $k$-linear functors  $\Cc \times \Dc \to \Ec$, left exact in each argument, and left exact functors $\Cc \boxtimes \Dc \to \Ec$. 
Given our assumptions on $\Cc$ and $\Dc$,
by \cite[Cor.\,5.4]{Deligne:2007}, the functor $\boxtimes : \Cc \times \Dc \to \Cc \boxtimes \Dc$ itself is exact in each argument. 

Recall that an abelian category $\Ac$, the Hom-functor from $\Ac^\mathrm{opp} \times \Ac$ to abelian groups given by $(A,B) \mapsto \Ac(A,B)$ is left exact in each argument. Thus the functor $\Cc^\mathrm{opp} \times \Cc^\mathrm{opp} \to \vect$ given by $(U,V) \mapsto \Cc(U,V^*)$ is left exact in each argument. The maps $\pi_{U,V}$ provide a natural isomorphism from this functor to the functor $(U,V) \mapsto \Cc(U \otimes V, \one^*)$, which therefore is also left exact in both arguments (even though it involves the right exact tensor product). The same reasoning applies to $\Dc$. The combined functors
\be \label{eq:CxD-prop-C-aux1}
\begin{array}{ll}
  \Cc^\mathrm{opp} \times \Dc^\mathrm{opp} \times \Cc^\mathrm{opp} \times \Dc^\mathrm{opp} \to \vect
  ~ , ~~
  &(U,A,V,B) \mapsto \Cc(U,V^*) \otimes_k \Dc(A,B^*) ~~ \text{and}
  \\[.5em]
  &(U,A,V,B) \mapsto \Cc(U \otimes V,\one^*) \otimes_k \Dc(A \otimes B, \one^*) \ , 
\end{array}
\ee
are equally left exact in each argument, and
thus give two functors $(\Cc \boxtimes \Dc)^\mathrm{opp} \times (\Cc \boxtimes \Dc)^\mathrm{opp} \to (\Cc \boxtimes \Dc \boxtimes \Cc \boxtimes \Dc)^\mathrm{opp} \to \vect$ which are left exact in each argument. In view of \eqref{eq:hom-iso}, these functors are necessarily given by 
\be \label{eq:CxD-prop-C-aux2}
(X,Y) \mapsto \Cc {\boxtimes} \Dc\,\big(X,Y^*\big)
\quad \text{and} \quad
(X,Y) \mapsto \Cc {\boxtimes} \Dc\big(X \otimes Y , \one^* \boxtimes \one^* \big) \ .
\ee
The equivalence \eqref{eq:defining-equiv} of functor categories -- which as we saw above also holds for the corresponding categories of left exact functors -- now shows that the natural isomorphism $\pi^{\Cc}_{U,V} \otimes_k \pi^{\Dc}_{A,B}$ between the functors \eqref{eq:CxD-prop-C-aux1} provides a natural isomorphism $\pi_{X,Y}$ between the functors \eqref{eq:CxD-prop-C-aux2}.
\end{proof}

Recall the definition of the functor $R$ in terms of the conjugates on $\Cc$ and $\CoC$ given in \eqref{eq:R(U)-expression}.

\begin{proof}[Proof of theorem \ref{thm:R-adj-to-T}]
The natural isomorphisms
\be \label{eq:xi_XU-def}
  \xi_{X,U} : \Cc(T(X),U) \overset{\sim}\longrightarrow \Cc^2(X,R(U))
\ee
are provided by the composition
\be
\begin{array}{ll}
  \Cc\big(\,T(X)\,,\,W\,\big)
  &\xrightarrow[\pi\text{ and }\delta]{~\sim~}~
  \Cc\big(\,T(X) \otimes_\Cc W^*\,,\, \one^*\,\big)
  ~\xrightarrow{~\sim~}~
  \Cc\big(\,T(X) \otimes_\Cc T(W^* \boxtimes \one)\,,\, \one^*\,\big)
\\[.7em]
  &\xrightarrow[T_2]{~\sim~}~
  \Cc\big(\,T(X  \otimes_{\Cc^2} (W^* \boxtimes \one))\,,\, \one^*\,\big)
  ~\xrightarrow[\chi\text{ from \eqref{eq:R1*-chi}}]{~\sim~}~
  \Cc^2\big(\,X  \otimes_{\Cc^2} (W^* \boxtimes \one)\,,\, R_{\one^*}\,\big)
\\[.7em]
  &\xrightarrow[\text{Lem.\,\ref{lem:CxD-prop-C}}]{~\sim~}~
  \Cc^2\big(\,[X  \otimes_{\Cc^2} (W^* \boxtimes \one)] \otimes_{\Cc^2} (R_{\one^*})^*\,,\, \one^* \boxtimes \one^* \,\big)
\\[.7em]
  &\xrightarrow[\text{assoc}]{~\sim~}~
  \Cc^2\big(\,X  \otimes_{\Cc^2} [(W^* \boxtimes \one) \otimes_{\Cc^2} (R_{\one^*})^*]\,,\, \one^* \boxtimes \one^* \,\big)
\\[.7em]
  &\xrightarrow[\text{Lem.\,\ref{lem:CxD-prop-C}}]{~\sim~}~
  \Cc^2\big(\,X  \,,\, [(W^* \boxtimes \one) \otimes_{\Cc^2} (R_{\one^*})^*]^* \,\big)
  \equiv \Cc^2\big(\,X\,,\,R(W)\,\big) \ .
\end{array}
\ee
\end{proof}

The adjunction natural transformations are, in terms of the isomorphism \eqref{eq:xi_XU-def},
\be \label{eq:alpha-beta-def}
  \eta_X := \xi_{X,T(X)}(\id_{T(X)}) : X \to R(T(X)) 
  ~,~~
  \eps_U := \xi_{R(U),U}^{-1}(\id_{R(U)}) : T(R(U)) \to U \ ,
\ee
and the satisfy the adjunction properties, for $X \in \CoC$ and $U \in \Cc$,
\be\label{eq:adjunction-prop}
\begin{array}{ll}
\big[\, R(U) \xrightarrow{\eta_{R(U)}} R\,TR(U) \xrightarrow{R(\eps_U)} R(U) \,\big] &=~ \id_{R(U)} \ ,
\\[1em]
\big[\, T(X) \xrightarrow{T(\eta_X)} TR\,T(X) \xrightarrow{\eps_{T(X)}} T(X) \,\big] &=~ \id_{T(X)} \ .
\end{array}
\ee
The functor $R$ is lax monoidal (and colax monoidal), with $R_0$ and $R_2$ given by
\be \label{eq:R0-R2-def}
\begin{array}{rcl}
R_0 &=& \Big[ \one \boxtimes \one \xrightarrow{\eta_{\one \boxtimes \one}} R(T(\one \boxtimes \one))
\xrightarrow{R(T_0^{-1})} R(\one) \Big] \ ,
\\[1.5em]
R_{2;U,V} &=& \Big[RU \otimes_{\Cc^2} RV 
\xrightarrow{\eta_{RU \otimes RV}} R\big(T(RU \otimes_{\Cc^2} RV)\big)
\\[.5em]
&& \hspace{4em}
\xrightarrow{T_{2;RU,RV}^{-1}} R\big(TR(U) \otimes_{\Cc} TR(V)\big)
\xrightarrow{R(\eps_U \otimes \eps_V)} R(U \otimes_{\Cc} V) \Big] \ ,
\end{array}
\ee
see, e.g., \cite[Def.\,2.1\,\&\,Lem.\,2.7]{Kong:2008ci}.

\subsection{More details on centres} \label{app:full-centre-via-left-centre}

This appendix contains an auxiliary result which implies the existence of the left centre in abelian monoidal categories which have conjugates as in (C), and it contains the proof of theorem \ref{thm:full-centre-via-left-centre}.

\medskip

Let $\Ac$ be an abelian monoidal category with conjugates according to condition (C). Let $m : A \otimes B \to C$ be a morphism in $\Ac$. Consider the category $\Qc$ whose objects are pairs $(U,u)$ where $u : U \to A$ is such that $m \circ (u \otimes \id_B)=0$. Morphisms $f : (U,u) \to (V,v)$ in $\Qc$ are maps $f : U \to V$ in $\Ac$ such that $v \circ f = u$.

Using conjugates, from $m$ we obtain a morphism $\tilde m : A \to (B \otimes C^*)^*$ via
\be
  \Ac(A B,C) \xrightarrow{~\sim~} 
  \Ac(A B, C^{**}) \xrightarrow{~\sim~}
  \Ac((A B) C^*, \one^* ) \xrightarrow{~\sim~}
  \Ac(A (B C^*), \one^* ) \xrightarrow{~\sim~}
  \Ac(A  , (B C^*)^* )  \ .
\ee
When applied to $m$, this chain of natural isomorphisms yields 
\be
  \tilde m = \pi^{-1}_{A,BC^*}\big[\pi_{AB,C^*}(\delta_C \circ m) \circ \alpha_{A,B,C^*} \big] \ .
\ee  
Naturality in $A$ (cf.\,\eqref{eq:pi-natural-aux}) implies that for any $u : U \to A$ we have 
\be \label{eq:m-naturality-aux}
  \tilde m \circ u 
  = 
  \pi^{-1}_{U,BC^*}\big[\pi_{UB,C^*}(\delta_C \circ m \circ (u \otimes \id_B)) \circ \alpha_{U,B,C^*} \big] \ ,
\ee

\begin{lemma}
The following are equivalent.
\\[.3em]
(i) $(K,k)$ is terminal in $\Qc$.
\\[.3em]
(ii) $k : K \to A$ is the kernel of $\tilde m : A \to (B \otimes C^*)^*$.
\end{lemma}

\begin{proof}
(i)$\,\Rightarrow\,$(ii): Let $u : U \to A$ be a morphism such that $\tilde m \circ u = 0$. By \eqref{eq:m-naturality-aux}, then also $m \circ (u \otimes \id_B) = 0$. Thus $(U,u)$ is an object in $\Qc$. By terminality there is a unique arrow $f : U \to K$ such that $k \circ f = u$. This is the thought-for $f$ in the universal property of the kernel.
\\
(ii)$\,\Rightarrow\,$(i): Let $(U,u)$ be an object in $\Qc$. Then $m \circ (u \otimes \id_B) = 0$ and as above we see that $\tilde m \circ u = 0$. By the universal property of the kernel, there is a unique map $f : U \to K$ such that $k \circ f = u$. Thus there is a unique morphism $f : (U,u) \to (K,k)$.
\end{proof}

This lemma implies the existence of the left centre of an algebra $B$ in the category $\Ac$. Indeed, the universal property of the left centre from definition \ref{def:left-centre} amounts to the terminal object condition in the category $\Qc$ from above with the choice 
\be
  m = \mu_B \circ (\id_{B\otimes B} - c_{B,B}) : B \otimes B \to B \ .
\ee
As the kernel of $\tilde m$ exists in $\Ac$, so does the terminal object in $\Qc$ and hence the left centre.

\medskip

Next, we turn to the proof of  theorem \ref{thm:full-centre-via-left-centre}.

\begin{proof}[Proof of theorem \ref{thm:full-centre-via-left-centre}]
We need to check that $(Z,z) \equiv (\,C_l(R(A))\,,\,\eps_A \circ T(e)\,)$ satisfies the universal property in definition \ref{eq:full-centre-in-C2}. Let thus $(X,x)$ be a pair such that \eqref{eq:full-cent-univprop} commutes.

By lemma \ref{lem:repres+terminal} with $U=A$, $R'=R(A)$ and $r' = \eps_A$, there is a unique map $\tilde x : X \to R(A)$ such that \eqref{eq:R'-terminal-ob-prop} commutes, i.e.\ such that $\eps_A \circ T(\tilde x) = x$. This map is given by (use \eqref{eq:alpha-beta-def} and naturality of $\xi$)
\be \label{eq:tilde-x-via-adjunct}
  \tilde x ~=~
  \Big[ X \xrightarrow{~\xi_{X,A}(x)~} R(A) \Big]
  ~=~
  \Big[ X \xrightarrow{~\eta_X~} R(T(X)) \xrightarrow{~R(x)~} R(A) \Big] \ .
\ee
We will see below that the map $\tilde x$ satisfies the condition for the universal property of the left centre, that is, diagram \eqref{eq:left-centre-diagram} commutes for the pair $(X,\tilde x)$. Thus, the map $\tilde x$ factors as 
\be
\raisebox{2em}{\xymatrix{
& C_l(R(A)) \ar[dr]^e
\\
X \ar[ur]^{x'} \ar[rr]^{\tilde x} && R(A)
}}
\ee
Since $\tilde x$ is unique and $e$ is mono, also $x'$ is unique. That $x'$ makes \eqref{eq:full-cent-pairmap} commute follows from $x = \eps_A \circ T(\tilde x) = \eps_A \circ T(e) \circ T(x') = z \circ T(x')$. This shows that $(Z,z)$ is the full centre of $A$.

\begin{figure}[bt]
$$
\small\xymatrix{
X \, R(A) \ar@{}[rrrd]|{(1)} \ar[rrr]^{\eta_X\,1} \ar[dr]_{\eta_{X\,R(A)}} \ar[ddddddd]^{c_{X,RA}} \ar@{}[ddddddr]|{(5)}
&&& RT(X)\,R(A)  \ar[d]_{\eta_{RT(X)\,R(A)}} \ar@/^2em/[dddr]^{R_2}  \ar@{}[dddr]|{(4)}
\\
& RT\hspace{-1pt}\big(X \, R(A)\big) \ar@{}[rrd]|{(2)} \ar[rr]^{RT(\eta_X \,1)} \ar[dr]_{R(T_2^{-1})} \ar[ddddd]^{RT(c_{X,RA})} \ar@{}[ddddr]|{(6)}
&& RT\hspace{-1pt}\big(RT(X)\,R(A)\big) \ar[d]_{R(T_2^{-1})} 
\\
&& R\hspace{-1pt}\big(T(X) \, TR(A)\big) \ar@{=}[dr]^{(3)} \ar@/^.5em/[r]^-{R(T(\eta_X)\,1)} \ar[ddd]^{R(\hat \varphi_{X,RA}) \equiv R(\varphi_{X,TRA})}
& R\hspace{-1pt}\big(TRT(X)\,TR(A)\big) \ar[d]^{R(\eps_{T(X)} 1)}
\\
&&& R\hspace{-1pt}\big(T(X) \, TR(A)\big)  \ar[r]^-{R(1\,\eps_A)} \ar@{}[dr]|{(7)}
& R\hspace{-1pt}\big(T(X)\,A\big) \ar[d]^{R(\varphi_{X,A})} 
\\
&&& R\hspace{-1pt}\big(TR(A) \, T(X)\big)  \ar[r]_-{R(\eps_A\,1)}
& R\hspace{-1pt}\big(A\,T(X)\big)
\\
&& R\hspace{-1pt}\big(TR(A) \, T(X)\big) \ar@{=}[ur]_{(9)} \ar@/_.5em/[r]_-{R(1\,T(\eta_X))}
& R\hspace{-1pt}\big(TR(A)\,TRT(X)\big) \ar[u]_{R(1\,\eps_{T(X)})}
\\
& RT\hspace{-1pt}\big(R(A)\,X \big) \ar@{}[rru]|{(10)} \ar[rr]_{RT(1\,\eta_X )} \ar[ur]^{R(T_2^{-1})}
&& RT\hspace{-1pt}\big(R(A)\,RT(X)\big) \ar[u]^{R(T_2^{-1})} 
\\
R(A) \, X \ar@{}[rrru]|{(11)} \ar[rrr]_{1\,\eta_X} \ar[ur]_{\eta_{X\,R(A)}}
&&& R(A)\,RT(X)  \ar[u]^{\eta_{RT(X)\,R(A)}} \ar@/_2em/[uuur]_{R_2}  \ar@{}[uuur]|{(8)}
}
$$
\caption{Commutativity of subdiagram (1) in \eqref{eq:full-centre-via-left-centre-aux1}. 
All $\otimes$ have been omitted, instead of $\id$ the shorthand $1$ in used, and only a minimum of brackets is given. The commutativity of the individual cells is explained in the main text.}
\label{fig:subdiag-1}
\end{figure}
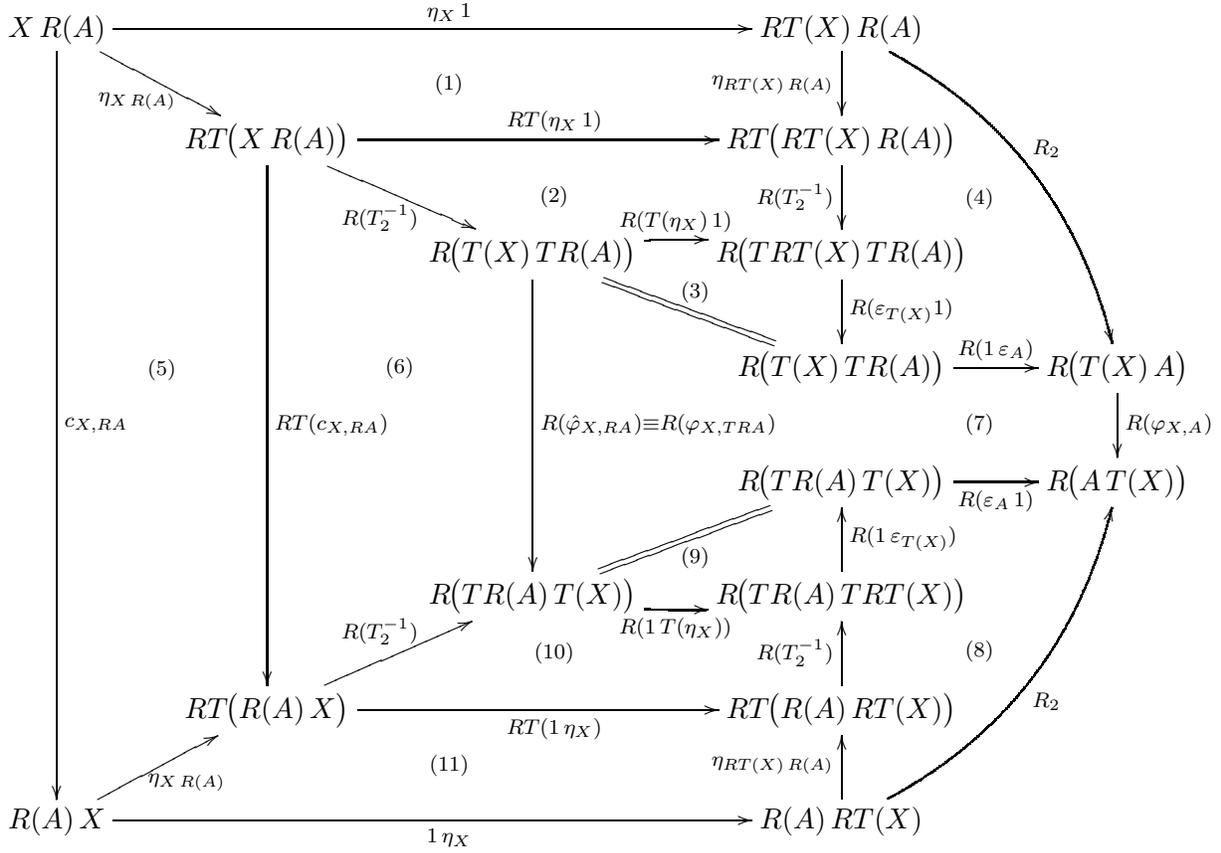

It remains to check that $\tilde x$ satisfies \eqref{eq:left-centre-diagram}. For convenience, we reproduce (minimally adapted to our setting) the proof given in \cite[Thm.\,5.4]{Davydov:2009}. Substituting the expression \eqref{eq:tilde-x-via-adjunct} for $\tilde x$, we need to show commutativity of
\be \label{eq:full-centre-via-left-centre-aux1}
\raisebox{2em}{\xymatrix@R=1em@C=3em{
X \otimes_{\Cc^2} R(A) \ar[r]^-{\eta_X \otimes \id} \ar[dddddd]_{c_{X,R(A)}} \ar@{}[ddddddr]|{(1)} 
& R(T(X)) \otimes_{\Cc^2} R(A) \ar[r]^-{R(x) \otimes \id} \ar[dd]_{R_2} \ar@{}[ddr]|{(2)}
& R(A) \otimes_{\Cc^2} R(A) \ar[dd]_{R_2} \ar[dddr]^{\mu_{R(A)}}
\\
\\
& R(T(X) \otimes_{\Cc} A) \ar[dd]_{R(\varphi_{X,A})} \ar[r]_{R(x \otimes \id)} \ar@{}[ddr]|{(3)}
& R(A \otimes_{\Cc} A) \ar[dr]_{R(\mu_A)}
\\
&&& R(A)
\\
& R(A \otimes_{\Cc} T(X))  \ar[r]^{R(\id \otimes x)} \ar@{}[ddr]|{(4)}
& R(A \otimes_{\Cc} A) \ar[ur]^{R(\mu_A)}
\\
\\
R(A) \otimes_{\Cc^2} X \ar[r]^-{\id \otimes \eta_X} 
& R(A) \otimes_{\Cc^2} R(T(X)) \ar[r]^-{\id \otimes R(x)} \ar[uu]^{R_2}
& R(A) \otimes_{\Cc^2} R(A) \ar[uu]^{R_2} \ar[uuur]_{\mu_{R(A)}}
}}
\ee
The rightmost triangles are the definition \eqref{eq:F(A)-product} of $\mu_{R(A)}$. Squares 2 and 4 are naturality of $R_2$. Subdiagram 3 is $R$ applied to the defining property \eqref{eq:full-cent-univprop} of $x$. Subdiagram 1 is somewhat tedious and is further analysed in figure \ref{fig:subdiag-1}. In explaining the commutativity of the various cells, let us start with the key step: the two ways of writing the arrow between cells 6 and 7, which amounts to $R$ applied to lemma \ref{lem:phi-hatphi-rel}. Using $R(\hat\varphi_{X,RA})$, cell 6 is $R$ applied to the definition of $\hat\varphi$ in \eqref{eq:varphi-hat-def}, and using $R(\varphi_{X,TRA})$, cell 7 is $R$ applied to naturality of $\varphi_{X,-}$. The remaining cells are as follows: cells 1, 5, 11 are naturality of $\eta$, cells 2, 10 are naturality of $T_2$, cells 3, 9 are the adjunction property \eqref{eq:adjunction-prop}, cells 4, 8 are the definition \eqref{eq:R0-R2-def} of $R_2$.
\end{proof}

\section{OPEs in $R(\Wc^*)$} \label{app:OPEs}

\subsection{OPEs involving $\Omega$}

We will demonstrate that for all fields $\phi \in R(\Wc^*)$ one has
\be \label{eq:Om-OPE-app}
  M_z(\Omega \otimes \phi) = M_z(\phi \otimes \Omega) = \pi(\phi) \cdot \Omega \ ,
\ee
where $M_z(\phi \otimes \psi)$ is the OPE as introduced in \eqref{eq:Mz-def}. We will also use the conventional notation $\phi(z)\psi(0)$ for the OPE. 
To establish \eqref{eq:Om-OPE-app} we will first show that $C_2(u|z,0,\Omega,L_m \phi) = 0$ for all $\phi \in F$, $u \in \overline F{}^*$, and $m \in \Zb$ (recall the notation from section \ref{sec:genout-ideal}). Suppose the contrary and let $M \in \Zb$ be the largest integer such that $C_2(u|z,0,\Omega,L_M \phi) \neq 0$. Choose $N>0$ such that $L_n u = 0$ for all $n \ge N$. Apply property (C5') for $f(\zeta) = (z-\zeta)^{-N-M}\zeta^{M+1}$. Since $\Omega$ is annihilated by all $L_n$, one checks that \eqref{eq:mode-moving-generalout} becomes $0 = z^{-N-M} C_2(u|z,0,\Omega,L_M \phi)$, in contradiction to our assumption. Similarly one checks that $C_2(u|z,0,\Omega,\overline L_m \phi) = 0$ for all $u,\phi,m$. Let $\Delta,N$ be such that $(L_0+\overline L_0-\Delta)^N \phi = 0$. Then
\be
  0 = C_2(u|z,0,\Omega,(L_0+\overline L_0-\Delta)^N \phi) = (-\Delta)^N C_2(u|z,0,\Omega,\phi) \ ,
\ee
and so the OPE $\Omega(z)\phi(0)$ can only be non-vanishing for $\phi \in F^{(0)}$. Since $\omega = L_0 \eta$ and $\Omega = L_0^2 \eta$, the OPE vanishes for $\phi = \omega,\Omega$. To confirm \eqref{eq:Om-OPE-app} it only remains to check $\Omega(z)\eta(0) = \Omega(0)$.
Using once more that $\Omega$ is annihilated by all Virasoro modes, we have $L_m M_z(\Omega\otimes\eta) = M_z(\Omega\otimes L_m \eta)$, which is zero by the above argument, and analogously $\overline L_m M_z(\Omega\otimes\eta) = 0$. This applies in particular to $m=-1$, and the intersection of the kernels of $L_{-1}$ and $\overline L_{-1}$ is $\Cb\,\Omega$. Thus $\Omega(z)\eta(0) = a \cdot \Omega(0)$ for some $a \in \Cb$. But then also $\eta(z)\Omega(0) = a \cdot \Omega(0)$ and $\eta(z)\eta(w)\Omega(0) = a^2 \cdot \Omega(0)$. On the other hand, from $\pi(\eta(z)\eta(0)) = 1$ we see that $\eta(z)\eta(0) = \eta(0) + (\text{other fields})$. Thus $a=1$.

\subsection{OPEs involving the holomorphic field $\Tc$}

The next-simplest set of OPEs are those of the form $\Tc(z)\phi(0)$. Since $\Tc$ is holomorphic, this OPE does not involve logarithmic singularities (or it would not be single-valued). For example, the most general ansatz for the OPE with $\omega$ is
\be \label{eq:T-om-OPE-aux1}
  \Tc(z) \omega(0) = z^{-2}(P\cdot \eta(0) + Q\cdot \omega(0) + R\cdot \Omega(0)) + z^{-1}( S\cdot L_{-1} \eta(0) + U \cdot L_{-1} \omega(0) ) + \Oc(z^0) 
\ee
for some constants $P,Q,R,S,U \in \Cb$. These constants are further constrained by the identity
\be\label{eq:Lm-T-comm}
\begin{array}{l}
  L_m M_z(\Tc\otimes\phi) 
  = \sum_{k=0}^{3} {{m+1}\choose k} z^{m+1-k} M_z(L_{k-1}\Tc\otimes\phi) + M_z(\Tc\otimes L_m \phi) 
\\[.5em]
  = z^m\big(2(m{+}1)+z\tfrac{\partial}{\partial z}\big) M_z(\Tc\otimes\phi) - \tfrac{5}{6}(m^3{-}m) \, z^{m-2} \pi(\phi) \cdot \Omega + M_z(\Tc\otimes L_m \phi)  \ ,
\end{array}  
\ee
where $m\in \Zb$ and $\phi \in F$ are arbitrary. The first equality follows from (C5') and the second uses \eqref{eq:M-modeactions} and \eqref{eq:Om-OPE-app}. If one applies this identity for $m=0$ and $m=1$ to \eqref{eq:T-om-OPE-aux1}, one quickly finds that $P=Q=S=0$ and $R=U$. Thus
\be \label{eq:T-om-OPE-aux2}
  \Tc(z) \omega(0) = R \cdot (z^{-2} \, \Omega(0) + z^{-1} L_{-1} \omega(0) ) + \Oc(z^0)  \ .
\ee
This also provides the two-point correlator 
\be \label{eq:T-om-2pt}
  {}^\eta\langle \Tc(z) \omega(0) \rangle = R \cdot z^{-2} \ .
\ee
Actually, $R$ is necessarily zero, though it will take us a little while to get there.
Since the states of generalised weight $(1,0)$,  $(2,0)$ and  $(3,0)$ are Virasoro descendents of $\eta$, the same method allows one to determine the $z^0$ and $z^1$ coefficient in this OPE. The calculations become more lengthy, but the answer is simply
\be \label{eq:T-om-OPE-aux3}
  \Tc(z) \omega(0) = R \cdot (z^{-2} \Omega(0) + z^{-1} L_{-1} \omega(0) + L_{-2} \omega(0) + z L_{-3} \omega(0) ) + \Oc(z^2)  \ .
\ee
Next we compute $\Tc(z)\Tc(0)$ by using \eqref{eq:Lm-T-comm} to move all Virasoro modes in $M_z(\Tc \otimes (L_{-2} - \tfrac32 L_{-1}L_{-1})\omega)$ to the left and by then inserting the OPE \eqref{eq:T-om-OPE-aux3}. A short calculation yields 
\be\label{eq:T-T-OPE-aux4}
  \Tc(z)\Tc(0) = R \cdot \big\{ -5 z^{-4} \Omega(0) + 2 z^{-2} \Tc(0) + z^{-1} (L_{-1} \Tc)(0) \big\} + \Oc(z^0) \ .
\ee
The OPEs \eqref{eq:T-om-OPE-aux3} and \eqref{eq:T-T-OPE-aux4} allow one to determine the three-point function ${}^\eta\langle \Tc(z) \Tc(w) \omega(0) \rangle$ by singularity subtraction. Thinking of the three-point function as a function of $z$, this function vanishes at infinity and has poles only at $w$ and $0$. Subtracting these poles we hence find a holomorphic function on $\Cb$ vanishing at infinity, i.e.\ a function which is identically zero:
\be
\begin{array}{ll}
  0 ~=~ {}^\eta\langle \Tc(z) \Tc(w) \omega(0) \rangle \!\!&- ~ R \cdot \Big(~ 
  \tfrac{2}{(z-w)^2}  {}^\eta\langle \Tc(w) \omega(0) \rangle +
  \tfrac{1}{(z-w)} \tfrac{\partial}{\partial w} \,  {}^\eta\langle \Tc(w) \omega(0) \rangle
\\[.5em]
&\hspace{10em}
  + \tfrac{1}{z} \big(- \tfrac{\partial}{\partial w} \big) \, {}^\eta\langle \Tc(w) \omega(0) \rangle
  ~\Big) \ .
\end{array}
\ee
Substituting \eqref{eq:T-om-2pt}, the result is
\be
   {}^\eta\langle \Tc(z) \Tc(w) \omega(0) \rangle = \frac{2 \cdot R^2}{z\,w\,(z-w)^2} \ .
\ee
Note that this function is invariant under the exchange of $z$ and $w$ as it has to be. 
Repeating the above steps to constrain the OPE $\Tc(z)\eta(0)$ leads to
\be \label{eq:T-eta-OPE-aux1}
\begin{array}{ll}
\Tc(z)\eta(0)
&=~ z^{-2} \cdot \big\{ R \cdot \omega(0) + A \cdot \Omega(0) \big\} ~+~ 
z^{-1} \cdot\big\{ R \cdot(L_{-1}\eta)(0) + A \cdot (L_{-1}\omega)(0) \big\} 
\\[.5em]
& \qquad +~
\big\{ R \cdot(L_{-2}\eta)(0) + (A{+}1) \cdot (L_{-2}\omega)(0) - \tfrac32 (L_{-1}L_{-1}\omega)(0) \big\}
\\[.5em]
& \qquad +~
z \cdot \big\{ R \cdot(L_{-3}\eta)(0) + (A{+}1) \cdot (L_{-3}\omega)(0) 
+ (L_{-2}L_{-1}\omega)(0) 
\\[.5em]
& \hspace{5em}
- \tfrac32 (L_{-1}L_{-1}L_{-1}\omega)(0) 
\big\} ~+~ \Oc(z^2) \ ,
\end{array}
\ee
where $A \in \Cb$ is a new constant. The corresponding two-point correlator is
\be
  {}^\eta\langle \Tc(z) \eta(0) \rangle = A \cdot z^{-2} \ .
\ee
As before, the above OPE can be used to determine the OPE of $\Tc$ and $t$ with the result
\be \label{eq:T-t-OPE-aux1}
\begin{array}{ll}
\Tc(z)t(0)
&=~ z^{-4} \cdot \big\{ -5R \cdot \omega(0) + (9R-5(A{+}1)) \cdot \Omega(0) \big\}
\\[.5em]
& \qquad +~
z^{-2} \cdot\big\{ 2R \cdot t(0) + (R+2(A{+}1))\cdot \Tc(0) \big\}
\\[.5em]
& \qquad +~
z^{-1} \cdot \big\{ R \cdot (L_{-1}t)(0) + (A{+}1) \cdot (L_{-1}\Tc)(0) \big\}
 ~+~ \Oc(z^0) \ .
\end{array}
\ee
In a slightly tedious exercise the OPEs determine the three-point correlator of $\langle \Tc\Tc\eta \rangle$ by singularity subtraction to be
\be \label{eq:TTom-corr}
{}^\eta\langle \Tc(z) \Tc(w) \eta(0) \rangle~=~
\frac{-5R}{(z-w)^4} + \frac{R^2}{z^2\,w^2} + \frac{2RA}{z\,w\,(z-w)^2} + \frac{2RA}{z\,w^3} \ .
\ee
Because of the last summand, this expression is only invariant under $z \leftrightarrow w$ if $R=0$ or $A=0$. To see that actually $R=0$ is required, one can compute (in another such tedious exercise)
\be
   {}^\eta\langle \Tc(z) \Tc(w) t(0) \rangle -
   {}^\eta\langle \Tc(w) \Tc(z) t(0) \rangle 
   ~=~ 20 R (A{+}1) \cdot \frac{(z-w)(z^3+ \tfrac32 z^2 w + \tfrac32 z w^2 + w^3 )}{z^5 w^5} \ .
\ee
With $R=0$ the correlators ${}^\eta\langle \Tc(z) \Tc(w) \phi(0) \rangle$ are zero for $\phi$ any of $\omega,\eta,\Tc,t$. The OPEs \eqref{eq:T-om-OPE-aux3}, \eqref{eq:T-T-OPE-aux4}, \eqref{eq:T-eta-OPE-aux1} and \eqref{eq:T-t-OPE-aux1} reproduce the formulas in \eqref{eq:RW*-OPEs-T}.

\subsection{OPEs of generalised weight zero fields}

Next we consider the non-holomorphic OPE $\omega(z)\omega(0)$. The identity corresponding to \eqref{eq:Lm-T-comm} reads in this case
\be \label{eq:Lm-past-om}
  L_m M_z(\omega\otimes\phi) 
  = z^{m+1} \tfrac{\partial}{\partial z} M_z(\omega \otimes\phi) + (m+1)z^m \pi(\phi) \cdot \Omega + M_z(\omega\otimes L_m \phi)  \ ,
\ee
for all $\phi \in F$ and $m \in \Zb$, and analogously for $\overline L_m$. For $m=0$ we find in particular that
\be
  \big(L_0 - z \tfrac{\partial}{\partial z} \big) M_z(\omega \otimes \omega) 
  = 0 =
  \big(\overline L_0 - \bar z \tfrac{\partial}{\partial \bar z} \big) M_z(\omega \otimes \omega) \ .
\ee
The general solution to this first order differential equation reads 
\be
  M_z(\omega \otimes \omega ) = \exp\{ \ln(z) L_0 + \ln(\bar z) \overline L_0  \} \Psi
  \qquad \text{for some} ~~\Psi \in \overline F \ .
\ee
This shows that the leading term in the OPE is (take the component of $\Psi$ in $F^{(0)}$ to be $X \cdot \eta + Y \cdot \omega + B \cdot \Omega$)
\be
  \omega(z) \omega(0) = X \cdot \eta(0) + \big\{ Y + X\, \ln(|z|^2) \big\} \omega(0) + \big\{B + Y \ln(|z|^2) + \tfrac12 X \cdot \big(\ln(|z|^2)\big)^2 \big\} \Omega(0) + \dots
\ee
This expression can be used to compute the leading term in the OPE $\omega(z)\Tc(0)$, which we already know from \eqref{eq:RW*-OPEs-T} to be of order $\Oc(z^0)$. Using \eqref{eq:Lm-past-om} to move the $L_m$ modes past $\omega(z)$ one quickly finds the requirement that $X=Y=0$. This reproduces \eqref{eq:gen(00)-OPEs}.
For $\eta(z)\omega(0)$ we use \eqref{eq:Lm-T-comm} in the form
\be \label{eq:Lm-past-eta}
  L_m M_z(\eta\otimes\omega) 
  = z^{m+1} \tfrac{\partial}{\partial z} M_z(\eta \otimes\omega) + (m+1)z^m M_z(\omega \otimes\omega) + M_z(\eta\otimes L_m \omega)  \ ,
\ee
and analogously for $\overline L_m$. We can again make a general ansatz for the leading term in the OPE $\eta(z)\omega(0)$ and use the knowledge of $\omega(z)\omega(0)$ and $\eta(z)\Tc(0)$ (from \eqref{eq:RW*-OPEs-T}) to constrain the coefficients. The result is as stated in \eqref{eq:gen(00)-OPEs}, we skip the details.

\newpage

\newcommand\arxiv[2]      {\href{http://arXiv.org/abs/#1}{#2}}
\newcommand\doi[2]        {\href{http://dx.doi.org/#1}{#2}}
\newcommand\httpurl[2]    {\href{http://#1}{#2}}

\end{document}